\documentclass[12pt]{article}
\usepackage[utf8]{inputenc}
\usepackage[sectionbib]{natbib}

\title{\vspace{-2.3cm}Two-Stage Robust Sparse Gradient Methods for Regression Under Heavy-Tailed Designs}
\author{Kaiyuan Zhou$^\dagger$, Xiaoyu Zhang$^\ddagger$\footnote{Kaiyuan Zhou and Xiaoyu Zhang contribute equally and are joint first authors.}, Wenyang Zhang$^\S$, and Di Wang$^\dagger$ \\ \small{$\dagger$ Shanghai Jiao Tong University, $\ddagger$ Tongji University, $\S$ University of Macau}}
\date{}
\usepackage[margin=1in]{geometry}
\usepackage{graphicx}
\usepackage{mathrsfs}
\usepackage{multirow}
\usepackage{amsfonts}
\usepackage{amsmath}
\usepackage{comment}
\usepackage{amsthm}
\usepackage{mathtools}
\usepackage{algorithm}
\usepackage{sectsty}
\usepackage[colorlinks]{hyperref}
\usepackage{amssymb}
\usepackage[toc,page]{appendix}
\usepackage[mathscr]{euscript} 

\let\counterwithin\relax
\usepackage{chngcntr}
\usepackage{apptools}
\usepackage{booktabs}
\usepackage{lscape}
\usepackage{arydshln}
\usepackage{soul}
\usepackage{setspace}
\usepackage{epstopdf}
\usepackage{xr}
\usepackage{array}
\usepackage{makecell}
\usepackage{xcolor}
\usepackage{algorithmic}
\usepackage{tikz}
\usepackage{pgfplots}
\usepackage{threeparttable}

\makeatletter
\newcommand*{\addFileDependency}[1]{
  \typeout{(#1)}
  \@addtofilelist{#1}
  \IfFileExists{#1}{}{\typeout{No file #1.}}
}
\makeatother

\AtAppendix{\counterwithin{lemma}{section}}

\newtheorem{assumption}{Assumption}
\newtheorem{definition}{Definition}
\newtheorem{lemma}{Lemma}
\newtheorem{proposition}{Proposition}
\newtheorem{theorem}{Theorem}

\newtheorem{remark}{Remark}

\theoremstyle{definition}

\hypersetup{
    colorlinks=true,
    linkcolor=blue,
    citecolor=blue,
    filecolor=magenta,
    urlcolor=magenta,
}

\newcommand{\bm}{\mathbf}
\newcommand{\bbm}{\boldsymbol}

\DeclareFontFamily{U}{mathx}{}
\DeclareFontShape{U}{mathx}{m}{n}{<-> mathx10}{}
\DeclareSymbolFont{mathx}{U}{mathx}{m}{n}
\DeclareMathAccent{\widecheck}{0}{mathx}{"71}

\mathtoolsset{showonlyrefs}

\begin{document}

\setlength{\parindent}{16pt}

\maketitle

\vspace{-1.1cm}
\begin{abstract}
We study high-dimensional sparse regression under simultaneous heavy-tailed covariates and noise. Heavy-tailed data affect sparse optimization in two different ways: extreme covariates can destabilize the gradient field during global localization, while heavy-tailed noise limits the final statistical accuracy during local refinement. Motivated by this two-phase structure, we propose two-stage RIGHT, a robust sparse first-order method based on coordinate-wise median-of-means (MoM) gradient estimation and delayed sample splitting. The MoM gradient estimator is computationally simple, compatible with hard-thresholded updates, and admits phase-adaptive concentration bounds whose rates depend on the current localization radius. Delayed splitting reuses data during global localization and reserves fresh batches for the shorter refinement stage, reducing the sample-splitting cost. The theoretical results reveal a decoupled rate structure: the design-tail index controls gradient stability and sample complexity, whereas the noise-tail index controls the final statistical rate. We also provide phase-wise lower-bound benchmarks showing that the design-driven localization barrier is intrinsic. Extensive simulation experiments and real data analysis showcase the efficacy of the proposed method over existing competitors.
\end{abstract}

\vspace{-0.3cm}

\newpage

\setlength\abovedisplayskip{2.5pt}
\setlength\belowdisplayskip{2.5pt}

\section{Introduction}\label{sec:introduction}

\vspace{-8pt}
\subsection{Motivation}

High-dimensional sparse regression aims to recover a structured parameter when the ambient dimension \(p\) may exceed the sample size \(n\). Classical methods, including the Lasso \citep{bickel2009simultaneous, negahban2012unified} and Iterative Hard Thresholding \citep[IHT,][]{jain2014iterative}, are well understood under light-tailed assumptions. In many modern applications, however, both design covariates and responses may be heavy-tailed. Sampling from such distributions can generate extreme observations that are far from typical samples. We focus on robustness to heavy-tailed sampling rather than arbitrary contamination; see \citet{loh2025theoretical} for a broader discussion of robustness in high-dimensional regression.

Heavy-tailed noise and heavy-tailed designs create different statistical difficulties. Robust loss functions, such as Huber-type losses \citep{fan2017estimation, sun2020adaptive}, are effective tools for reducing the influence of extreme residuals. However, such residual-side robustification does not directly control design tails, and a few extreme covariate vectors can still distort empirical first-order information and degrade sparse estimation. Design-side truncation and shrinkage methods \citep{fan2021shrinkage,zhu2021taming} directly mitigate extreme covariates, but many existing analyses impose fourth-moment-type or stronger conditions. This raises a basic question: \textit{how can one perform sparse regression when the design may have infinite fourth moments and the noise may have infinite second moments?}

To understand the joint impact of heavy-tailed designs and noise, we take a gradient, or score-based, perspective. Gradients are the basic objects used by first-order algorithms, while score equations determine the statistical target and the fluctuations of empirical scores determine estimation accuracy. Thus, under heavy-tailed data, the central issue is how to estimate the population gradient reliably along the iterative path. This motivates gradient-level robustification: rather than robustifying residuals or preprocessing covariates, we directly robustify the first-order information used by sparse optimization algorithms.

This gradient perspective reveals a two-phase structure. The phenomenon is most transparent in sparse linear regression. For the linear model
\(y_i=\bm x_i^\top\bbm\theta^*+\epsilon_i\), the per-observation score at $\bbm\theta$ admits the decomposition $\nabla_{\bbm\theta}\mathcal L(\bbm\theta;Z_i)=\bm x_i\bm x_i^\top(\bbm\theta-\bbm\theta^*)-\epsilon_i\bm x_i$. When the iterate is far from \(\bbm\theta^*\), the design-driven term \(\bm x_i\bm x_i^\top(\bbm\theta-\bbm\theta^*)\) controls whether the estimated gradient preserves a reliable descent direction. Its tail behavior is governed by quadratic covariate products and is therefore especially sensitive to heavy-tailed designs. Once the iterate is close to \(\bbm\theta^*\), this term is attenuated, while the noise-driven score \(\epsilon_i\bm x_i\) remains present and determines the stochastic floor. Thus the same iterative procedure faces two distinct tasks: global localization and local refinement. Importantly, the required statistical control of the gradient changes across these tasks: the global phase is governed by design-driven gradient stability, whereas the local phase is governed by noise-driven score fluctuations.

Another challenge is data dependence. Robustified gradients need not arise as gradients of a fixed empirical loss, so standard oracle inequalities for global empirical minimizers do not directly apply. Instead, one must track the realized iterates of the algorithm \citep{liu2019high,liu2020high,prasad2020robust}. When the same data are used both to generate the iterates and to evaluate robust gradients, the resulting dependence complicates the analysis. Fresh sample splitting at every iteration restores conditional independence \citep{balakrishnan2017statistical, prasad2020robust,liu2019high}, but it reduces the effective sample size from \(n\) to roughly \(n/T\), where \(T\asymp \log n\) is the typical number of iterations. These considerations motivate us to develop a robust sparse first-order framework that combines \textit{phase-adaptive gradient estimation} with a \textit{sample-efficient data-access strategy}.

\vspace{-10pt}
\subsection{Proposed Framework and Contributions}
We develop TS-RIGHT, a robust sparse first-order framework for heavy-tailed regressions. The framework combines two ingredients: the coordinate-wise median-of-means (MoM) gradient estimator and the delayed sample-splitting strategy. The former provides robust first-order information suited to sparse hard-thresholded updates, while the latter controls data dependence without splitting the sample across the full optimization path.
 
Our first contribution is the coordinate-wise MoM gradient estimator tailored to sparse iterative updates. Although robust gradient estimation has been adopted in prior work, the coordinate-wise MoM construction is computationally simple, naturally compatible with hard-thresholding, and admits a phase-adaptive concentration theory. Its stochastic error depends on the current localization radius: in the global regime, the rate is controlled by the tail behavior of design-driven gradient terms, whereas in the local regime, the rate is controlled by noise-driven score fluctuations. This phase-adaptive concentration is the key property that allows the same gradient estimator to support both localization and refinement.

Our second contribution is the two-stage robust gradient descent framework and its deterministic convergence analysis. The algorithm uses the MoM gradient estimator within an iterative hard-thresholded update and implements delayed sample splitting: the first stage reuses data to obtain global localization, while the second uses fresh batches only for the shorter local refinement path. This data access strategy improves sample efficiency by splitting the data only where it is most needed. We prove a deterministic two-stage convergence theorem under mild sparsity-restricted curvature and sparse robust gradient-stability conditions. This separates the optimization argument from model-specific concentration analysis.

Our third contribution is the verification of the framework for sparse linear and logistic regression under weak moment assumptions. For sparse linear regression, we allow finite restricted \((2+2\lambda)\)-moments of the design and finite \((1+\delta)\)-moments of the noise, for some $0<\lambda,\delta\leq 1$. The resulting rates have a decoupled structure: the design-tail index \(\lambda\) controls the localization barrier through gradient stability and sample complexity, while the noise-tail index \(\delta\) controls the local refinement floor and hence the final statistical rate. For sparse logistic regression, the bounded score multiplier removes the quadratic design-product difficulty that appears in linear regression, leading to milder moment requirements. The resulting rates are nearly optimal in their tail-dependent exponents, and we provide a phase-wise lower-bound interpretation, including a lower bound for the design-driven localization barrier.

\vspace{-10pt}
\subsection{Related Literature}
Our work intersects high-dimensional statistics, robust optimization, and heavy-tailed probability. We briefly survey related literature to highlight the distinctions of our framework.

\textbf{High-dimensional robust sparse regression.}
A major line for this problem is robust sparse M-estimation, which combines robust losses with sparsity regularization \citep{belloni2011,fan2017estimation,sun2020adaptive}; see also \citet{filzmoser2021robust} and \citet{loh2025theoretical} for broader reviews.
These methods are effective for heavy-tailed errors or contaminated responses, but their guarantees typically still rely on empirical curvature and design concentration, which can be violated under heavy-tailed covariates. A complementary line explicitly shrinks or truncates  covariates before fitting sparse models \citep{fan2021shrinkage,zhu2021taming}. 
In contrast, TS-RIGHT robustifies the first-order direction along an IHT trajectory, rather than robustifying only the scalar loss, preprocessing the design, or solving a global robust ERM problem.

\textbf{Robust Gradient Descent and Iterative Methods.}
Recent methods achieve robustness using robust gradient surrogates \citep{prasad2020robust,holland2019efficient}, but they are not designed for high-dimensional sparse settings. 
Median- or truncation-based gradient descents have been studied for nonconvex problems \citep{chi2019median,li2020non,zhang2024robust}, but they focus on low-rank or phase-retrieval structures rather than sparse vectors. The most related work is \citet{liu2019high} which combines robust gradient estimation with iterative hard thresholding for sparse regression under finite fourth moments of covariates and finite second moments of noise. Their analysis uses equal sample splitting, leading to an effective sample size of order \(n/\log n\) along the iterative path. While sharing the same idea of robust IHT, our algorithm introduces a delayed sample-splitting strategy and requires weaker moment conditions ($(2+2\lambda)$-th for covariates and $(1+\delta)$-th for noise, for some $0<\lambda,\delta\leq 1$). Together with our two-regime analysis, these innovations improve computational performance and achieve nearly optimal rates.

\textbf{Median-of-Means and Heavy-Tailed Estimation.}
The median-of-means (MoM) principle \citep{nemirovsky1983problem,jerrum1986random,alon1999space} is a standard tool for heavy-tailed estimation \citep{lugosi2019mean}.  In sparse regression, the closest MoM competitors are robust ERM procedures like regularized MoM tournaments and minmax-MoM estimators \citep{lugosi2019regularization,lecue2020robust}. 
These methods achieve strong guarantees, often under moment or small-ball conditions that differ from the weak-moment gradient-stability assumptions studied here. In contrast, our MoM gradient estimator is applied directly to first-order information and avoids solving a global robust ERM or tournament problem at each iteration. While \citet{minsker2015geometric} and \citet{hsu2016loss} also explore MoM for heavy-tailed regression, they still require bounded fixed designs or finite fourth moments on the covariates.

\textbf{Two-phase Phenomena in Regression Optimization.}
Two-phase behavior has appeared in several regression and statistical optimization problems \citep{fan2018lamm, pan2021iteratively}.
\citet{shen2025computationally} show that the nonsmooth loss is smoothed by heavy-tailed noise when the parameter is close to the truth, thus leading to a two-phase step-size rule.
Similar behavior also appears in online quantile regression \citep{shen2025online}.
Our underlying mechanism for the two-phase behavior differs from the existing literature, which primarily focuses on the loss behavior. In contrast, we highlight the phenomenon that the stochastic behavior of the gradient may change along the optimization path under heavy-tailed distributions. 

In summary, Table~\ref{tab:rw_robustification_strategies} compares TS-RIGHT with existing robustification strategies for heavy-tailed sparse regression, highlighting its relaxed moment conditions, broad applicability, and improved sample efficiency.
\begin{table}[!htp]
\centering
\renewcommand{\arraystretch}{1.22}
\setlength{\tabcolsep}{2pt}
\footnotesize
\begin{threeparttable}
\caption{Comparison of robustification strategies for heavy-tailed sparse models.\tnote{a}}
\label{tab:rw_robustification_strategies}
\begin{tabular}{@{}>{\raggedright\arraybackslash}p{0.15\linewidth}
                >{\raggedright\arraybackslash}p{0.16\linewidth}
                >{\centering\arraybackslash}p{0.14\linewidth}
                >{\centering\arraybackslash}p{0.22\linewidth}
                >{\centering\arraybackslash}p{0.16\linewidth}
                >{\centering\arraybackslash}p{0.12\linewidth}@{}}
\toprule
\textbf{Route} & \textbf{Methods} & \textbf{Model Applicability} & \textbf{Typical Design Moments} & \textbf{Typical Noise Moments\tnote{b}} & \textbf{Sample Efficiency} \\
\midrule
Loss robustification & Adaptive Huber regression & Linear & Fixed design or Sub-Gaussian & $(1+\delta)$-th & $n$ \\
\addlinespace[3pt]
Data preprocessing & Shrinkage / truncation & Linear; GLMs & Sub-Gaussian or $4$th & $2k$-th $(k>1)$ or $4$th & $n$ \\
\addlinespace[3pt]
\vspace{-2pt}\multirow{2}{=}{\raggedright Gradient robustification} & Robust GD & Linear; GLMs & $4$th & 2nd & $n/\log n$ \\
\addlinespace[2pt]
 & \textbf{TS-RIGHT (Ours)} & Linear; GLMs & $(2+2\lambda)$-th (linear); 2nd (logistic) & $(1+\delta)$-th & \textbf{$n/\log s$} \\
\bottomrule
\end{tabular}
\begin{tablenotes}
\scriptsize
\item[a] Entries summarize representative assumptions rather than exhaustive theorem statements. Specific comparison of rates and conditions is provided in Tables \ref{tab:comparison} and \ref{tab:logistic_comparison}.
\item[b] Noise moments are available for linear regression model.
\end{tablenotes}
\end{threeparttable}
\end{table}

\vspace{-10pt}
\subsection{Notation and Organization}
We adopt the following notational conventions throughout this article. For an integer $m$, write $[m]=\{1,\ldots,m\}$.  We denote scalars by lowercase letters (e.g., $a, b$), vectors by bold lowercase letters (e.g., $\bm{u}, \bm{v}$), and matrices by bold uppercase letters (e.g., $\bm{A}, \bbm{\Sigma}$). For a vector $\bm{v} \in \mathbb{R}^p$, $\bm{v}_i$ denotes its $i$-th entry. For a matrix $\bm{A} \in \mathbb{R}^{n \times p}$, $\bm{A}_{ij}$ denotes the entry in the $i$-th row and $j$-th column, and $\bm{A}_{i\cdot}$ denotes the $i$-th row as a column vector.
For a vector $\bm{v} \in \mathbb{R}^p$ and a subset of indices $S \subseteq [p]$, we denote by $\bm{v}_S$ the vector that retains the elements in $S$ and sets the others to zero. The support of a vector is defined as $\operatorname{supp}(\bm{v}) \coloneqq \{j \in [p] : \bm{v}_j \neq 0\}$. The $\ell_q$-norm of a vector is defined as $\|\bm{v}\|_q \coloneqq (\sum_{j=1}^p |\bm{v}_j|^q)^{1/q}$ for $q \ge 1$, with $\|\bm{v}\|_\infty \coloneqq \max_{j} |\bm{v}_j|$. The $\ell_0$-pseudo-norm, denoted $\|\bm{v}\|_0$, counts the number of non-zero entries in $\bm{v}$. For matrices $\bm{A} \in \mathbb{R}^{p \times m}$, $\|\bm{A}\|_2$ denotes the spectral norm. For two sequences of non-negative numbers $a_n$ and $b_n$, we write $a_n \lesssim b_n$ if there exists a universal constant $C > 0$, independent of $n$, $p$, and problem-specific parameters, such that $a_n \le C b_n$. We write $a_n \asymp b_n$ if $a_n \lesssim b_n$ and $b_n \lesssim a_n$. We use $C$, $c$ to denote generic positive constants that may change from line to line. 

The remainder of this article is organized as follows. Section \ref{sec:methodology} introduces the TS-RIGHT framework. Section \ref{sec:theory} presents the deterministic convergence theory based on sparse gradient curvature and robust gradient stability. Section \ref{sec:4} applies this framework to linear and logistic regression and establishes convergence rates under weak moment conditions. Section \ref{sec:lower_bound} provides phase-wise lower-bound benchmarks. Section \ref{sec:numerical_studies} presents numerical experiments and Section \ref{sec:conclusion} offers conclusions. All proofs are relegated to Supplementary Materials.

\vspace{-10pt}
\section{Methodology}
\label{sec:methodology}

\vspace{-10pt}
\subsection{Problem Setup and Score-Equation Target}
\label{subsec:score_setup}

Let $\mathcal D_n=\{Z_i=(\bm x_i,y_i)\}_{i=1}^n$ be independent and identically distributed observations from an unknown distribution on $\mathbb R^p\times\mathcal Y$, where $p$ may be much larger than $n$. For a differentiable loss $\mathcal L(\bbm\theta;Z)$, the usual population risk is $\mathcal R(\bbm\theta)\coloneqq \mathbb{E}[\mathcal L(\bbm\theta;Z)]$, whenever this expectation exists. Under heavy-tailed responses or covariates, however, the risk may be infinite even when the first-order moment of the score is well defined. We therefore write $G(\bbm\theta)\coloneqq \mathbb{E}[\nabla_{\bbm\theta}\mathcal L(\bbm\theta;Z)]$ for the population score and define the target $\bbm\theta^*$ as a sparse solution satisfying
\begin{equation}
    G(\bbm\theta^*)=0, \quad\text{and} \quad \|\bbm\theta^*\|_0\le s^* .
\label{eq:score_target}
\end{equation}
Throughout this article, $\bbm\theta^*$ is assumed to be the unique sparse solution of \eqref{eq:score_target} in the parameter class under consideration. When $\mathcal R(\bbm\theta)$ is finite and differentiable, this formulation agrees with the usual first-order condition for population risk minimization.

The score-equation framework encompasses many important high-dimensional models, including linear regression, logistic regression, generalized linear models, and other $M$-estimation problems, whenever the population score $G(\bbm\theta)$ is well defined. 

\vspace{-10pt}
\subsection{Two Gradient Regimes in Heavy-Tailed Linear Regression}
\label{subsec:two_gradient_regimes}
Although the algorithmic framework is formulated for general sparse score-equation models, we use sparse linear regression as the running example because its score decomposition makes the two-phase gradient behavior especially transparent. Specifically, we consider sparse linear regression with the least-squares loss:
\begin{equation}
    y_i=\bm x_i^\top\bbm\theta^*+\epsilon_i, \qquad
    \mathcal L(\bbm\theta;Z_i)=\frac{1}{2}(y_i-\bm x_i^\top\bbm\theta)^2,
    \label{eq:linear_regression_loss}
\end{equation}
Let $\Delta=\bbm\theta-\bbm\theta^*$. The per-observation score admits the decomposition
\begin{equation}
\nabla_{\bbm\theta}\mathcal L(\bbm\theta;Z_i)=\underbrace{\bm x_i\bm x_i^\top\Delta}_{\text{design-driven term}}-\underbrace{\epsilon_i\bm x_i}_{\text{noise-driven score term}} .
\label{eq:gradient_decomposition}
\end{equation}
This identity shows that the stochastic behavior of the score depends on the current distance to the target. When $\|\Delta\|_2$ is large, the score contains a substantial design-driven component involving the quadratic product $\bm x_i\bm x_i^\top$. When $\|\Delta\|_2$ is close to zero, this design-driven component is attenuated, and the remaining stochastic fluctuation is mainly governed by the local score noise $\epsilon_i\bm x_i$.

\begin{figure}[!htp]
\centering
\includegraphics[width=0.85\textwidth]{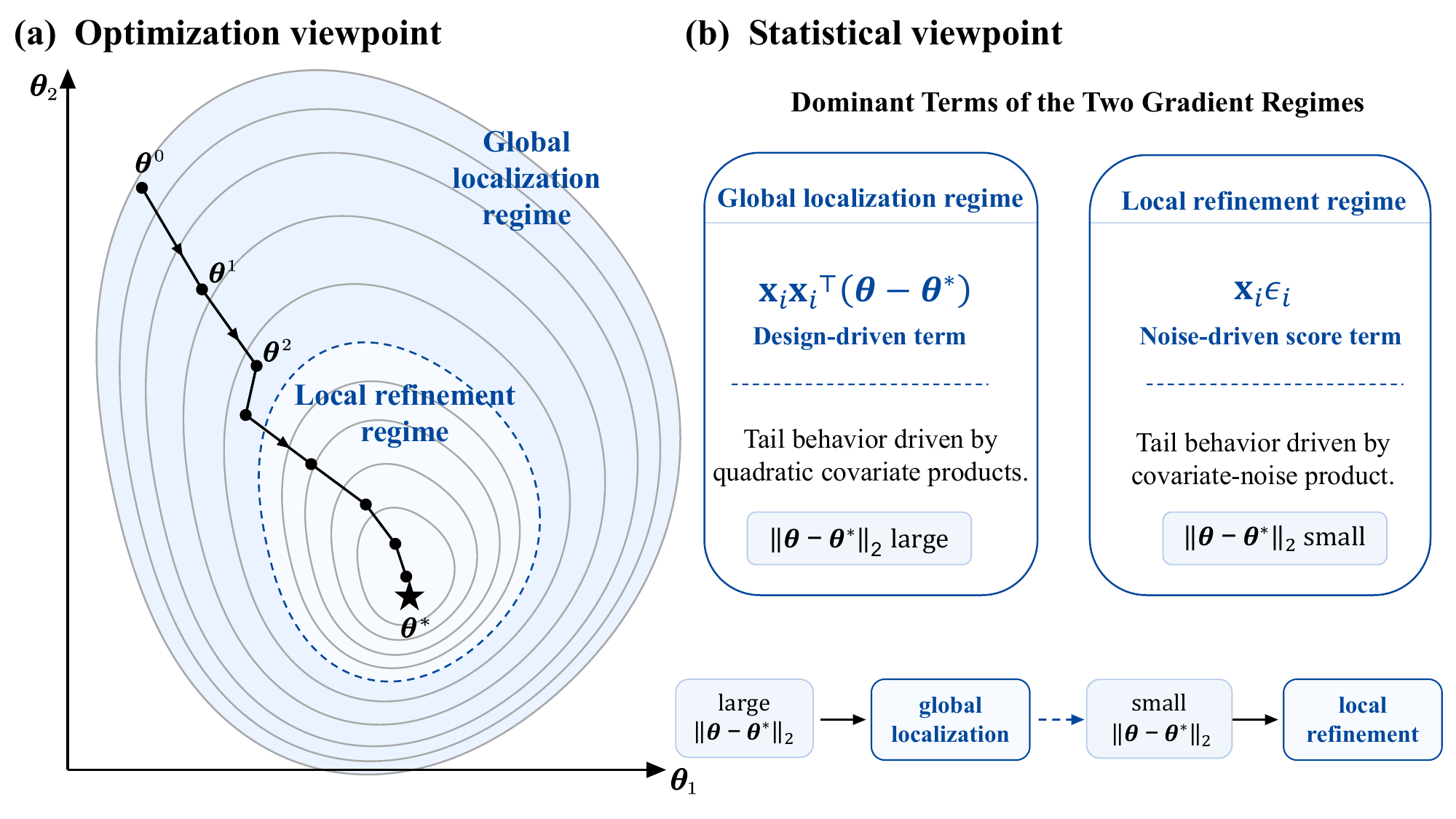}\vspace{-0.5cm}
\caption{Two gradient regimes in heavy-tailed sparse linear regression.}
\label{fig:two_regime}
\end{figure}

From the viewpoint of first-order optimization, the negative population score $-G(\bbm\theta)$ is the ideal direction toward the sparse score root; see Figure~\ref{fig:two_regime}. In practice, $G(\bbm\theta)$ is unknown and must be estimated from the data, so we write $g(\bbm\theta)$ informally for the estimated gradient in this discussion. Early in the optimization path, $g(\bbm\theta)$ is used mainly for \emph{global localization}: the goal is to move the iterate into a neighborhood of $\bbm\theta^*$, rather than to attain the final statistical accuracy. Once the iterate is inside such a neighborhood, the role of $g(\bbm\theta)$ shifts to \emph{local refinement}, where the remaining task is to reduce the local statistical error.

These two regimes also exhibit distinct statistical properties, as the difficulty of estimating $G(\bbm\theta)$ depends on the distribution of the per-observation score in \eqref{eq:gradient_decomposition}. In the global localization regime, the design-driven term $\bm x_i\bm x_i^\top\Delta$ is prominent, making gradient estimation governed by quadratic covariate products. Under a finite $(2+2\lambda)$-moment condition with $0<\lambda\le1$ on $\bm x_i$, the tail behavior of the score is controlled by the design-tail index $\lambda$. In the local refinement regime, $\Delta$ is small, and the noise-driven score $\epsilon_i\bm x_i$ emerges as the leading stochastic component. Its tail behavior is determined by the covariate-noise product, for example through a finite $(1+\delta)$-moment condition with $0<\delta\le1$ on $\epsilon_i\bm x_i$ along sparse directions. Consequently, the tail behavior needed for reliable global localization can differ significantly from that determining the final estimation accuracy.

This two-regime view motivates the methodological development below. We need a gradient estimation method that remains reliable when the score is affected by heavy-tailed design products, but that also attains the sharper local accuracy dictated by the noise-driven score near the target. In other words, the estimated gradient should work in both the global localization and local refinement regimes and should adapt to the changing distributional structure along the optimization path.

\vspace{-10pt}
\subsection{Robust Sparse Gradient Estimation by Median-of-Means}
\label{subsec:robust_sparse_gradient_mom}

For each fixed parameter value $\bbm\theta$, estimating the ideal direction $G(\bbm\theta)$ is a mean estimation problem for the random score vector $\nabla_{\bbm\theta}\mathcal L(\bbm\theta;Z)$. This viewpoint applies to any differentiable score-equation model. In heavy-tailed settings, the score vector may have unstable coordinates; in the linear regression example, this instability arises through the design-driven product $\bm x_i\bm x_i^\top(\bbm\theta-\bbm\theta^*)$ and the noise-driven score $\epsilon_i\bm x_i$. We therefore replace the empirical mean gradient by a robust mean estimator of the population score.

Classical iterative hard thresholding uses the empirical gradient, the ordinary sample mean of the per-observation scores. Under heavy-tailed covariates, this empirical mean can be unstable because a small number of observations may dominate the average. Since our target is defined by the population score equation \eqref{eq:score_target}, we robustify this first-order moment directly, rather than only modifying the loss. Robust losses can reduce the influence of large residuals, but they do not directly address heavy-tailed score vectors induced by covariate products. See the discussions in Appendix~\ref{append:A} of Supplementary Materials for more details.

For a data subset $\mathcal A\subseteq\mathcal D_n$, we write $g(\bbm\theta;\mathcal A)$ for a robust estimate of $G(\bbm\theta)$ based on the sample scores $\{\nabla_{\bbm\theta}\mathcal L(\bbm\theta;Z_i):Z_i\in\mathcal A\}$. Given such an estimator, the basic sparse update is
\begin{equation}
\bbm\theta^+=\mathcal P_s\{\bbm\theta-\eta g(\bbm\theta;\mathcal A)\},
\label{eq:generic_right_update}
\end{equation}
where $\eta>0$ is a step size and $\mathcal P_s$ is the hard-thresholding operator that keeps the $s$ largest coordinates in absolute value and sets the remaining coordinates to zero. The working sparsity level $s$ is taken to be a constant multiple of $s^*$ in the theory; in applications, it can be selected by validation or stability selection.

The hard-thresholded update is inherently sparse: only a small number of coordinates can enter the current iterate, the next iterate, and the target. Therefore, the gradient estimator does not need to be optimal for dense Euclidean mean estimation in $\mathbb R^p$; what matters is reliable gradient information on sparse coordinate sets. This observation makes coordinate-wise robustification a natural and computationally simple choice.

We use a coordinate-wise Median-of-Means (MoM) estimator as a concrete robust gradient oracle. Given the data subset $\mathcal A$ and an integer $K$, partition $\mathcal A$ into disjoint blocks $\mathcal B_1(\mathcal A),\ldots,\mathcal B_K(\mathcal A)$ of equal or nearly equal sizes. For $j=1,\ldots,p$, define
\begin{equation}
g_j(\bbm\theta;K,\mathcal A)=\operatorname{median}_{1\le k\le K}\left[\frac{1}{|\mathcal B_k(\mathcal A)|}\sum_{Z_i\in\mathcal B_k(\mathcal A)}\nabla_j\mathcal L(\bbm\theta;Z_i)\right].
\label{eq:median_of_means_estimator}
\end{equation}
The median aggregation prevents a small number of unstable block averages from dominating the gradient estimate. The number of blocks $K$ controls the usual robustness-efficiency trade-off: larger $K$ gives stronger protection against heavy-tailed block means, while smaller blocks increase variability.

\begin{remark}[Coordinate-wise MoM]
The coordinate-wise construction is well matched to sparse hard thresholding. It is not intended to be an optimal dense Euclidean mean estimator in $\mathbb R^p$. Its role is more specific: the hard-thresholding recursion requires reliable gradient information on sparse coordinate sets, and coordinate-wise robustification provides a direct way to control such errors. Computationally, the estimator is also simple, requiring only block averages and $p$ univariate medians at each gradient evaluation. The formal sparse-set stability requirement is introduced in Section~\ref{subsec:srs_forms}.
\end{remark}

The robust gradient estimator in \eqref{eq:median_of_means_estimator}, together with the sparse update in \eqref{eq:generic_right_update}, forms the basic computational step of the proposed algorithm. The next subsection describes how this step is organized along a data-dependent optimization path.

\vspace{-10pt}
\subsection{Path Dependence and Delayed Sample Splitting}
\label{subsec:path_dependence_delayed_splitting}

Since our method uses a robust gradient vector field, which is not necessarily the gradient of a fixed empirical loss, one cannot analyze it by showing that a global optimizer satisfies a specific oracle inequality. The analysis must control the realized iterates. This iterative method introduces an additional difficulty: the gradient estimation error must be controlled along the optimization path, but the parameter at which the gradient is evaluated is itself produced from the data. If the same observations are reused across iterations, then the iterate $\bbm\theta^t$ depends on the data used to compute $g(\bbm\theta^t;\mathcal D_n)$. Consequently, a concentration bound for $g(\bbm\theta;\mathcal D_n)$ at a data-independent $\bbm\theta$ does not automatically justify the same bound at the data-dependent iterate $\bbm\theta^t$. This is a well-known path-dependence problem in statistical optimization \citep{prasad2020robust}; see the detailed discussion of path dependence in Appendix~\ref{append:B} of Supplementary Materials.

One natural strategy is \emph{no splitting}, where every gradient evaluation uses the full sample. This is attractive computationally and could be statistically efficient in practice, because each gradient estimate uses all available observations. Its analysis, however, requires uniform control of the robust gradient process over the relevant sparse region, since iterates are data dependent. Under heavy-tailed covariates, control of the unstable design-product terms $\bm x_i\bm x_i^\top(\bbm\theta-\bbm\theta^*)$ across all $\theta$ is demanding, because it pays an additional spatial complexity cost over a much larger class than the realized path itself. 

The opposite extreme is \emph{equal splitting}, where each iteration uses a fresh batch of observations \citep{balakrishnan2017computationally,prasad2020robust}. Conditional on the past, the current iterate is then fixed relative to the data used for the next gradient evaluation, so fixed-point concentration can be applied directly. The drawback is that the total sample is divided across the entire optimization path. If the method is run for $T\asymp\log n$ iterations, each gradient estimate uses only a fraction of the data, and the final statistical rate pays a corresponding splitting cost.

The two-regime structure in Section~\ref{subsec:two_gradient_regimes} suggests a more targeted strategy. In the \emph{global localization regime}, the goal is only to move the iterate into a neighborhood of $\bbm\theta^*$, rather than to attain the final statistical radius. It is therefore natural to reuse one part of the sample with a conservative robust gradient estimate that is stable enough for localization. Once the iterate has entered a local neighborhood, the remaining task is \emph{local refinement}. For this second task, we reserve an independent part of the sample and split it only across the shorter local refinement path. We call this strategy \emph{delayed sample splitting}; see Figure~\ref{fig:splitting_strategies}.

\begin{figure}[!htp]
    \includegraphics[width=0.93\textwidth,height=8.1cm]{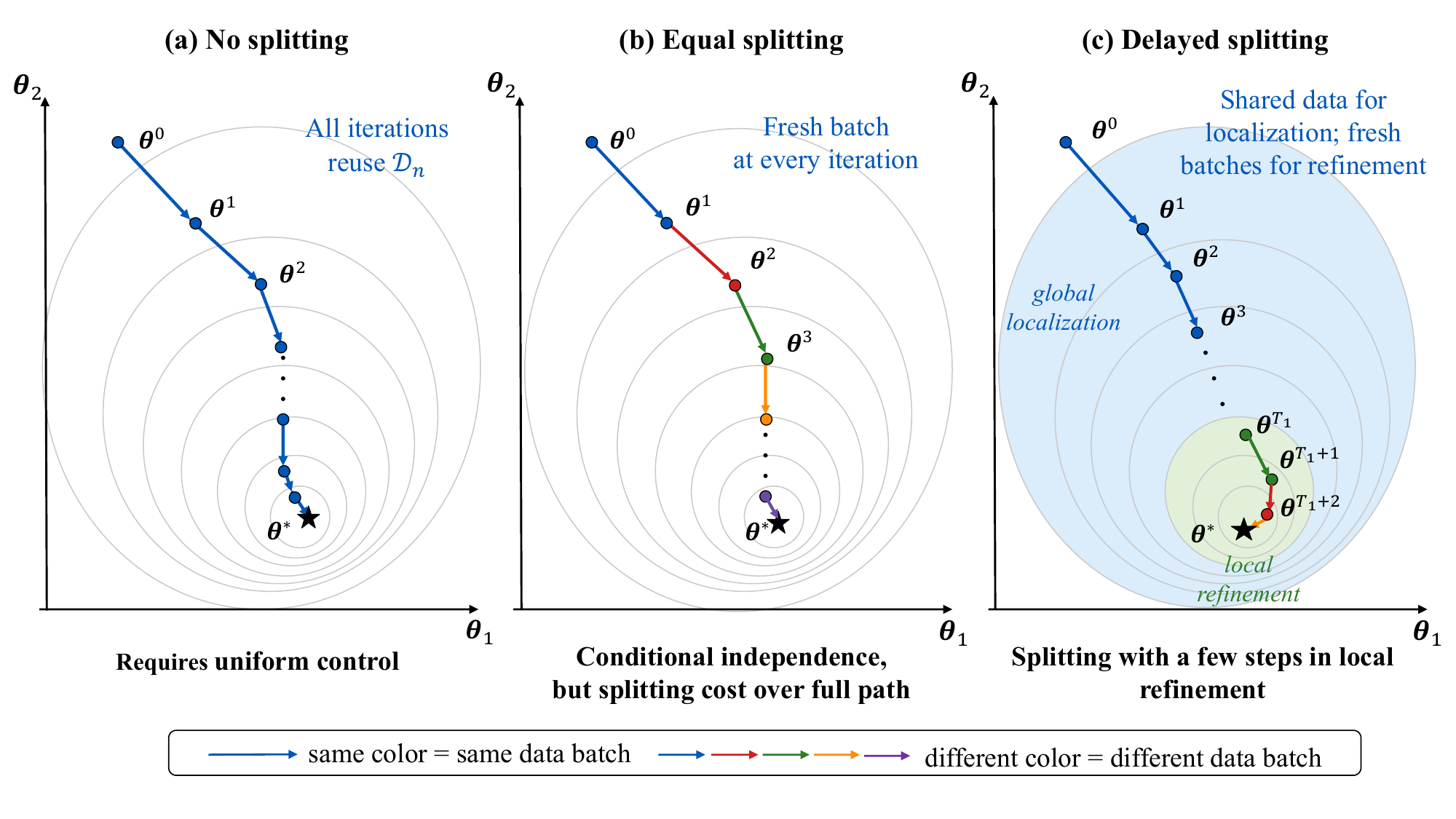}
    \vspace{-0.5cm}
    \caption{Illustration of sample splitting strategies.}
    \label{fig:splitting_strategies}
\end{figure}

The advantage of delayed sample splitting is that independence is used only where it is most valuable. The second phase does not need to repeat the entire optimization path; it starts from a localized estimator and only needs to close the gap between the coarse localization radius and the final local score radius. As a result, the splitting cost is paid only over the shorter refinement stage, which requires on the order of $\log s$ steps rather than $\log n$ fully split iterations.

Delayed sample splitting is therefore not merely a technical device for obtaining conditional independence. It aligns the data-access strategy with the statistical role of the two regimes: shared data are used for coarse localization under heavy-tailed design products, while fresh data are reserved for local refinement and final estimation accuracy. We now combine this idea with the MoM gradient update to obtain the proposed algorithm.

\vspace{-10pt}
\subsection{TS-RIGHT Algorithm}
\label{subsec:two_stage_right}

We now combine the robust sparse-gradient update in \eqref{eq:generic_right_update}, the coordinate-wise MoM gradient estimator in \eqref{eq:median_of_means_estimator}, and the delayed sample-splitting strategy described in Section~\ref{subsec:path_dependence_delayed_splitting}. The resulting procedure, called Two-Stage Robust Iterative Gradient Descent with Hard Thresholding (TS-RIGHT), has two phases. Phase I performs global localization using one part of the data, and Phase II performs local refinement using fresh independent batches from the remaining data.

Split the sample into two disjoint parts $\mathcal D^{(1)}$ and $\mathcal D^{(2)}$, with sample sizes $n_1$ and $n_2$, where $n_1\asymp n_2\asymp n$. The first phase reuses $\mathcal D^{(1)}$ for $T_1$ robust gradient steps. Its purpose is to move the estimator into a neighborhood of $\bbm\theta^*$, rather than to achieve the final statistical accuracy. The second phase splits $\mathcal D^{(2)}$ into $T_2$ disjoint batches and uses one fresh batch at each local refinement step. This second phase starts from the localized estimator produced by Phase I and refines it using independent local score estimates. The algorithm is written for a generic differentiable loss through the robust score estimator $g(\bbm\theta;\mathcal A)$; the model-specific analysis enters only through the verification of gradient stability and population-score regularity.

\begin{algorithm}[!htp]
\caption{TS-RIGHT}
\label{alg:two_stage_right}
\begin{algorithmic}[1]
\REQUIRE Data $\mathcal D_n$, initialization $\bbm\theta^0$, working sparsity $s$, step size $\eta$, iteration numbers $T_1,T_2$, and MoM block numbers $K_1,K_2$.
\STATE Split $\mathcal D_n$ into two disjoint parts $\mathcal D^{(1)}$ and $\mathcal D^{(2)}$, with $|\mathcal D^{(1)}|=n_1$, $|\mathcal D^{(2)}|=n_2$, and $n_1\asymp n_2\asymp n$.
\STATE \textbf{Phase I: global localization.}
\FOR{$t=0,\ldots,T_1-1$}
\STATE Compute the MoM gradient $g(\bbm\theta^t;K_1,\mathcal D^{(1)})$ by \eqref{eq:median_of_means_estimator}.
\STATE Set $\bbm\theta^{t+1}\leftarrow\mathcal P_s\{\bbm\theta^t-\eta g(\bbm\theta^t;K_1,\mathcal D^{(1)})\}$.
\ENDFOR
\STATE Split $\mathcal D^{(2)}$ into disjoint batches $\mathcal D^{(2)}_1,\ldots,\mathcal D^{(2)}_{T_2}$.
\STATE \textbf{Phase II: local refinement.}
\FOR{$\ell=1,\ldots,T_2$}
\STATE Compute the MoM gradient $g(\bbm\theta^{T_1+\ell-1};K_2,\mathcal D^{(2)}_\ell)$ by \eqref{eq:median_of_means_estimator}.
\STATE Set $\bbm\theta^{T_1+\ell}\leftarrow\mathcal P_s\{\bbm\theta^{T_1+\ell-1}-\eta g(\bbm\theta^{T_1+\ell-1};K_2,\mathcal D^{(2)}_\ell)\}$.
\ENDFOR
\ENSURE $\widehat{\bbm\theta}=\bbm\theta^{T_1+T_2}$.
\end{algorithmic}
\end{algorithm}

The default theoretical choices reflect the roles of the two phases. In Phase I, we take $T_1\asymp\log n$ iterations and use a larger number of MoM blocks, of order $s\log p$, to support sparsity-uniform localization. In Phase II, we take $T_2\asymp\log s$ iterations and use $K_2\asymp\log p$ blocks in each fresh batch, which is sufficient for high-dimensional local score estimation along the refinement path. 

TS-RIGHT is a general robust sparse first-order framework beyond linear regression. While the specific two-regime decomposition in Section~\ref{subsec:two_gradient_regimes} is model-dependent, the convergence analysis only requires two structural ingredients: sparse gradient stability along the realized path and a population-score regularity condition. Section~\ref{sec:theory} states this abstract theory, and Section~\ref{sec:4} verifies it for sparse regression models.

\vspace{-10pt}
\section{Deterministic Convergence Analysis}
\label{sec:theory}

This section gives an abstract convergence analysis for TS-RIGHT. The purpose is to separate the deterministic optimization argument from the probabilistic verification of robust gradient bounds. The analysis has two ingredients. The first is sparse gradient stability, which controls the error of the robust gradient estimator on coordinate sets relevant to hard thresholding. The second is a population-score regularity condition, which ensures that the population score points toward the sparse target strongly enough to generate contraction. 

\vspace{-10pt}
\subsection{Fixed, Uniform, and Pathwise Sparse Gradient Stability}
\label{subsec:srs_forms}

The convergence of TS-RIGHT depends on how accurately the robust gradient estimator approximates the population score on sparse sets. Since the iterates are data dependent, we distinguish three forms of sparse gradient stability. Recall that $G(\bbm\theta)=\mathbb{E}[\nabla_{\bbm\theta}\mathcal L(\bbm\theta;Z)]$
denotes the population score. Throughout this article, we write \(\bar s=2s+s^*\). This is the relevant sparse support size for hard-thresholded updates, as the active set always depends on the current iterate, the next iterate, and the target parameter.

\begin{definition}[Sparse gradient stability, SRS]
\label{def:SRS}
Let \(\mathcal A\) be a data subset and let \(\Theta_s\) be an \(s\)-sparse
parameter class. For \(\bbm\theta\in\Theta_s\), write $\mathsf{SRS}(\bbm\theta;K,\mathcal A,\bar s,\phi,\gamma)$ for the event that
\begin{equation}
\|[g(\bbm\theta;K,\mathcal A)-G(\bbm\theta)]_S\|_2
\le
\phi\|\bbm\theta-\bbm\theta^*\|_2+\gamma,
\quad
\text{for every } S\subset[p]\text{ with } |S|\le \bar s .
\label{eq:SRS}
\end{equation}
We use the following three forms.

\emph{(1) Fixed-point SRS:}
For a deterministic \(\bbm\theta\in\Theta_s\), we say that
\(g(\cdot;K,\mathcal A)\) satisfies fixed-point
\((\bar s,\phi,\gamma)\)-SRS at \(\bbm\theta\) if
\(\mathsf{SRS}(\bbm\theta;K,\mathcal A,\bar s,\phi,\gamma)\) holds.

\emph{(2) Uniform SRS:}
We say that \(g(\cdot;K,\mathcal A)\) satisfies uniform
\((\Theta_s,\bar s,\phi,\gamma)\)-SRS if simultaneously for all \(\bbm\theta\in\Theta_s\), $\mathsf{SRS}(\bbm\theta;K,\mathcal A,\bar s,\phi,\gamma)$ holds.

\emph{(3) Pathwise SRS:}
Let \(\{\bbm\theta^t\}_{t\in\mathcal I}\subset\Theta_s\) be a sequence of
iterates. Let \(\{\mathcal A_j\}_{j\in\mathcal J}\) be a collection of
data subsets, and let
\(\pi:\mathcal I\to\mathcal J\) be a schedule map specifying which data subset
is used at each iterate. We keep the block number \(K\) fixed across iterations for simplicity. The scheduled gradient evaluations satisfy
pathwise \((\mathcal I,\pi,\bar s,\phi,\gamma)\)-SRS if
$
\mathsf{SRS}(\bbm\theta^t;K,\mathcal A_{\pi(t)},\bar s,\phi,\gamma)
$
holds for every \(t\in\mathcal I\).

\end{definition}

These three SRS forms facilitate the gradient deviation analysis of the two-stage procedure in Algorithm~\ref{alg:two_stage_right}, which requires SRS to hold along the realized optimization path, i.e., the pathwise form. The two stages achieve this pathwise control through different mechanisms.

In Phase I, the same data subset \(\mathcal D^{(1)}\), with sample size \(n_1\), is reused throughout the global localization iterations. The Phase I iterates are therefore data-dependent functions of \(\mathcal D^{(1)}\), so a fixed-point concentration bound for a deterministic sparse parameter cannot be applied directly. Instead, we establish a uniform SRS event over \(\Theta_s\). Once this event holds, it automatically controls every sparse iterate produced in Phase I, and hence implies pathwise SRS for the first stage.

In Phase II, the sample-splitting structure provides the required independence. At local refinement step \(\ell\), the current iterate is measurable with respect to the information generated before the fresh batch \(\mathcal D^{(2)}_\ell\) is used, while \(\mathcal D^{(2)}_\ell\) is independent of this past information. Conditional on the past, the current iterate is therefore fixed, and the fixed-point SRS bound can be applied to the gradient computed from \(\mathcal D^{(2)}_\ell\) and therefore implies pathwise SRS for the second stage.
Lemma~\ref{lem:srs_to_pathwise} in Appendix~\ref{append:C} formalizes these two routes to pathwise SRS.

\vspace{-10pt}
\subsection{Population-Score Regularity and Two-Stage Convergence}
\label{subsec:srcg_two_stage_convergence}

Sparse gradient stability controls the stochastic error of robust gradient estimates. To translate this into iterate contraction, we need a regularity condition on the population score map. Since the population risk may be infinite under heavy-tailed responses, we state this condition directly in terms of \(G\) rather than the risk surface.

\begin{definition}[Sparsity-Restricted Correlated Gradient, SRCG]
\label{def:SRCG}
For \(s\ge s^*\) and constants \(a>0\) and \(0<b\le 1/(4a)\), we say that the population score map \(G\) satisfies the \((s,a,b)\)-SRCG condition if, for every \(\bbm\theta\in\mathbb R^p\) with \(\|\bbm\theta\|_0\le s\) and every \(S\subset[p]\) satisfying \(\operatorname{supp}(\bbm\theta)\cup\operatorname{supp}(\bbm\theta^*)\subseteq S\) and \(|S|\le \bar s\),
\begin{equation}
\langle G(\bbm\theta),\bbm\theta-\bbm\theta^*\rangle\ge a\|G(\bbm\theta)_S\|_2^2+b\|\bbm\theta-\bbm\theta^*\|_2^2 .
\label{eq:SRCG}
\end{equation}
\end{definition}

SRCG is a first-order restricted regularity condition. It ensures that, on sparse directions, the population score is sufficiently aligned with the estimation error and has enough strength to drive contraction. Unlike restricted strong convexity of the population risk, SRCG only requires the population score to be well defined. When the population risk exists, SRCG follows from standard restricted strong convexity and restricted strong smoothness conditions on the risk; the elementary verification is included in Appendix~\ref{append:SRCG} for completeness.

The following lemma is the deterministic contraction step used in both phases of the algorithm. It is stated for an arbitrary sequence of gradient estimates satisfying pathwise SRS, and therefore applies equally to uniform-control and fresh-batch settings once the corresponding pathwise stability event has been established.

\begin{lemma}[One-stage contraction under pathwise SRS]
\label{lem:one_stage_contraction_pathwise_srs}
Suppose that the population score map \(G\) satisfies the \((s,a,b)\)-SRCG condition. Consider the hard-thresholded update \(\bbm\theta^{t+1}=\mathcal P_s\{\bbm\theta^t-\eta g_t(\bbm\theta^t)\}\) for \(t=t_0,\ldots,t_0+T-1\), initialized at an \(s\)-sparse vector \(\bbm\theta^{t_0}\in \Theta_s\). Let the step size be \(\eta=2a\eta_0\) with \(\eta_0\in(0,1]\). Suppose that the gradient sequence satisfies pathwise \((\{t_0,\ldots,t_0+T-1\},\pi,\bar s,\phi,\gamma)\)-SRS. Then there exist constants \(\bar\phi>0\), \(\rho\in(0,1)\), \(C_\gamma>0\), and \(c_s>1\), depending only on \(a\), \(b\), and \(\eta_0\), such that if \(\phi\le\bar\phi\) and \(s\ge c_s s^*\), then
\begin{equation}
\|\bbm\theta^{t_0+T}-\bbm\theta^*\|_2\le \rho^T\|\bbm\theta^{t_0}-\bbm\theta^*\|_2+C_\gamma\gamma .
\label{eq:one_stage_contraction}
\end{equation}
\end{lemma}

The working sparsity level $s$ in Lemma~\ref{lem:one_stage_contraction_pathwise_srs} can be taken as a constant multiple of $s^*$. Throughout the article, we take $s\asymp s^*$, so that $\bar s\asymp s^*$. This lemma separates the deterministic optimization argument from the probabilistic verification of robust gradient stability. The condition $\phi\le\bar \phi$ limits the multiplicative error \(\phi\|\bbm\theta^t-\bbm\theta^*\|_2\) to preserve contraction, while the additive term \(\gamma\) determines the statistical error floor.

We now apply the one-stage contraction lemma to the TS-RIGHT algorithm.

\begin{theorem}[Two-stage deterministic convergence]
\label{thm:two_stage_right_convergence}
Suppose that the population score map \(G\) satisfies the \((s,a,b)\)-SRCG condition. Run Algorithm~\ref{alg:two_stage_right} with step size \(\eta=2a\eta_0\), where \(\eta_0\in(0,1]\), and suppose that \(s\ge c_s s^*\), where \(c_s\) is the constant in Lemma~\ref{lem:one_stage_contraction_pathwise_srs}. Assume that the following two stability events hold: Phase I satisfies uniform \((\Theta_s,\bar s,\phi_1,\gamma_1)\)-SRS on \(\mathcal D^{(1)}\), with all Phase I iterates belonging to \(\Theta_s\). Assume also that, with
\(\mathcal I_2=\{T_1,\ldots,T_1+T_2-1\}\), \(\pi_2(T_1+\ell-1)=\ell\) and \(\mathcal A_\ell=\mathcal D^{(2)}_\ell\), Phase II satisfies pathwise
\((\mathcal I_2,\pi_2,\bar s,\phi_2,\gamma_2)\)-SRS.
If \(\phi_1\le\bar\phi\) and \(\phi_2\le\bar\phi\), where \(\bar\phi\) is the constant in Lemma~\ref{lem:one_stage_contraction_pathwise_srs}, then the Phase I estimator statisfies
\begin{equation}
\|\bbm\theta^{T_1}-\bbm\theta^*\|_2\le \rho^{T_1}\|\bbm\theta^0-\bbm\theta^*\|_2+C_\gamma\gamma_1,
\label{eq:phase_one_localization_bound}
\end{equation}
and the final estimator satisfies
\begin{equation}
\|\widehat{\bbm\theta}-\bbm\theta^*\|_2\le \rho^{T_2}\|\bbm\theta^{T_1}-\bbm\theta^*\|_2+C_\gamma\gamma_2.
\label{eq:phase_two_refinement_bound}
\end{equation}
Consequently,
\begin{equation}
\|\widehat{\bbm\theta}-\bbm\theta^*\|_2\le \rho^{T_1+T_2}\|\bbm\theta^0-\bbm\theta^*\|_2+C_\gamma\rho^{T_2}\gamma_1+C_\gamma\gamma_2 .
\label{eq:two_stage_combined_bound}
\end{equation}
\end{theorem}

The theorem formalizes the global localization and local refinement interpretation. Phase I contracts the initial error to a coarse radius of order \(\gamma_1\). Phase II then contracts this coarse radius to the sharper local floor \(\gamma_2\). In particular, if \(T_1\) and \(T_2\) are chosen so that \(\rho^{T_1}\|\bbm\theta^0-\bbm\theta^*\|_2\lesssim\gamma_1\) and \(\rho^{T_2}\gamma_1\lesssim\gamma_2\), then \eqref{eq:two_stage_combined_bound} gives \(\|\widehat{\bbm\theta}-\bbm\theta^*\|_2\lesssim\gamma_2\), up to constants depending only on the population-score regularity and the step-size factor.

The theorem also explains the iteration choices in Algorithm~\ref{alg:two_stage_right}. Phase I may require \(T_1\asymp\log n\) iterations to reduce a generic initial error to the coarse localization radius. Phase II starts from this localized estimator and only needs to close the gap between \(\gamma_1\) and \(\gamma_2\); in the sparse linear regression application, this leads to \(T_2\asymp\log s^*\) local refinement steps. Thus the deterministic convergence analysis matches the delayed sample-splitting strategy introduced in Section~\ref{subsec:path_dependence_delayed_splitting}.

Theorem~\ref{thm:two_stage_right_convergence} provides a recipe for model-specific analysis: verify SRCG for the population score, establish the appropriate SRS bounds for the MoM gradient estimator, and choose proper $s$, $K_1$, $K_2$, $T_1$, $T_2$, and $\eta$ so that the multiplicative SRS terms preserve contraction and the additive terms determine the statistical radius.

\vspace{-10pt}
\section{Model-Specific Stochastic Analysis}
\label{sec:4}

This section verifies the conditions in Section~\ref{sec:theory} for sparse regression models. We focus on heavy-tailed sparse linear regression, where the score exhibits the two-regime structure developed in Section~\ref{subsec:two_gradient_regimes}. We then discuss sparse logistic regression as a contrasting example.

\vspace{-10pt}
\subsection{Heavy-Tailed Sparse Linear Regression}
\label{subsec:linear_regression_upper}

We first consider the linear model
\begin{equation}
    y_i=\bm x_i^\top\bbm\theta^*+\epsilon_i,
\end{equation}
where $\mathbb{E}[\epsilon_i|\bm x_i]=0$ and $\|\bbm\theta^*\|_0\le s^*$. For simplicity, we assume that the intercept has been removed. Under the least-squares loss in \eqref{eq:linear_regression_loss}, the population score is $G(\bbm\theta)=\bm\Sigma_{\bm x}(\bbm\theta-\bbm\theta^*)$, where $\bm\Sigma_{\bm x}=\mathbb{E}[\bm x_i\bm x_i^\top]$. Throughout this subsection, let $\bar s=2s+s^*$, and assume $s\ge2$. In this subsection, we denote $\Theta_s=\{\bbm\theta\in\mathbb R^p:\|\bbm\theta\|_0\le s\}$.

For any $q>0$ and integer $\ell\ge1$, define the sparsity-restricted design moment $M_{\bm x,q,\ell}=\sup_{\|\bm v\|_2=1,\,\|\bm v\|_0\le \ell}\mathbb{E}[|\bm x_i^\top\bm v|^q]$ and the conditional noise moment $M_{\epsilon,q}=\sup_{\bm x}\mathbb{E}[|\epsilon_i|^q\mid \bm x_i=\bm x]$.

\begin{assumption}[Sparsity-restricted moments for linear regression]
\label{assump:linear_moments}
There exist tail indices $\lambda\in(0,1]$ and $\delta\in(0,1]$ such that $M_{\bm x,2+2\lambda,\bar s}<\infty$ and $M_{\epsilon,1+\delta}<\infty$.
\end{assumption}

\begin{assumption}[Sparsity-restricted eigenvalues]
\label{assump:linear_sparse_eigenvalues}
The covariance matrix $\bm\Sigma_{\bm x}=\mathbb{E}[\bm x_i\bm x_i^\top]$ satisfies $\kappa_-\|\bm v\|_2^2\le \bm v^\top\bm\Sigma_{\bm x}\bm v\le \kappa_+\|\bm v\|_2^2$,
for all $\bm v\in\mathbb R^p$ with $\|\bm v\|_0\le \bar s$, where $0<\kappa_-\le\kappa_+<\infty$. We denote the condition number by $\kappa=\kappa_+/\kappa_-$.
\end{assumption}

Since $1+\delta\le2\le2+2\lambda$, the lower-order design moment $M_{\bm x,1+\delta,\bar s}$ is finite under Assumption~\ref{assump:linear_moments}. We define the effective $(1+\delta)$-th moment of the noise-driven score by $M_{\mathrm{eff},1+\delta}=M_{\epsilon,1+\delta}M_{\bm x,1+\delta,\bar s}$.
This quantity captures the tail behavior of the covariate-noise product $\epsilon_i\bm x_i$ along sparse directions. The parameter $\lambda$ controls the tail behavior of the design-driven term $\bm x_i\bm x_i^\top(\bbm\theta-\bbm\theta^*)$, whereas $\delta$ controls the tail behavior of the noise-driven score.

We first verify the population-score regularity condition.

\begin{proposition}[SRCG for linear regression]
\label{prop:linear_srcg}
Under Assumption~\ref{assump:linear_sparse_eigenvalues}, the population score map $G(\bbm\theta)=\bm\Sigma_{\bm x}(\bbm\theta-\bbm\theta^*)$ satisfies the $(s,a,b)$-SRCG condition with $a=\kappa_-/(2\kappa_+^2)$ and $b=\kappa_-/2$.
\end{proposition}

The next proposition verifies sparse gradient stability for the coordinate-wise MoM gradient estimator. It provides the two forms required by the two-stage analysis: uniform SRS for global localization and fixed-point SRS for local refinement.

\begin{proposition}[Linear-regression SRS for MoM gradients]
\label{prop:linear_mom_srs}
Let $g(\bbm\theta;K,\mathcal A)$ be the coordinate-wise MoM gradient estimator in \eqref{eq:median_of_means_estimator}, computed on a data subset $\mathcal A$ with $m=|\mathcal A|$. Under Assumption~\ref{assump:linear_moments}, there exist constants $c,C>0$, depending only on $\lambda$ and $\delta$, such that the following statements hold. Define
\begin{equation}
\phi_{\mathrm{lin}}(m,K)=C\sqrt{\bar s}\,M_{\bm x,2+2\lambda,\bar s}^{1/(1+\lambda)}\left(\frac{K}{m}\right)^{\lambda/(1+\lambda)}
\text{ and  }
\gamma_{\mathrm{lin}}(m,K)=C\sqrt{\bar s}\,M_{\mathrm{eff},1+\delta}^{1/(1+\delta)}\left(\frac{K}{m}\right)^{\delta/(1+\delta)}.
\label{eq:linear_gamma}
\end{equation}

\begin{enumerate}
\item If $K\ge C\bar s\log p$ and $K\le cm$, then with probability at least \(1-C\exp(-c\bar s\log p)\), $g(\cdot;K,\mathcal A)$ satisfies uniform $(\Theta_s,\bar s,\phi_{\mathrm{lin}}(m,K),\gamma_{\mathrm{lin}}(m,K))$-SRS over $\Theta_s$.

\item For any deterministic $\bbm\theta\in\Theta_s$, if $K\ge C\log p$ and $K\le cm$, then with probability at least $1-Cp^{-c}$, $g(\bbm\theta;K,\mathcal A)$ satisfies fixed-point $(\bar s,\phi_{\mathrm{lin}}(m,K),\gamma_{\mathrm{lin}}(m,K))$-SRS at $\bbm\theta$.
\end{enumerate}
\end{proposition}
This proposition highlights that delayed sample-splitting avoids the spatial complexity waste of uniform convergence. Full iterate control under no sample splitting demands uniform SRS, requiring order $\bar s\log p$ MoM blocks, whereas our two-stage strategy needs only $\log p$ blocks in the second phase, yielding sharper statistical rates.

\begin{remark}[Design-driven gradient stability and adpativity of MoM estimation]
Suppose $n\asymp n_1\asymp n_2$, $K_1\asymp \bar s \log p$, $K_2\asymp \log p$, $T_1\asymp\log n$, and $T_2\asymp\log s$. 
The SRS bounds for the two stages are $\phi_{\mathrm{lin}}(n_1,K_1)\Vert\bbm\theta^t-\bbm\theta^*\Vert_2+\gamma_{\mathrm{lin}}(n_1,K_1)$ and $\phi_{\mathrm{lin}}(n_2/T_2,K_2)\Vert\bbm\theta^t-\bbm\theta^*\Vert_2+\gamma_{\mathrm{lin}}(n_2/T_2,K_2)$, respectively.

First, we emphasize an important finding: the design tail index impacts the gradient stability at the global localization phase.
Specifically, for iterates in Phase I, the estimation error $\|\bbm\theta^t-\bbm\theta^*\|_2$ may be relatively large, making the design-driven term dominant. Consequently, the gradient error is primarily controlled by the design tail index $\lambda$, i.e.,
\begin{equation}\label{eq:linear_phase_one_gradient_error}
    \|[g(\bbm\theta^t;K_1,\mathcal D^{(1)})-G(\bbm\theta^t)]_S\|_2\lesssim \sqrt{\bar s}\,M_{\bm x,2+2\lambda,\bar s}^{1/(1+\lambda)}\left(\frac{\bar s\log p}{n}\right)^{\lambda/(1+\lambda)}\|\bbm\theta^t-\bbm\theta^*\|_2.
\end{equation}
Once the iterates transition into Phase II, i.e., when $\|\bbm\theta^t-\bbm\theta^*\|_2\le \gamma_{\mathrm{lin}}(n_1,K_1)$, the SRS error is at most of order $n^{-\delta/(1+\delta)}$. This bound illustrates that the gradient performance is ultimately controlled by the tail behavior of the noise-driven score.

Second, the SRS bound of coordinate-wise MoM gradients reveals the decoupled effects of design and noise tails on $\phi$ and $\gamma$. The shifting dominant term along the optimization path elegantly captures the gradient's two-regime structure, demonstrating the adaptivity of MoM estimation.
\end{remark}

Combining Propositions~\ref{prop:linear_srcg} and \ref{prop:linear_mom_srs} with Theorem~\ref{thm:two_stage_right_convergence} gives the main upper bound for heavy-tailed sparse linear regression.

\begin{theorem}[TS-RIGHT for heavy-tailed sparse linear regression]
\label{thm:linear_two_stage_right}
Suppose Assumptions~\ref{assump:linear_moments} and \ref{assump:linear_sparse_eigenvalues} hold. Run Algorithm~\ref{alg:two_stage_right} with $n_1\asymp n_2\asymp n$, working sparsity $s\ge c_s s^*$ with $s\asymp s^*$, step size $\eta=2a\eta_0$ for some $\eta_0\in(0,1]$, MoM block numbers $K_1\asymp \bar s\log p$ and $K_2\asymp\log p$, and iteration numbers $T_1\asymp\log n$ and $T_2\asymp\log s$. Suppose the initialization satisfies $\Vert\bbm\theta^0\Vert_0\le s$ and $\|\bbm\theta^0-\bbm\theta^*\|_2\le R_0$, where $R_0$ is at most polynomial in $n$. If
\begin{equation}
n\gtrsim_{\kappa,\lambda,\eta_0} M_{\bm x,2+2\lambda,\bar s}^{1/\lambda}\bar s^{1+(1+\lambda)/(2\lambda)}\log p,
\label{eq:linear_sample_size_condition}
\end{equation}
then, with probability at least $1-C\exp(-c\bar s\log p)-Cp^{-c}$, the phase-I estimator satisfies
\begin{equation}
\|\bbm\theta^{T_1}-\bbm\theta^*\|_2
\lesssim_{\kappa,\lambda,\delta,\eta_0}
\sqrt{\bar s}\,M_{\mathrm{eff},1+\delta}^{1/(1+\delta)}
\left(\frac{\bar s\log p}{n}\right)^{\delta/(1+\delta)}
\label{eq:linear_phase_one_rate}
\end{equation}
and the final estimator $\widehat{\bbm\theta}$ satisfies
\begin{equation}
\|\widehat{\bbm\theta}-\bbm\theta^*\|_2
\lesssim_{\kappa,\lambda,\delta,\eta_0}
\sqrt{\bar s}\,M_{\mathrm{eff},1+\delta}^{1/(1+\delta)}
\left(\frac{\log s \cdot \log p}{n}\right)^{\delta/(1+\delta)}.
\label{eq:linear_final_rate}
\end{equation}
\end{theorem}

Theorem~\ref{thm:linear_two_stage_right} formalizes the two-regime phenomenon from Section~\ref{subsec:two_gradient_regimes}. The design-tail index $\lambda$ appears through the multiplicative stability term $\phi_{\mathrm{lin}}$, which impacts the sample complexity \eqref{eq:linear_sample_size_condition} required for the robust gradient step to contract, as in \eqref{eq:linear_phase_one_gradient_error}. The final error is determined by the additive stability term $\gamma_{\mathrm{lin}}$, which is governed by the effective noise-driven score moment $M_{\mathrm{eff},1+\delta}$ and the noise tail index $\delta$. In this sense, design tails control stable localization, while noise-score tails control final refinement.



Table~\ref{tab:comparison} compares our rates with existing results for heavy-tailed sparse linear regression. Compared with Huber-type losses, TS-RIGHT directly controls the design-driven product \(\bm x_i\bm x_i^\top(\bbm\theta-\bbm\theta^*)\), allowing for weaker design moment assumptions.
Unlike \citet{prasad2020robust,liu2019high}, where full sample splitting inflates the error by the total iterations $T$ (typically $\log n$), TS-RIGHT delays sample splitting until Phase II. Starting from a localized estimator, Phase II only requires $T_2\asymp\log s$ fresh batches to close the gap between $\gamma_1$ and $\gamma_2$, adding just a $\log s$ factor to the final rate. Furthermore, restricting the SRS condition to a sparse set avoids depending on $\mathrm{tr}(\Sigma_{\bm x})$, yielding a sharper $\sqrt{s^*}$ dependence. Overall, TS-RIGHT achieves near-optimal rates, up to $(\log s^*)^{\delta/(1+\delta)}$, which is extremely small in practice, under minimal moment conditions on the design and noise.

\begin{table}[!htp]
\centering
\renewcommand{\arraystretch}{1.3}
\small
\begin{threeparttable}
\caption{Comparison of assumptions and $\ell_2$ statistical rates for sparse linear regression.}
\label{tab:comparison}
\begin{tabular}{@{} l @{\hspace{2pt}} c @{\hspace{4pt}} c @{\hspace{4pt}} c @{}}
\toprule
\textbf{Method}  & \makecell{\textbf{Design} \\ \textbf{Moments}} & \makecell{\textbf{Noise} \\ \textbf{Moments}} & \textbf{Statistical Rate}\\
\midrule
LASSO/IHT  & Sub-Gaussian & Gaussian  & $\sqrt{s^*\log p/{n}}$\\
Adaptive Huber \citep{sun2020adaptive}\tnote{a}  & Sub-Gaussian & $1+\delta$ & $\sqrt{s^*}\left({\log p}/{n}\right)^{\delta/(1+\delta)}$\\
Robust GD \citep{prasad2020robust}\tnote{b}  & $4$th & 2nd & $\sqrt{{\mathrm{tr}(\Sigma_{\bm x})T(\log T+s^*\log p)}/{n}}$\\
Robust HT \citep{liu2019high}\tnote{c} & $4$th & 2nd & $\sqrt{Ts^*\log p/n}$\\
Data Shrinkage \citep{fan2021shrinkage} & $4$th & $4$th & $\sqrt{s^*\log p/n}$\\
MOM-LASSO \citep{lecue2020robust} &  $\log p$ & $2$nd & $\sqrt{s^*\log(ep/s^*)/n}$\\
\textbf{TS-RIGHT (Ours)}\tnote{d}  & \textbf{$2+2\lambda$} & \textbf{$1+\delta$} & $\sqrt{s^*}\left({\log s^* \cdot \log p}/{n}\right)^{\delta/(1+\delta)}$  \\
\bottomrule
\end{tabular}
\begin{tablenotes}
\footnotesize 
  \item[a] They also get $\sqrt{s^*{\log p}/{n}}$ for $4$th moment design and $4$th moment noise.
  \item[b] They get $\sqrt{{\mathrm{tr}(\Sigma_{\bm x})T\log(T/\varrho)}/{n}}$ with probability at least $1-\varrho$. Although they do not assume sparsity on the true parameter, we take $\varrho=\exp\{-s^*\log p\}$ for comparison.
  \item[c] They require $T=O(\log n)$.
  \item[d] Since $ s\asymp s^* $ and $\bar{s} \asymp s^* $, we replace $\bar s$ and $s$ with $s^*$ in the final rate for comparison.
\end{tablenotes}
\end{threeparttable}
\end{table}

\vspace{-18pt}
\subsection{Heavy-Tailed Sparse Logistic Regression}
\label{subsec:logistic_regression_upper}

We next consider sparse logistic regression as a contrasting example. Let \(y_i\in\{0,1\}\) and \(\mathbb{P}(y_i=1\mid \bm x_i)=\sigma(\bm x_i^\top\bbm\theta^*)\), where \(\sigma(t)=e^t/(1+e^t)\) and \(\|\bbm\theta^*\|_0\le s^*\). The negative log-likelihood loss is
\[
\mathcal L(\bbm\theta;Z_i)=\Phi(\bm x_i^\top\bbm\theta)-y_i\bm x_i^\top\bbm\theta,\qquad \Phi(t)=\log(1+e^t),
\]
and its score is
\begin{equation}
\nabla_{\bbm\theta}\mathcal L(\bbm\theta;Z_i)=\{\sigma(\bm x_i^\top\bbm\theta)-y_i\}\bm x_i .
\label{eq:logistic_score}
\end{equation}
Unlike the least-squares score, \eqref{eq:logistic_score} does not contain a quadratic design product. The scalar multiplier \(\sigma(\bm x_i^\top\bbm\theta)-y_i\) is bounded between \(-1\) and \(1\). Thus the logistic score is still heavy-tailed when \(\bm x_i\) is heavy-tailed, but its tail behavior is first-order in the covariates. This is the key reason that finite second moments of the covariates are sufficient for MoM gradient stability.

Throughout this subsection, let \(\bar s=2s+s^*\) and assume \(s\ge2\). We use the same sparsity-restricted design moment \(M_{\bm x,q,\ell}\) defined in Section~\ref{subsec:linear_regression_upper} and require Assumption~\ref{assump:linear_sparse_eigenvalues} on the covariance matrix \(\bm\Sigma_{\bm x}\). For the population-score regularity condition, we restrict attention to a bounded sparse parameter region. Let \(B(R)=\{\bbm\theta\in\mathbb R^p:\|\bbm\theta\|_2\le R\}\), $\Theta_s=\{\bbm\theta\in\mathbb R^p:\|\bbm\theta\|_0\le s\}$ and \(\Theta_s(R)=\{\bbm\theta\in B(R):\|\bbm\theta\|_0\le s\}\).

\begin{assumption}[Sparsity-restricted logistic curvature]
\label{assump:sparse_logistic_curvature}
There exists a constant $c_\kappa\in(0,1/4]$ such that, for every
$S\subset[p]$ with $|S|\le \bar s$, every
${\boldsymbol\theta}\in B(R)$ satisfying
$\operatorname{supp}({\boldsymbol\theta})\subseteq S$, and every
$\bm v\in\mathbb R^p$ satisfying
$\operatorname{supp}(\bm v)\subseteq S$, one has
\[
\mathbb E\left[
\sigma'(\bm x_i^\top{\boldsymbol\theta})
(\bm x_i^\top\bm v)^2
\right]
\ge
c_\kappa
\mathbb E\left[
(\bm x_i^\top\bm v)^2
\right].
\]
\end{assumption}

Assumption~\ref{assump:sparse_logistic_curvature} prevents the population logistic curvature from degenerating on the sparse parameter region. The condition is needed because \(\sigma'(t)\) can be close to zero when \(|t|\) is large. It is a population-level condition and does not require sub-Gaussian covariates. We verify this assumption for a multivariate \(t\)-distribution in Lemma~\ref{lem:student_t_sparse_logistic_curvature} of Appendix~\ref{append:E}.

We now verify the population-score regularity condition as follows.

\begin{proposition}[SRCG for logistic regression]
\label{prop:logistic_srcg}
Suppose Assumptions~\ref{assump:linear_sparse_eigenvalues} and \ref{assump:sparse_logistic_curvature} hold, and \(\bbm\theta^*\in B(R)\). Then the population logistic score map \(G(\bbm\theta)=\mathbb{E}[\{\sigma(\bm x_i^\top\bbm\theta)-y_i\}\bm x_i]\) satisfies the \((s,a,b)\)-SRCG condition over \(\Theta_s(R)\), with \(a=2/\kappa_+\) and \(b=c_\kappa\kappa_-/2\).
\end{proposition}

The next proposition verifies sparse gradient stability for the coordinate-wise MoM gradient estimator. Compared with Proposition~\ref{prop:linear_mom_srs}, the bound has no design-driven term.

\begin{proposition}[Logistic-regression SRS for MoM gradients]
\label{prop:logistic_mom_srs}
Let \(g(\bbm\theta;K,\mathcal A)\) be the MoM gradient estimator in \eqref{eq:median_of_means_estimator}, computed on a data subset \(\mathcal A\) with \(m=|\mathcal A|\). Suppose Assumption \ref{assump:linear_sparse_eigenvalues} holds. There exist universal constants \(c,C>0\) such that the following statements hold. Define $\gamma_{\mathrm{log}}(m,K)=C\sqrt{\bar s\,M_{\bm x,2,1}}\left(K/m\right)^{1/2}$.

\begin{enumerate}
\item If \(K\ge C\bar s\log p\) and \(K\le cm\), then with probability at least \(1-C\exp\{-c\bar s \log p\}\), \(g(\cdot;K,\mathcal A)\) satisfies uniform \((\Theta_s,\bar s,0,\gamma_{\mathrm{log}}(m,K))\)-SRS.

\item For any deterministic \(\bbm\theta\in\Theta_s\), if \(K\ge C\log p\) and \(K\le cm\), then with probability at least \(1-Cp^{-c}\), \(g(\bbm\theta;K,\mathcal A)\) satisfies fixed-point \((\bar s,0,\gamma_{\mathrm{log}}(m,K))\)-SRS at \(\bbm\theta\).
\end{enumerate}
\end{proposition}

The absence of a multiplicative term reflects the bounded residual multiplier in \eqref{eq:logistic_score}. In contrast to linear regression, the stochastic score does not contain \(\bm x_i\bm x_i^\top(\bbm\theta-\bbm\theta^*)\), and the MoM gradient error is controlled by the second moment of the covariates.

Combining Propositions \ref{prop:logistic_srcg} and \ref{prop:logistic_mom_srs} with Theorem~\ref{thm:two_stage_right_convergence} yields the following convergence guarantee.

\begin{theorem}[TS-RIGHT for heavy-tailed sparse logistic regression]
\label{thm:logistic_two_stage_right}
Suppose Assumptions~\ref{assump:linear_sparse_eigenvalues} and \ref{assump:sparse_logistic_curvature} hold. Run Algorithm~\ref{alg:two_stage_right} with \(n_1\asymp n_2\asymp n\), working sparsity \(s\ge c_s s^*\) with $s \asymp s^*$, step size \(\eta=2a\eta_0\) for some \(\eta_0\in(0,1]\), MoM block numbers \(K_1\asymp\bar s\log p\) and \(K_2\asymp\log p\), and iteration numbers \(T_1\asymp\log n\) and \(T_2\asymp\log s\). Suppose the initialization satisfies \(\Vert\bbm\theta^0\Vert_0\le s\), \(\|\bbm\theta^0-\bbm\theta^*\|_2\le R_0\), where \(R_0\) is at most polynomial in \(n\), and let \(R=R_0+\Vert\bbm\theta^*\Vert_2+1\). If 
\[n\gtrsim_{c_\kappa,\kappa,\eta_0} M_{\bm x,2,1}\bar s^2\log p,\]
 then, with probability at least \(1-C\exp\{-c\bar s \log p\}-Cp^{-c}\), the phase-I estimator satisfies
\begin{equation}
\|\bbm\theta^{T_1}-\bbm\theta^*\|_2
\lesssim_{c_\kappa,\kappa,\eta_0}
\sqrt{\bar s\,M_{\bm x,2,1}}
\left(\frac{\bar s\log p}{n}\right)^{1/2}
\label{eq:logistic_phase_one_rate}
\end{equation}
and the final estimator satisfies
\begin{equation}
\|\widehat{\bbm\theta}-\bbm\theta^*\|_2
\lesssim_{c_\kappa,\kappa,\eta_0}
\sqrt{\bar s\,M_{\bm x,2,1}}
\left(\frac{\log s\cdot \log p}{n}\right)^{1/2}.
\label{eq:logistic_final_rate}
\end{equation}
\end{theorem}

\begin{remark}[Bounded residual multiplier]
The logistic gradient is not bounded when the covariates are heavy-tailed. What is bounded is the scalar residual multiplier \(\sigma(\bm x_i^\top\bbm\theta)-y_i\). This bounded multiplier prevents the quadratic amplification of covariate tails that appears in least-squares regression. Consequently, finite second moments of the covariates are sufficient for sparse MoM gradient stability, and TS-RIGHT attains the square-root local refinement rate in \eqref{eq:logistic_final_rate}. Similar behavior is also observed by \citet{lecue2020robustclassification} who utilize the Lipschitz property of logistic-type classification losses.
\end{remark}

Table~\ref{tab:logistic_comparison} compares our logistic regression rates with existing results. The statistical rate of TS-RIGHT is an analogue of the standard sparse logistic rate under finite second moments, up to an additional \(\log s\) factor introduced by the local refinement batches. Furthermore, when compared with \citet{prasad2020robust}, the delayed-splitting strategy and the restricted SRS condition demonstrate advantages similar to those observed in linear regression.

\vspace{-1pt}
\begin{table}[bhtp]
\centering
\renewcommand{\arraystretch}{1.3}
\setlength{\tabcolsep}{2pt}
\small
\begin{threeparttable}
\caption{Comparison of assumptions and statistical rates for sparse logistic regression.}
\label{tab:logistic_comparison}
\begin{tabular}{p{0.45\linewidth}p{0.21\linewidth}p{0.31\linewidth}}
\toprule
\textbf{Method}  & \textbf{Design Moments} &  \textbf{Statistical Rate}\\
\midrule

Logistic LASSO 
\citep{negahban2012unified}
& Sub-Gaussian 
& $\sqrt{s^*\log p/n}$ \\

Robust GD 
\citep{prasad2020robust}\tnote{a}
& $4$th
& $\sqrt{{\mathrm{tr}(\Sigma_{\bm x})T(\log T+s^* \log p)}/{n}}$ \\

Robust HT \citep{liu2019high}\tnote{b}  & 2nd & $\sqrt{Ts^*\log p/n}$\\

Feature Shrinkage 
\citep{zhu2021taming}
& $4$th
& $\sqrt{s^*\log p/n}$ \\

\textbf{TS-RIGHT (Ours)}\tnote{c}
& \textbf{2nd Mom.}
& $\sqrt{{s^*\log s^*\cdot\log p}/{n}}$ \\

\bottomrule
\end{tabular}
\begin{tablenotes}
\footnotesize
  \item[a] They get $\sqrt{{\mathrm{tr}(\Sigma_{\bm x})T\log(T/\varrho)}/{n}}$ with probability at least $1-\varrho$. Although they do not assume sparsity on the true parameter, we take $\varrho=\exp\{-s^*\log p\}$ for comparison.
  \item[b] They require $T=O(\log n)$.
  \item[c] Since $s\asymp s^*$ and $\bar s=2s+s^*\asymp s^*$, we replace $s$ and $\bar s$ by $s^*$ in the displayed final rate. Moment and curvature constants are suppressed.
\end{tablenotes}
\end{threeparttable}
\end{table}

\vspace{-18pt}
\section{Phase-Wise Lower-Bound Interpretation}
\label{sec:lower_bound}

We focus in this section on sparse linear regression, where the two-phase gradient structure is most transparent and the distinction between design-driven localization and noise-driven refinement is essential. The corresponding lower-bound discussion for sparse logistic regression is given in Appendix~\ref{sec:appendix_minimax}. Our goal here is not to establish a unified minimax theorem for the full random-design sparse regression problem, but rather to explain the phase-wise lower-bound interpretation of the two terms appearing in the SRS bound.

The upper theory for sparse linear regression is driven by the SRS
decomposition
\[
    \|[g(\bbm\theta)-G(\bbm\theta)]_S\|_2
    \le
    \phi_{\rm lin}\|\bbm\theta-\bbm\theta^*\|_2+\gamma_{\rm lin}.
\]
The two terms correspond to different statistical tasks along the optimization path. The multiplicative term controls whether the estimated gradient field is stable enough to yield global localization, whereas the additive term controls the final stochastic floor after localization has been achieved. We now explain why both terms have intrinsic lower-bound interpretations.

The additive term is calibrated by existing lower-bound theory for adaptive Huber regression. In particular, \citet{sun2020adaptive} show that under only a finite \((1+\delta)\)-moment condition on the regression noise, the final statistical error cannot generally improve on the tail exponent $\xi(\delta)$, where $\xi(a)=\min\{a/(1+a),1/2\}$. In sparse high-dimensional regression, this corresponds to the benchmark $\sqrt{s^*}(\log p/n)^{\xi(\delta)}$, up to constants and logarithmic factors. Thus the exponent appearing in the
additive term \(\gamma_{\rm lin}\) is not new to this paper; it is the known
local-refinement barrier caused by heavy-tailed noise. For self containedness, we provide the precise statement of this lower bound in Theorem~\ref{thm:fixed_design_local_refinement_lower} of Appendix~\ref{append:local_refinement_lower}.
Our lower-bound
contribution concerns the complementary multiplicative term, which is specific
to the design-driven gradient stability needed for global localization.

To state this lower bound, fix a support \(S\subset[p]\) with
\(|S|=s^*\). For a design distribution \(P\), write
\(G_P(\Delta)=\bbm\Sigma_P\Delta\), where
\(\bbm\Sigma_P=\mathbb E_P(\bm x\bm x^\top)\). Let
\(\mathcal Q_\lambda(S)\) be a class of design distributions satisfying sparse
eigenvalue bounds on the relevant sparse directions and a finite
sparsity-restricted \((2+2\lambda)\)-moment condition
\(M_{\bm x,2+2\lambda,s^*}\le M_x\). The precise construction of
\(\mathcal Q_\lambda(S)\) is given in Appendix~\ref{sec:appendix_minimax} of Supplementary Materials.

\begin{theorem}[Sparse operator-norm lower bound for gradient stability]
\label{thm:sparse_operator_gradient_lower}
Let \(0<\lambda\le1\), and suppose \(s^*\le c n\) for a sufficiently small
constant \(c>0\). There exists a class \(\mathcal Q_\lambda(S)\) satisfying the
conditions above such that
\begin{equation}\label{eq:sparse_operator_gradient_lower}
\inf_{\widehat G}
\sup_{P\in\mathcal Q_\lambda(S)}
\mathbb E_P\left[
\sup_{\Delta:\ \|\Delta\|_2=1,\ \operatorname{supp}(\Delta)\subset S}
\|\widehat G(\Delta)-G_P(\Delta)\|_2
\right]
\ge
c_\lambda M_x^{1/(1+\lambda)}
\left(\frac{s^*}{n}\right)^{\lambda/(1+\lambda)},
\end{equation}
where the infimum is over all estimators of the gradient map based on \(n\)
independent design observations from \(P\), and \(c_\lambda>0\) depends only on
\(\lambda\).
\end{theorem}

Theorem~\ref{thm:sparse_operator_gradient_lower} shows that the
\(\lambda\)-dependent rate in the multiplicative SRS term is not merely a proof
artifact. Under only finite sparsity-restricted \((2+2\lambda)\)-moments of
the covariates, the sparse design-gradient map
\(\Delta\mapsto \bbm\Sigma_{\bm x}\Delta\) cannot generally be estimated over
sparse directions at a faster tail-dependent exponent. Thus the design-tail
index \(\lambda\) controls the intrinsic difficulty of stable global
localization.

Combining this theorem with the known adaptive-Huber local benchmark gives the
decoupled lower-bound interpretation $\phi_{\rm lin}\|\bbm\theta-\bbm\theta^*\|_2 + \gamma_{\rm lin}$.
The two exponents have the same functional form, but they govern different objects. The exponent \(\xi(\lambda)\) limits
design-driven gradient stability and hence global localization, while
\(\xi(\delta)\) limits the final local refinement accuracy. Therefore, these
results should be interpreted as phase-wise lower-bound benchmarks rather than
as a unified random-design minimax theorem for the full sparse regression
problem. In short, design tails limit stable localization, whereas noise-score
tails limit final refinement.

\vspace{-10pt}
\section{Numerical Studies}\label{sec:numerical_studies}

This section presents numerical experiments validating our theory. First, we verify how the estimation error rate and phase-I gradient stability depend on the noise tail index $\delta$ and design tail index $\lambda$, respectively. Then we benchmark the proposed TS-RIGHT algorithm against other robust high-dimensional methods in linear and logistic regression. We also demonstrate the computational improvement from the second stage and compare the delayed sample-splitting strategy with no splitting and equal splitting; these additional results are deferred to Appendix~\ref{append:delayed_splitting_advantage} of Supplementary Materials.

\vspace{-10pt}
\subsection{Verification of Estimation Error Rate Adaptation}

We first verify the theoretical prediction of Theorem \ref{thm:linear_two_stage_right}, which states that the convergence rate adapts to the noise tail index $\delta$ as $\mathcal{O}(n^{-\delta/(1+\delta)})$. Specifically, after sufficient iterations, the log-error should scale as:
\begin{equation}
\log \left( \Vert\widehat{\bbm\theta}-\bbm\theta^*\Vert_2 \right) \approx -\zeta(\delta) \log n + C_{\delta,\bar s}, \quad \text{with } \zeta(\delta) = \min\left(\delta/(1+\delta), 1/2\right).
\end{equation}

We set $n \in [300, 20006]$, $p=600$, and $s^*=5$ with $\bbm\theta^*=(5,-5,6,-6,7,0,0,\dots,0)^\top$. The covariates are generated from a standard Gaussian distribution, $\bm{x}_i\sim N_p(\bm 0,\bm I_p)$, to isolate the impact of heavy-tailed noise. The noise follows a Student's $t$-distribution with degrees of freedom $\nu \in \{1.2, 1.4, 1.6, 1.8, 2.5, 6\}$. These correspond to tail indices $\delta \in \{0.15, 0.35, 0.55, 0.75, 1.45, 4.95\}$ (since the $(1+\delta)$-th moment exists if $\delta < \nu - 1$). We take the median $l_2$ norm of estimation errors over 100 trials.

Figure \ref{fig:rate_plot_linear_regression} plots the log-error versus $\log n$. To facilitate comparison of slopes, curves are normalized to intersect at $n=366$. The accompanying table reports the empirical slopes obtained via linear regression on the log-log data.
The results are consistent with the theoretical prediction: as $\delta$ increases, the slope steepens, closely matching the theoretical prediction $-\delta/(1+\delta)$. Furthermore, we observe the predicted saturation: for $\delta > 1$ (i.e., $\delta=1.45$ and $4.95$), the slopes stabilize near $-0.5$, confirming that once the noise has finite variance, the rate locks into the standard parametric rate $n^{-1/2}$.
\begin{figure}[htb]
    \centering
    \begin{minipage}[c]{0.64\textwidth}
        \centering
        \includegraphics[width=\linewidth]{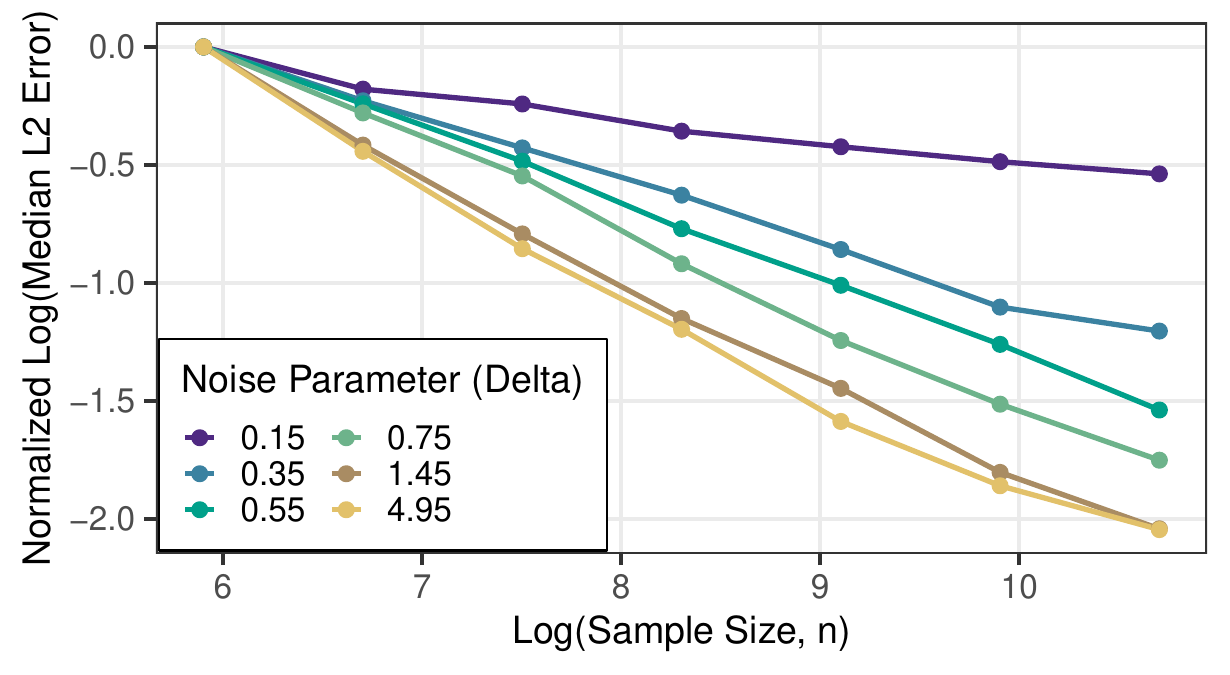}
    \end{minipage}
    \hfill 
    \begin{minipage}[c]{0.35\textwidth}
        \centering
       \vspace{-26pt}
        {\small Empirical vs. Theoretical Slopes\par}
         \vspace{4pt}
        \resizebox{\linewidth}{!}{%
            \begin{tabular}{@{}cccc@{}}
                \toprule
                $\delta$ & Empirical & Theoretical & $R^2$\\
                \midrule
                0.15 & -0.108 & -0.130 & 0.962 \\
                0.35 & -0.259 & -0.259 & 0.995 \\
                0.55 & -0.321 & -0.355 & 1.000 \\
                0.75 & -0.376 & -0.429 & 0.997 \\
                1.45 & -0.427 & -0.500 & 0.995 \\
                4.95 & -0.433 & -0.500 & 0.987 \\
                \bottomrule
            \end{tabular}%
        }
    \end{minipage}
    \vspace{-0.5cm}
    \caption{Estimation error rates of the TS-RIGHT across different $\delta$.}
    \label{fig:rate_plot_linear_regression}
\end{figure}

\vspace{-10pt}
\subsection{Verification of Gradient Stability Barrier}
We verify that when $\Vert\bbm\theta-\bbm\theta^*\Vert_2$ is large, the stability of the gradient estimation is limited by the design tail index $\lambda$, as shown in both \eqref{eq:linear_phase_one_gradient_error} and \eqref{eq:sparse_operator_gradient_lower} for linear regression. We use the MoM estimator with a fixed $\bbm\theta$ away from the true parameter and the whole dataset, and we expect the estimation error to scale as:
\begin{equation}
\log \Vert {\bm g}(\bbm\theta;K;\mathcal D_n)-G(\bbm\theta)\Vert_2 \asymp -\zeta(\lambda)\log n +C,
\end{equation} 
where $\zeta(\lambda)=\min(\lambda/(1+\lambda),1/2)$.

We vary the sample size $n$ over $[300, 9000]$, fix $p=600$, and set the ground-truth parameter to $\bbm\theta^*=(5,-5,6,-6,7,0,\dots,0)$. Features are drawn from a multivariate $t$ distribution with degrees of freedom $\nu\in\{2.2,2.4,2.8,3.2,3.6,22,32\}$, corresponding to moment parameters $\lambda\in\{0.09,0.19,0.39,0.59,0.79,9.99,15.99\}$. The noise follows a standard normal distribution. The gradient is estimated at $\bbm\theta=(0,0,0,0,0,5,-5,6,-6,7,0,\dots,0)$. 

Results are averaged over $200$ independent replications and summarized in Figure \ref{fig:sample_complexity}. The empirical slopes steepen as $\lambda$ increases up to $\lambda=1$ and closely match the theoretical slopes given by $\zeta(\lambda)$, thereby validating the dependence of the gradient stability on the design tail index for early iterations in Phase I.
\begin{figure}[htbp]
    \centering
    \begin{minipage}[c]{0.64\textwidth}
        \centering
        \includegraphics[width=\linewidth]{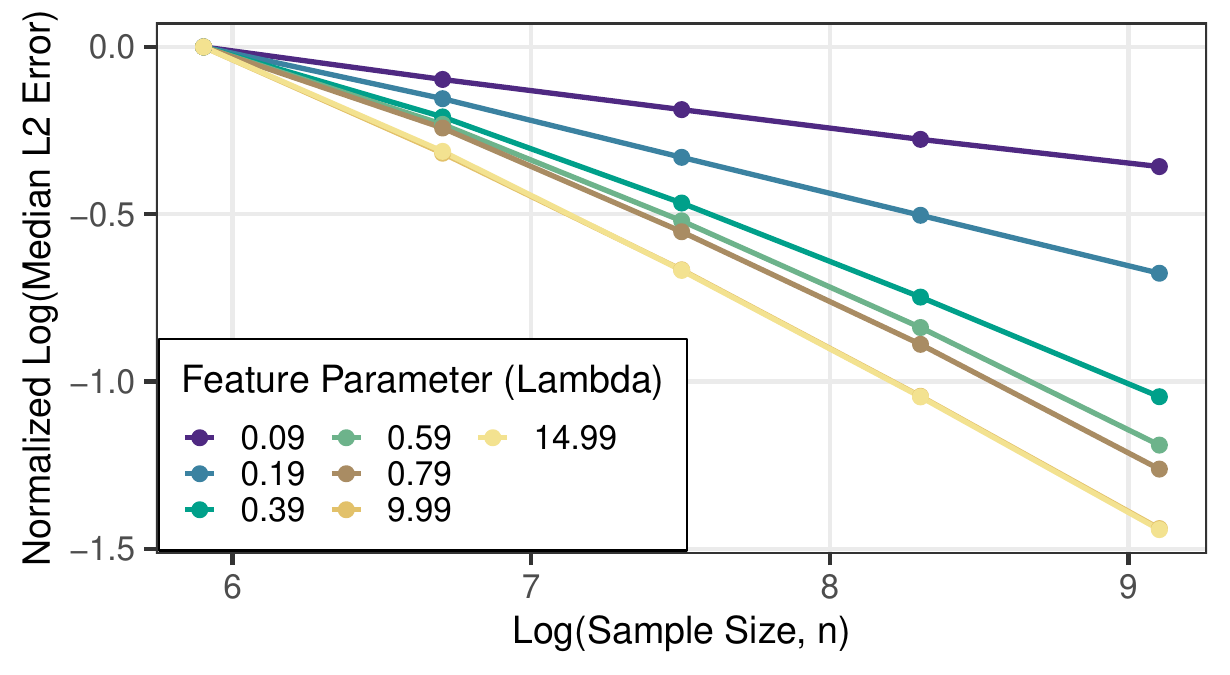}
    \end{minipage}
    \hfill 
    \begin{minipage}[c]{0.35\textwidth}
        \centering
        \vspace{-28pt}
        {\small Empirical vs. Theoretical Slopes\par}
         \vspace{4pt}
        
        \resizebox{\linewidth}{!}{%
        \renewcommand{\arraystretch}{0.8}
    \begin{tabular}{@{}cccc@{}}
        \toprule
        $\lambda$ & Empirical  & Theoretical & $R^2$\\
        \midrule
        0.09 & -0.112  & -0.083 & 0.999\\
        0.19 & -0.213 & -0.160 & 1.000 \\
        0.39 & -0.329 & -0.281 & 0.996 \\
        0.59 & -0.373 & -0.371 & 0.994 \\
        0.79 & -0.396 & -0.441 & 0.994 \\
        9.99 & -0.450 & -0.500 & 0.998 \\
        15.99 & -0.452 & -0.500 & 0.998 \\
        \bottomrule
    \end{tabular}%
}
    \end{minipage}
    \vspace{-0.5cm}
    \caption{Convergence rates of gradient estimation error across different $\lambda$.}
    \label{fig:sample_complexity}
\end{figure}

\vspace{-10pt}
\subsection{Comparative Performance Evaluation}

We benchmark TS-RIGHT against leading robust sparse regression methods: Lasso \citep{tibshirani1996regression}, Adaptive Huber \citep{sun2020adaptive}, Shrinkage Method \citep{fan2021shrinkage}, and standard IHT \citep{jain2014iterative} in linear and logistic regression. Only linear regression results are presented here; logistic regression results are given in Appendix~\ref{append:logistic_comparison}.

We generate data $\mathbf{y} = \mathbf{X}\bbm\theta^* + \bbm\epsilon$ with $p=800$, and varying $n$. To create both heavy-tailed design and noise, we draw rows of $\mathbf{X}$ from a multivariate $t$-distribution with $\nu_x=2.5$ (infinite 4th moments), and $\epsilon$ from a $t$-distribution with $\nu_\epsilon=1.5$ (infinite variance).
Hyperparameters for all methods are tuned via cross-validation on an independent validation set (see details in Appendix~\ref{app:tuning-linear-comparison} of Supplementary Materials).

Figure \ref{fig:linear_comparison} presents the results. Non-robust methods (IHT and Lasso) exhibit wider error distributions than robust methods (Adaptive Huber, Shrinkage, and TS-RIGHT). Specifically, even with careful tuning, estimation errors of IHT display a heavy upper tail, with many runs diverging to large errors due to extreme covariate values generated by the heavy-tailed design. In contrast, TS-RIGHT consistently achieves low estimation errors across most sample sizes. This advantage stems from our gradient-based handling of the design matrix, whereas shrinkage methods may struggle with low design moments and suffer from information loss due to truncation.
\begin{figure}[!htp]
  \centering
     \includegraphics[width=0.75\textwidth,height=6.8cm]{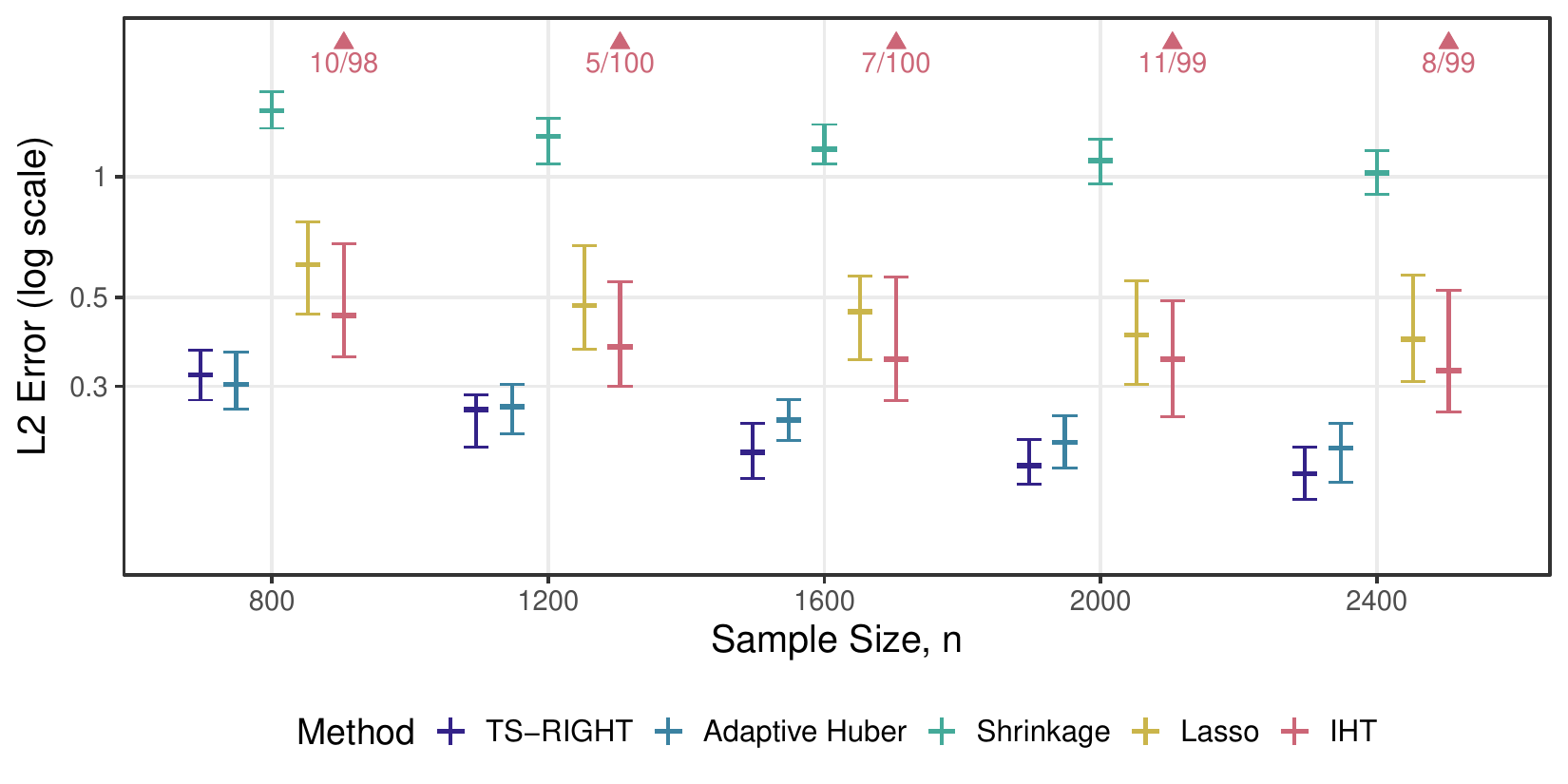}
    \caption{Linear regression performance comparison. Vertical lines show 0.25, 0.75 quantiles; thick center ticks show the median $l_2$ error. IHT triangles mark values above the display cap, where the display cap is set to the 99th percentile of the non-IHT methods' errors.}
     \label{fig:linear_comparison}
 \end{figure}

\vspace{-10pt}
\subsection{Real Data Application}
We used the GDSC:277 Linifanib response task from the GDSC/CCLE cancer drug-response data processed in the RedCDR repository. After filtering missing values, the design matrix $X\in \mathbb R^{524 \times 697}$ contains 697 coding-gene expression features for 524 cancer cell lines, and the response vector \(y\in\mathbb{R}^{524}\) contains the corresponding log-scale IC50 sensitivity scores for Linifanib. The data were obtained from the RedCDR GitHub repository (\url{https://raw.githubusercontent.com/mhxu1998/RedCDR}).

This high-dimensional linear regression task exhibits heavy-tailed design and noise: $84\%$ of features have kurtosis $>3$ ($14\%$ exceed $10$), and the response kurtosis is 30.31. We robustly center the response by its median and standardize all covariates using the median and the median absolute deviation (MAD). The data are randomly partitioned into a training set ($80\%$) and a testing set ($20\%$) and repeated over 100 independent splits. The RMSE results are summarized in Figure~\ref{fig:real_data_comparison}. The distribution of TS-RIGHT is shifted to the left relative to the competing methods, indicating better predictive accuracy on this real-data task. The improvement is most apparent compared with the shrinkage estimator, while Adaptive Huber, IHT, and Lasso show partially overlapping but generally higher RMSE values.

\begin{figure}[!htp]
  \centering
     \includegraphics[width=0.8\textwidth,height=7.3cm]{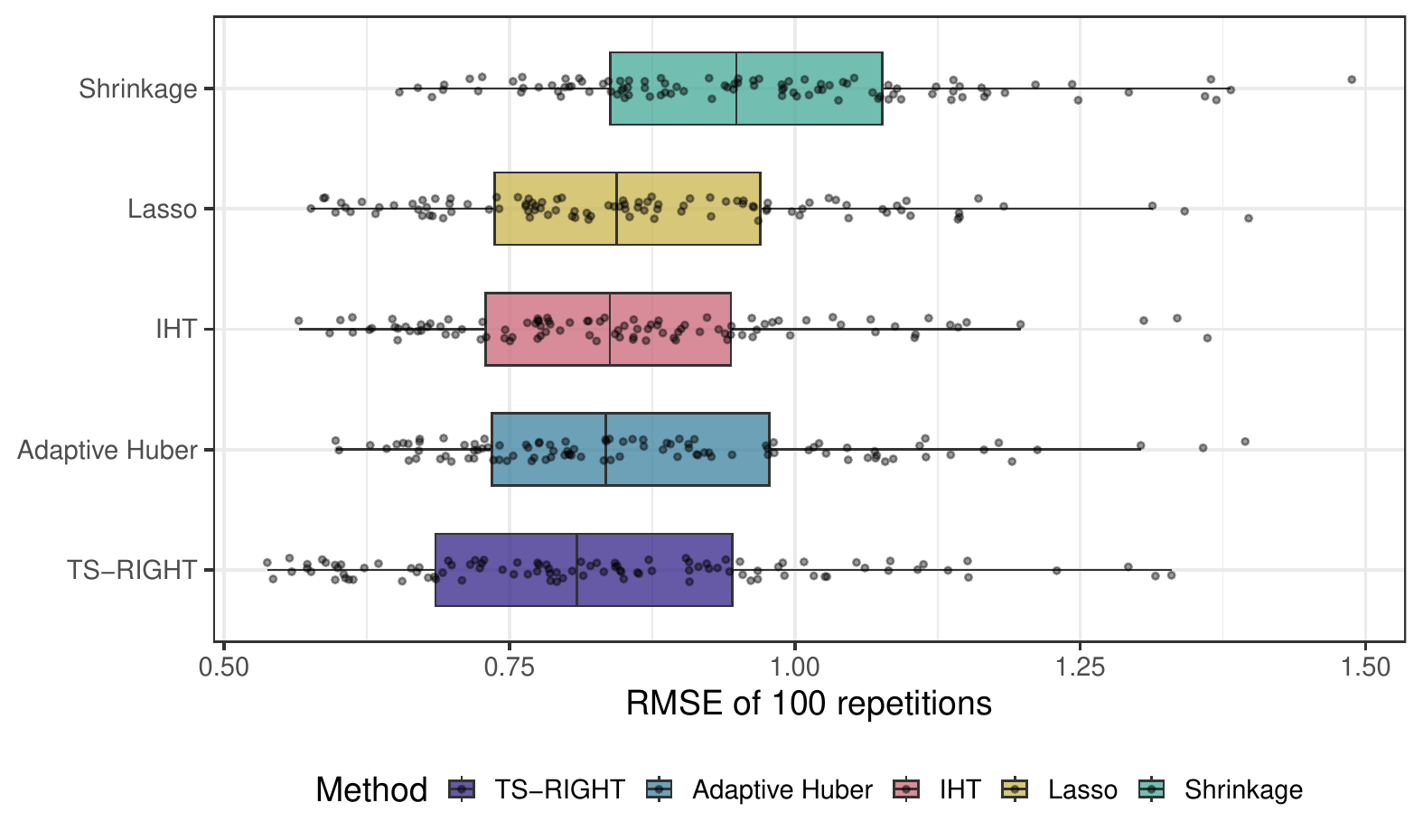}
    \caption{RMSE comparison on the GDSC:277 Linifanib response task. }
     \label{fig:real_data_comparison}
 \end{figure}

\section{Conclusion and Discussion}
\label{sec:conclusion}
We propose a robust and computationally efficient first-order method for heavy-tailed sparse regression. By combining a coordinate-wise median-of-means gradient estimator with a two-stage delayed-splitting strategy, the proposed method achieves near-optimal statistical rates. Identifying two distinct gradient regimes and corresponding optimization phases provides new insights into how heavy tails affect the optimization landscape. While our theory focuses on regression, these principles may extend to other high-dimensional heavy-tailed estimation problems.

There are several avenues for future work. First, while we focus on the sparse vector structure, it would be interesting to extend the framework to other structures, such as low-rank matrices or sparse tensors, by adapting the projection operator $\mathcal P_s(\cdot)$ to pursue the corresponding structure. Second, although our current theory addresses heavy-tailed problems, the proposed methodology may also be useful for other types of data contamination, such as adversarial corruption. Third, more advanced gradient estimation techniques could further improve the statistical rates; any advancement in multivariate mean estimation can be readily integrated into our framework.

\appendix
\onecolumn
\par
\vspace{2\baselineskip} 
\begin{center}
    \large\bfseries Supplementary Material for ``Two-Stage Robust Sparse Gradient Methods for Regression Under Heavy-Tailed Designs''
\end{center}
\par
\vspace{\baselineskip}

The supplementary material contains technical proofs and auxiliary information for the results in the main article. Appendix~\ref{append:A} compares the robust gradient approach with the robust loss approach. Appendix~\ref{append:B} presents the challenges of data dependence in path analysis of robust gradient approach. Appendix~\ref{append:C} presents the proofs of the deterministic convergence analysis of TS-RIGHT algorithm, as well as auxiliary lemmas. Appendices~\ref{append:D} and \ref{append:E} contain proofs of the stochastic analysis for heavy-tailed linear regression and logistic regression, respectively. Appendix~\ref{sec:appendix_minimax} presents proofs of the lower bound results. Implementation details and additional simulation results are presented in Appendix~\ref{append:G}.


\section{Robust Gradient and Robust Loss Methods}
\label{append:A}

\subsection{Elementwise MoM Estimator Is Not Necessarily the Gradient of a Robust Loss}
\label{subsec:right-not-loss-mom-gradient}
We point out that our MoM gradient estimator is not necessarily the gradient of any robust loss function. In particular, we clarify the distinction between the loss-MOM method of \citet{lecue2020robustclassification} and our gradient-based method since both methods apply the median-of-means operator to achieve robustness.

For a loss $L$, let \(B_1,\ldots,B_K\) be a partition of the sample, and write
\[
    \bar L_k(\bbm\theta)
    =
    \frac{1}{|B_k|}
    \sum_{i\in B_k}
    \ell(\bbm\theta;\bm x_i,y_i),
    \qquad
    \bar{\bm G}_k(\bbm\theta)
    =
    \nabla_{\bbm\theta}\bar L_k(\bbm\theta).
\]
The loss-MOM objective in \citet{lecue2020robustclassification} is
\[
    \widehat R_{\mathrm{LMOM}}(\bbm\theta)
    =
   \operatorname{median}_{1\le k\le K}
    \bigl\{\bar L_k(\bbm\theta)\bigr\},
\]
let
\[
    k_L(\bbm\theta)
    \in
    \operatorname{argmed}_{1\le k\le K}
    \bigl\{\bar L_k(\bbm\theta)\bigr\}
\]
be the median block selected by the scalar loss. The gradient of the loss-MOM objective is
\[
    \nabla_{\bbm\theta}
    \widehat R_{\mathrm{LMOM}}(\bbm\theta)
    =
    \bar{\bm G}_{k_L(\bbm\theta)}(\bbm\theta).
\]
Thus the loss-MOM method takes a gradient step using one block selected by
loss centrality.

Our method takes a different approach. It first computes blockwise gradients
and then applies the MoM operator coordinatewise:
\[
    g_{j}(\bbm\theta;K)
    =
   \operatorname{median}_{1\le k\le K}
    \left\{
    [\bar{\bm G}_k(\bbm\theta)]_j
    \right\},
    \qquad j=1,\ldots,p.
\]
Equivalently, for each coordinate \(j\), there exists a possibly
coordinate-specific median-gradient block \(k_j^G(\bbm\theta)\) that achieves the median; however, these blocks need not be the same across coordinates and they need not be the same as the median-loss block \(k_L(\bbm\theta)\).

Thus the coordinatewise MoM gradient generally does not
correspond to the gradient of the loss-MOM objective, and in high dimensions
it need not correspond to the gradient of any ordinary scalar empirical loss.

\subsection{Why Score-Level Robustification Directly Controls Heavy-Tailed Covariates}
\label{subsec:score-level-robustification-examples}

We now give two examples showing why robustifying the score is different from
robustifying the scalar loss.

\noindent\textbf{Example 1: logistic regression with a heavy-tailed inactive coordinate.}

Let \(\bm x=(U,V)^\top\in \mathbb{R}^2\), \(\bbm\theta^*=(\beta,0)^\top\in \mathbb{R}^2\), and \(\beta\neq 0\),
and assume the correctly specified logistic model
\[
    y\mid (U,V)
    \sim
    \operatorname{Bernoulli}\{\sigma(\beta U)\},
    \qquad
    \sigma(t)=\frac{\exp(t)}{1+\exp(t)}.
\]
Assume that \(V\) is heavy-tailed; for instance, \(V\) only has a finite second moment.

At the true parameter, $\ell(\bbm\theta^*;\mathcal D_n)=n^{-1}\sum_{i=1}^{n}(\Phi(\beta U_i)-y_i\beta U_i)$, where $\Phi(t) = \log(1+e^{t})$. The loss does not depend on \(V\).  Consequently, a loss-MOM method selects the
median-loss block
\[
    k_L(\bbm\theta^*)
    \in
    \operatorname{argmed}_{1\le k\le K}
    \left\{
    \frac{1}{|B_k|}
    \sum_{i\in B_k}
    \bigl[\Phi(\beta U_i)-y_i\beta U_i\bigr]
    \right\}
\]
using only the scalar losses.  The selection is blind to the heavy-tailed
coordinate \(V\).

However, the logistic score at \(\bbm\theta^*\) is
\[
    \nabla_{\bbm\theta}
    \ell(\bbm\theta^*;\mathcal D_n)
    =\frac{1}{n}\sum_{i=1}^{n}
    \{\sigma(\beta U_i)-y_i\}
    \begin{pmatrix}
        U_i\\
        V_i
    \end{pmatrix}.
\]
In particular, the second score coordinate is heavy-tailed due to the heavy-tailed \(V\). The loss-MOM gradient step uses
\[
    [\bar{\bm G}_{k_L(\bbm\theta^*)}(\bbm\theta^*)]_2
    =
    \frac{1}{|B_{k_L}|}
    \sum_{i\in B_{k_L(\bbm\theta^*)}}
    \{\sigma(\beta U_i)-y_i\}V_i.
\]
Thus the step direction can be influenced since a single extreme $V_i$ in the selected block \(B_{k_L(\bbm\theta^*)}\)  can dominate the second gradient
coordinate.

Our MoM gradient estimator robustifies the score coordinate itself:
\[
    g_{2}(\bbm\theta^*;K)
    =
   \operatorname{median}_{1\le k\le K}
    \left\{
    \frac{1}{|B_k|}
    \sum_{i\in B_k}
    \{\sigma(\beta U_i)-y_i\}V_i
    \right\}.
\]
Since \(\mathbb E\{\{\sigma(\beta U)-y\}^2V^2\}\) is finite, the MoM estimator achieves sub-Gaussian-type concentration for this coordinate, even though \(V\) is heavy-tailed.

\noindent\textbf{Example 2: Huber regression clips the residual but not the leverage factor.}

Consider the linear model
\[
    y_i=\bm x_i^\top\bbm\theta^*+\epsilon_i,
\]
and the Huber loss $\ell_H(\bbm\theta;\mathcal D_n)=n^{-1}\sum_{i=1}^{n}\rho_\tau(y_i-\bm x_i^\top\bbm\theta)$,
where
\[
    \rho_\tau(u)
    =
    \begin{cases}
        u^2/2, & |u|\le \tau,\\
        \tau |u|-\tau^2/2, & |u|>\tau,
    \end{cases}
    \qquad
    \psi_\tau(u)=\rho_\tau'(u)
    =
    \operatorname{sign}(u)\min\{|u|,\tau\}.
\]
The Huber score is
\[
    \nabla_{\bbm\theta}
    \ell_H(\bbm\theta;\mathcal D_n)
    =
    -\frac{1}{n}\sum_{i=1}^{n}\psi_\tau(y_i-\bm x_i^\top\bbm\theta)\bm x_i.
\]
At the true parameter,
\[
    \nabla_{\bbm\theta}
    \ell_H(\bbm\theta^*;\mathcal D_n)
    =
    -\frac{1}{n}\sum_{i=1}^{n}\psi_\tau(\epsilon_i)\bm x_i.
\]
Thus Huberization clips the residual factor $|\psi_\tau(\epsilon_i)|\le \tau$, 
but it does not clip the design \(\bm x_i\).  For coordinate \(j\),
\[
    \left|
    [\nabla_{\bbm\theta}
    \ell_H(\bbm\theta^*;\mathcal D_n)]_j
    \right|
    =
    \left|\frac{1}{n}\sum_{i=1}^{n}\psi_\tau(\epsilon_i)x_{ij}\right|
    \le
    \frac{\tau}{n} \sum_{i=1}^{n} |x_{ij}|,
\]
which can still inherit heavy-tail behavior from the design.

Our method addresses this object directly by applying MoM to the score coordinate:
\[
    g_{j}(\bbm\theta^*;K)
    =
   \operatorname{median}_{1\le k\le K}
    \left\{
    -\frac{1}{|B_k|}
    \sum_{i\in B_k}
    \epsilon_i x_{ij}
    \right\}.
\]
This estimator robustifies the mean of the score coordinate $\epsilon_i x_{ij}$ directly.


\section{Challenges of Dependence in Path Analysis}
\label{append:B}

Appendix~\ref{append:A} shows that the TS-RIGHT update is based on a robust gradient vector field, and that this vector field is not necessarily the gradient of a fixed empirical loss. We now explain why this distinction necessitates a pathwise analysis. We first recall the standard proof architecture for a regularized empirical-loss minimizer, and then explain why this architecture breaks down for TS-RIGHT. Finally, we discuss the key data-dependence issue that arises in pathwise analysis.

\subsection{Why Pathwise Analysis Is Needed}
\subsubsection{Why Empirical-Loss Minimizers Admit One-Shot Analysis}
We first recall the standard proof architecture for LASSO. Consider the sparse linear model with the empirical least-squares loss
\[
    \mathcal L_n(\bbm\theta)
    =
    \frac{1}{2n}\|\bm y-\bm X\bbm\theta\|_2^2 .
\]
The LASSO estimator is defined as a global minimizer of a penalized empirical
objective:
\[
    \widehat{\bbm\theta}_{\rm L}
    \in
    \arg\min_{\bbm\theta\in\mathbb R^p}
    \left\{
        \mathcal L_n(\bbm\theta)+\lambda\|\bbm\theta\|_1
    \right\}.
\]
Let \(\bbm\Delta=\widehat{\bbm\theta}_{\rm L}-\bbm\theta^*\). Since
\(\widehat{\bbm\theta}_{\rm L}\) is a global minimizer,
\begin{equation}\label{lasso_basic_inequality}
    \mathcal L_n(\widehat{\bbm\theta}_{\rm L})-\mathcal L_n(\bbm\theta^*)
    \le
    \lambda\|\bbm\theta^*\|_1 -\lambda\|\widehat{\bbm\theta}_{\rm L}\|_1.
\end{equation}
We give an informal proof sketch for the estimation error upper bound.

Under the restricted strong convexity (RSC) condition with parameter $\kappa_1$, the left-hand side is bounded from below by $\kappa_1\|\bbm\Delta\|_2^2$. Using the decomposability of the \(\ell_1\)-norm, the right-hand side is bounded by $\lambda\|\bbm\Delta_{S^*}\|_1-\lambda\|\bbm\Delta_{(S^*)^c}\|_1$, which can be further bounded by $\lambda\|\bbm\Delta_{S^*}\|_1$. By the Cauchy-Schwarz inequality, $\lambda\|\bbm\Delta_{S^*}\|_1\le \sqrt{s^*}\lambda\|\bbm\Delta\|_2$. Thus, dividing both sides by $\|\bbm\Delta\|_2$ yields an upper bound on $\|\bbm\Delta\|_2$.

\subsubsection{What Fails for Robust Gradient Methods}
\label{append:B:what_fails_for_right}

The preceding route relies on the existence of a scalar empirical objective whose
global minimizer is the estimator. This structure is not available for TS-RIGHT.
As discussed in Appendix~\ref{append:A}, the coordinatewise MoM gradient
vector field is generally not the gradient of a scalar empirical loss. That is,
one cannot in general introduce an empirical loss \(\widetilde{\mathcal L}_n\)
such that
\[
    g(\bbm\theta;K,\mathcal A)
    =
    \nabla \widetilde{\mathcal L}_n(\bbm\theta)
    \qquad
    \text{for all } \bbm\theta .
\]
More importantly, the TS-RIGHT output is not defined as
\[
    \widehat{\bbm\theta}_{\rm TS-RIGHT}
    \in
    \arg\min_{\bbm\theta}
    \left\{
        \widetilde{\mathcal L}_n(\bbm\theta)
        +
        \lambda\|\bbm\theta\|_1
    \right\}.
\]
Instead, it is produced by the hard-thresholded robust-gradient recursion
\begin{equation}\label{eq:right_path_update_appendix}
    \bbm\theta^{t+1}
    =
    \mathcal P_s
    \left\{
        \bbm\theta^t-\eta g(\bbm\theta^t;K_{t+1},\mathcal A_{t+1})
    \right\}.
\end{equation}
Consequently, the analogue of the LASSO basic inequality
\eqref{lasso_basic_inequality} stated below
\[
    Q_n(\widehat{\bbm\theta}_{\rm TS-RIGHT})
    \le
    Q_n(\bbm\theta^*),
    \qquad
    Q_n(\bbm\theta)
    =
    \widetilde{\mathcal L}_n(\bbm\theta)
    +
    \lambda\|\bbm\theta\|_1,
\]
is not available.

This also means that the LASSO cone is not obtained. Hard thresholding gives
sparsity of the iterates,
\[
    \|\bbm\theta^t\|_0\le s,
    \qquad
    \|\bbm\theta^*\|_0\le s^*,
    \qquad
    \|\bbm\theta^t-\bbm\theta^*\|_0\le s+s^*,
\]
but it does not imply the optimizer-induced cone condition
\[
    \|(\bbm\theta^T-\bbm\theta^*)_{(S^*)^c}\|_1
    \le
    3\|(\bbm\theta^T-\bbm\theta^*)_{S^*}\|_1 .
\]
Thus the sparse geometry in TS-RIGHT is not the LASSO cone geometry but the active sparse set used in the hard thresholding step.

The analysis must instead control the update map
\eqref{eq:right_path_update_appendix} along the realized sparse path
\(\{\bbm\theta^t\}_{t\ge0}\). This is the entry point for the pathwise analysis
developed below.

\subsection{The Data Dependence Problem}

The key difficulty for pathwise analysis of robust gradient methods
is that the data used to compute the gradient at the current iterate may intersect with the data used to generate the current iterate. We demonstrate this issue by considering the $(t+1)$-th iterate. 

The update is
\[
    \bbm\theta^{t+1}
    =
    \mathcal P_s
    \left\{
        \bbm\theta^t-\eta g(\bbm\theta^t;K_{t+1},\mathcal A_{t+1})
    \right\}.
\]

By induction, $\bbm\theta^t$ is a function of the data in $\mathcal A_1,\ldots,\mathcal A_t$. If $\mathcal A_{t+1}$ intersects with any of $\mathcal A_1,\ldots,\mathcal A_t$, or, in the extreme case, if $\mathcal A_1=\dots=\mathcal A_t=\mathcal A_{t+1}=\mathcal D_n$, then $\bbm\theta^t$ depends on the data $\mathcal A_{t+1}$ used to compute $g(\bbm\theta^t;K_{t+1},\mathcal A_{t+1})$. This data dependence distorts the usual MoM concentration argument for the gradient estimator, which relies on the randomness of $\mathcal A_{t+1}$ to control the estimation error. More specifically,
\[g_j(\bbm\theta^t;K_{t+1},\mathcal A_{t+1})=\operatorname{median}_{1\le k\le K_{t+1}}\left[\frac{1}{|\mathcal B_k(\mathcal A_{t+1})|}\sum_{Z_i\in\mathcal B_k(\mathcal A_{t+1})}\nabla_j\mathcal L(\bbm\theta^t;Z_i)\right].\]

Take the linear regression example for instance. The score coordinate is
\[
    \nabla_j\mathcal L(\bbm\theta^t;Z_i)
    =
    (y_i-\bm x_i^\top\bbm\theta^t)x_{ij}
    =x_{ij}\bm x_i^\top(\bbm\theta^*-\bbm\theta^t)+x_{ij}\epsilon_i.
\]

Typically, the concentration of the MoM estimator relies on the concentration of the block means. However, when \(\bbm\theta^t\) depends on the data \(\mathcal A_{t+1}\), the moments of $|\nabla_j\mathcal L(\bbm\theta^t;Z_i)-G(\bbm \theta^t)|$ cannot be directly controlled due to their complex data dependence. Consequently, the concentration of the block means cannot be established.
Thus, the usual MoM concentration argument does not apply to control the estimation error of \(g_j(\bbm\theta^t;K_{t+1},\mathcal A_{t+1})\) from $G(\bbm\theta^t)$.

\subsection{No-Splitting Strategy Needs Uniform Control}
If one sets $\mathcal A_1=\dots=\mathcal A_t=\mathcal A_{t+1}=\mathcal D_n$, i.e., no sample splitting is used, the aforementioned data dependence issue arises. We demonstrate how the uniform control approach can address this issue, and explain why it leads to suboptimal statistical rates.

The uniform control approach requires concentration of the gradient estimator over the entire sparse parameter space. More specifically, it requires that
$|g(\bbm\theta;K,\mathcal D_n) - G(\bbm\theta)|$ is bounded for all $s$-sparse $\bbm\theta$ with high probability.

Since all iterates are $s$-sparse, the uniform control ensures that $|g(\bbm\theta^t;K,\mathcal D_n)-G(\bbm\theta^t)|$ is controlled for all $t$ with the same high probability, allowing the convergence of the algorithm to be established.

However, this approach requires a much stronger uniform concentration result, as it considers the entire sparse parameter space, which has a significantly larger complexity than the realized sparse path. See Proposition \ref{prop:linear_mom_srs}, where uniform control requires $K\gtrsim \bar s \log p$, while fixed-point control only requires $K\gtrsim \log p$. This spatial complexity ultimately leads to a suboptimal statistical rate. 

\subsection{The Equal Sample Splitting Solution}
We now explain how the equal sample splitting strategy used by \cite{balakrishnan2017statistical, prasad2020robust,liu2019high} provides another way to address the dependence issue, and why equal sample splitting is sample inefficient.

The equal sample splitting design ensures that the data used to compute the gradient at the current iterate is independent of the data used to generate the iterate itself. In particular, $\{\mathcal A_t\}_{t=1}^T$ is a disjoint partition of the data, and each $\mathcal A_t$ is only used to generate the $t$-th iterate $\bbm\theta^t$. Thus, 
for the $(t+1)$-th iterate, the gradient is computed using $\mathcal A_{t+1}$, which is independent of $\mathcal A_1,\ldots,\mathcal A_t$ used to generate $\bbm\theta^t$. This independence enables the use of conditional concentration arguments. Specifically, conditional on $\mathcal A_1, \mathcal A_2,\dots,\mathcal A_t$, the iterate $\bbm\theta^t$ is fixed, and all the randomness in $|\nabla_j\mathcal L(\bbm\theta^t;Z_i)-G(\bbm \theta^t)|$ originates from the randomness of $Z_i \in \mathcal A_{t+1}$. Thus, its moments can be controlled through the assumptions on $\bm x$ and $\epsilon$. The conditional concentration of the block means can then be established, i.e.,
\[
  \mathcal P\left\{\left|\frac{1}{|\mathcal B_k(\mathcal A_{t+1})|}\sum_{Z_i\in\mathcal B_k(\mathcal A_{t+1})}\nabla_j\mathcal L(\bbm\theta^t;Z_i)
    -
    G_j(\bbm\theta^t)\right|\text{ is controlled} \;\middle|\; \mathcal A_1, \mathcal A_2,\dots,\mathcal A_t\right\}
    \ge
    1-\xi.
\]
The tower property ensures that the unconditional probability of this event is also at least \(1-\xi\). This allows the MoM concentration argument to be successfully applied to bound the estimation error of the gradient estimator.

The drawback of equal sample splitting is that the size of each data split, $|\mathcal A_{t+1}|=n/T$, is small, where $T$ is the total number of iterations. This leads to a loss in statistical efficiency because each gradient estimate relies on only a fraction of the available data.


\section{Proofs for the Deterministic Convergence Analysis}
\label{append:C}

In this section, we provide detailed proofs for the algorithmic convergence analysis presented in the main article. We first establish the connections among different forms of the SRS condition, which serves as a foundation for transitioning from fixed-point and uniform SRS to pathwise SRS. Subsequently, we present the step-by-step proofs for the main convergence theorems of our proposed algorithms.

\subsection{Conditional Fixed-Point SRS and Pathwise SRS}
\label{append:srs_to_pathwise}
We establish the connection between uniform SRS, fixed-point SRS, and pathwise SRS in the following lemma.

\begin{lemma}[Routes to pathwise SRS]
\label{lem:srs_to_pathwise}
The pathwise SRS condition can be established through either of the following two routes.

\emph{(1) Uniform-to-pathwise route.}
Let \(\{\bbm\theta^t\}_{t\in\mathcal I}\subset\Theta_s\) be an \(s\)-sparse
sequence of iterates. Suppose \(g(\cdot;K,\mathcal A)\) satisfies uniform
\((\Theta_s,\bar s,\phi,\gamma)\)-SRS. Let \(\mathcal J=\{0\}\), \(\mathcal A_0=\mathcal A\), and \(\pi(t)=0\) for all
\(t\in\mathcal I\). Then the scheduled gradient evaluations satisfy pathwise
\((\mathcal I,\pi,\bar s,\phi,\gamma)\)-SRS. Consequently, if the uniform SRS
event holds with probability at least \(1-\alpha\), then the induced pathwise
SRS event also holds with probability at least \(1-\alpha\).

\emph{(2) Independent-batch fixed-point route.}
Consider Phase II of Algorithm~\ref{alg:two_stage_right}. Let
\[
\mathcal I_2=\{T_1,\ldots,T_1+T_2-1\},
\qquad
\mathcal J_2=\{1,\ldots,T_2\},
\qquad
\pi_2(t)=t-T_1+1 .
\]
Let \(\mathcal F_{\ell-1}\) be the information generated before the batch
\(\mathcal D^{(2)}_\ell\) is used. Suppose
\(\bbm\theta^{T_1+\ell-1}\in\Theta_s\) is \(\mathcal F_{\ell-1}\)-measurable
and \(\mathcal D^{(2)}_\ell\) is independent of \(\mathcal F_{\ell-1}\). If,
for every deterministic \(\bbm\theta\in\Theta_s\), $\mathsf{SRS}(\bbm\theta;K_2,\mathcal D^{(2)}_\ell,\bar s,\phi_2,\gamma_2)$ holds with probability at least \(1-\alpha_\ell\) over
\(\mathcal D^{(2)}_\ell\), then the scheduled Phase-II gradient evaluations
satisfy pathwise \((\mathcal I_2,\pi_2,\bar s,\phi_2,\gamma_2)\)-SRS with
probability at least $1-\sum_{\ell=1}^{T_2}\alpha_\ell$.
\end{lemma}

\begin{proof}
For part \emph{(1)}, the uniform SRS event implies
\(\mathsf{SRS}(\bbm\theta;K,\mathcal A,\bar s,\phi,\gamma)\) for every
\(\bbm\theta\in\Theta_s\). Since every iterate \(\bbm\theta^t\) belongs to
\(\Theta_s\), the same event implies
\(
\mathsf{SRS}(\bbm\theta^t;K,\mathcal A,\bar s,\phi,\gamma)
\)
for every \(t\in\mathcal I\). This is exactly pathwise
\((\mathcal I,\pi,\bar s,\phi,\gamma)\)-SRS with
\(\pi(t)=0\). The probability statement follows because no additional event is
introduced.

For part \emph{(2)}, define
\(
E_\ell
=
\mathsf{SRS}
\bigl(
\bbm\theta^{T_1+\ell-1};
K_2,\mathcal D^{(2)}_\ell,\bar s,\phi_2,\gamma_2
\bigr)
\) for $\ell=1,\ldots,T_2$.
Conditioning on \(\mathcal F_{\ell-1}\), the iterate
\(\bbm\theta^{T_1+\ell-1}\) is fixed, while
\(\mathcal D^{(2)}_\ell\) is independent of \(\mathcal F_{\ell-1}\). Therefore,
by the fixed-point SRS assumption applied to the realized deterministic value of
\(\bbm\theta^{T_1+\ell-1}\),
\[
\mathbb P(E_\ell^c\mid \mathcal F_{\ell-1})\le \alpha_\ell .
\]
Taking expectations gives \(\mathbb P(E_\ell^c)\le \alpha_\ell\). Hence, by the
union bound,
\[
\mathbb P\left(\bigcap_{\ell=1}^{T_2}E_\ell\right)
\ge
1-\sum_{\ell=1}^{T_2}\mathbb P(E_\ell^c)
\ge
1-\sum_{\ell=1}^{T_2}\alpha_\ell .
\]
Since \(E_\ell\) is exactly the SRS event for the scheduled pair
\(t=T_1+\ell-1\) and \(\pi(t)=\ell\), the intersection
\(\bigcap_{\ell=1}^{T_2}E_\ell\) is the desired pathwise SRS event.
\end{proof}

\subsection{Restricted Curvature and SRCG}
\label{append:SRCG}
We compare the SRCG condition with the commonly used restricted strong convexity and restricted strong smoothness (RSC/RSS) conditions in the literature \cite{jain2017non}. SRCG is imposed directly on the population gradient, while RSC/RSS are conditions on the population risk surface. When the population risk exists, RSC/RSS imply SRCG as shown in the lemma below. However, SRCG can hold even when the population risk is infinite.
\begin{lemma}[SRCG Implied by RSC/RSS]
    \label{lem:RSS_RSC_imply_SRCG}
    If the population risk $\mathcal{R}(\bbm{\theta})$ exists and satisfies $\beta$-restricted strong smoothness (RSS) and $\alpha$-restricted strong convexity (RSC), i.e., 
    \begin{equation}
        \frac{\alpha}{2}\|\bbm{\theta}_2-\bbm{\theta}_1\|_2^2 \leq \mathcal{R}(\bbm{\theta}_1) - \mathcal{R}(\bbm{\theta}_2) - \langle\nabla\mathcal{R}(\bbm{\theta}_2),\bbm{\theta}_1-\bbm{\theta}_2\rangle \leq \frac{\beta}{2}\|\bbm{\theta}_1-\bbm{\theta}_2\|_2^2,
    \end{equation}
    over the set of $(2s+s^*)$-sparse vectors $\bbm{\theta}_1$ and $\bbm{\theta}_2$, then the population score map \(G=\nabla \mathcal R\) satisfies the \((s,a,b)\)-SRCG condition with $a=(2\beta)^{-1}$ and $b=\alpha/2$.
\end{lemma}

\begin{proof}[Proof of Lemma \ref{lem:RSS_RSC_imply_SRCG}]
	By the definitions of $\alpha$-RSC and $\beta$-RSS on the set of vectors
with sparsity at most $2s+s^*$, for any $\bbm{\theta}_1$ and $\bbm{\theta}_2$ satisfying $\|\bbm{\theta}_i\|_0\le 2s+s^*$, $i=1,2$, we have
	\begin{equation}\label{eq:RSC}
		\mathcal{R}(\bbm{\theta}_1)-\mathcal{R}(\bbm{\theta}_2)\ge\frac{\alpha}{2}\lVert \bbm{\theta}_1 - \bbm{\theta}_2 \rVert_2^2 +\left\langle\bbm{\theta}_1-\bbm{\theta}_2,\nabla \mathcal{R}(\bbm{\theta}_2)\right\rangle,
	\end{equation}
	and
	\begin{equation}\label{eq:RSS}
		\mathcal{R}(\bbm{\theta}_1)-\mathcal{R}(\bbm{\theta}_2)\le \frac{\beta}{2}\lVert \bbm{\theta}_1 - \bbm{\theta}_2 \rVert_2^2+ \left\langle\bbm{\theta}_1-\bbm{\theta}_2,\nabla \mathcal{R}(\bbm{\theta}_2)\right\rangle,
	\end{equation}
	respectively. 

	Setting $\bbm{\theta}_1=\bbm{\theta}^*$ and $\bbm{\theta}_2=\bbm{\theta}$ in \eqref{eq:RSC}, we have
	\begin{equation}
		\left\langle\nabla \mathcal{R}(\bbm{\theta}),\bbm{\theta}-\bbm{\theta}^*\right\rangle \ge \mathcal{R}(\bbm{\theta})-\mathcal{R}(\bbm{\theta}^*)+\frac{\alpha}{2}\lVert \bbm{\theta} - \bbm{\theta}^* \rVert_2^2.
	\end{equation}
	Similarly, setting $\bbm{\theta}_1=\bbm{\theta}-\beta^{-1}\nabla \mathcal{R}(\bbm{\theta})_S$ and $\bbm{\theta}_2=\bbm{\theta}$ in \eqref{eq:RSS}, we have
	\begin{equation}
		\mathcal R(\bbm{\theta})-\mathcal R(\bbm{\theta}-\beta^{-1}\nabla \mathcal{R}(\bbm{\theta})_S)\ge \frac{1}{2\beta}\Vert\nabla\mathcal R(\bbm{\theta})_S\Vert_2^2.
	\end{equation}
	Noting that $G(\bbm{\theta}^*)=0$, by \eqref{eq:RSC} we know that $\mathcal R(\bbm{\theta}^*)\le \mathcal R(\bbm{\theta})$ for any $\bbm{\theta}$. Therefore,
	\begin{equation}
		\mathcal R(\bbm \theta)-\mathcal R(\bbm{\theta}-\beta^{-1}\nabla \mathcal{R}(\bbm{\theta})_S)\le \mathcal R(\bbm{\theta})-\mathcal R(\bbm{\theta}^*),
	\end{equation}

	we have
	\begin{equation}
		\left\langle\nabla \mathcal{R}(\bbm{\theta}),\bbm{\theta}-\bbm{\theta}^*\right\rangle \ge \frac{1}{2\beta}\Vert\nabla\mathcal R(\bbm \theta)_S\Vert_2^2+\frac{\alpha}{2}\lVert \bbm{\theta} - \bbm{\theta}^* \rVert_2^2,
	\end{equation}
	which implies that the $(s,1/(2\beta),\alpha/2)$-SRCG condition holds.
    
\end{proof}

\subsection{Hard-Thresholding Inequality}
In this subsection, we present an auxiliary lemma used in the following proof. This lemma shows that the hard-thresholding operator is nearly non-expansive when the working sparsity level sufficiently exceeds the true sparsity. The result follows from Lemma 3.3 of \cite{li2016nonconvex}; we restate it below for completeness.

\begin{lemma}[Vector Hard Thresholding Property]\label{supp:lemma:property_of_hard_thresholding}
	Let $\bbm{\theta}^*\in \mathbb R^p$ and $\Vert \bbm{\theta}^* \Vert_0 \leq k^*$. For any $k > k^*$ and any vector $\bbm{\theta} \in \mathbb{R}^p$, the following inequality holds:
	\begin{equation}
		\lVert P_k(\bbm{\theta}) - \bbm{\theta}^* \rVert_2^2 \leq \left( 1 + \frac{2\sqrt{k^*}}{\sqrt{k - k^*}} \right) \lVert \bbm{\theta} - \bbm{\theta}^* \rVert _2^2.
	\end{equation}
\end{lemma}

\subsection{Proof of Lemma~\ref{lem:one_stage_contraction_pathwise_srs}}
\begin{proof}
The proof is deterministic on the event that the pathwise SRS condition holds. For \(t=t_0,\ldots,t_0+T-1\), write $g_t(\bbm\theta^t)=g(\bbm\theta^t;K,\mathcal A_{\pi(t)})$. Let
 $e_t=g_t(\bbm\theta^t)-G(\bbm\theta^t)$, and $D_t=\|\bbm\theta^t-\bbm\theta^*\|_2$.

Let $S^t=\operatorname{supp}(\bbm\theta^t)$, $S^*=\operatorname{supp}(\bbm\theta^*)$, and $\widetilde S^t=S^t\cup S^{t+1}\cup S^* $.
Since \(\bbm\theta^{t_0}\) is \(s\)-sparse and the update is hard-thresholded at sparsity level \(s\), all iterates are \(s\)-sparse. Hence
\[
|\widetilde S^t|\le 2s+s^*=\bar s .
\]
Moreover, \(\operatorname{supp}(\bbm\theta^t)\cup\operatorname{supp}(\bbm\theta^*)\subseteq \widetilde S^t\), so the SRCG condition can be applied with \(S=\widetilde S^t\).

Define $u^t=\bbm\theta^t-\eta g_t(\bbm\theta^t)$. By definition, \(\widetilde S^t\) contains the support selected by \(\mathcal P_s(u^t)\). Therefore, zeroing out coordinates outside \(\widetilde S^t\) does not change the hard-thresholded vector, thus $\bbm\theta^{t+1}=\mathcal P_s(u^t_{\widetilde S^t})$. 
By the hard-thresholding inequality in Lemma~\ref{supp:lemma:property_of_hard_thresholding}, applied with \(k=s\) and \(k^*=s^*\),
\begin{align}
\|\bbm\theta^{t+1}-\bbm\theta^*\|_2
&\le
\left(1+\frac{2\sqrt{s^*}}{\sqrt{s-s^*}}\right)^{1/2}
\left\|
\bbm\theta^t_{\widetilde S^t}
-\eta g_t(\bbm\theta^t)_{\widetilde S^t}
-\bbm\theta^*_{\widetilde S^t}
\right\|_2 .
\label{eq:supp_one_stage_ht}
\end{align}
For simplicity, let $q_s={2\sqrt{s^*}}/{\sqrt{s-s^*}}$.

We next control the norm on the right-hand side of \eqref{eq:supp_one_stage_ht}. Since $g_t(\bbm\theta^t)=G(\bbm\theta^t)+e_t $,

the triangle inequality gives
\begin{equation}
\left\| \bbm\theta^t_{\widetilde S^t} -\eta g_t(\bbm\theta^t)_{\widetilde S^t} -\bbm\theta^*_{\widetilde S^t} \right\|_2 \le \left\| (\bbm\theta^t-\bbm\theta^*)_{\widetilde S^t} -\eta G(\bbm\theta^t)_{\widetilde S^t} \right\|_2 + \eta\|e_{t,\widetilde S^t}\|_2 .
\label{eq:supp_one_stage_split}
\end{equation}

For the first term in \eqref{eq:supp_one_stage_split}, using
\(\operatorname{supp}(\bbm\theta^t-\bbm\theta^*)\subseteq \widetilde S^t\), we have
\begin{equation}
\left\|
(\bbm\theta^t-\bbm\theta^*)_{\widetilde S^t}
-\eta G(\bbm\theta^t)_{\widetilde S^t}
\right\|_2^2
\nonumber\\
 =
D_t^2
+
\eta^2\|G(\bbm\theta^t)_{\widetilde S^t}\|_2^2
-
2\eta
\left\langle
G(\bbm\theta^t),
\bbm\theta^t-\bbm\theta^*
\right\rangle .
\end{equation}
By the \((s,a,b)\)-SRCG condition,
\[
\left\langle
G(\bbm\theta^t),
\bbm\theta^t-\bbm\theta^*
\right\rangle
\ge
a\|G(\bbm\theta^t)_{\widetilde S^t}\|_2^2
+
bD_t^2 .
\]
Therefore,
\begin{equation}
\left\|
(\bbm\theta^t-\bbm\theta^*)_{\widetilde S^t}
-\eta G(\bbm\theta^t)_{\widetilde S^t}
\right\|_2^2
\nonumber\\
\le
(1-2\eta b)D_t^2
+
(\eta^2-2\eta a)
\|G(\bbm\theta^t)_{\widetilde S^t}\|_2^2 .
\end{equation}
Since \(\eta=2a\eta_0\) and \(\eta_0\in(0,1]\),
\[
\eta^2-2\eta a
=
4a^2\eta_0(\eta_0-1)
\le 0 .
\]
Let $\varphi=\sqrt{1-4ab\eta_0}$.
Then \(0\le\varphi<1\), and the preceding display implies
\begin{equation}
\left\|
(\bbm\theta^t-\bbm\theta^*)_{\widetilde S^t}
-\eta G(\bbm\theta^t)_{\widetilde S^t}
\right\|_2
\le
\varphi D_t .
\label{eq:supp_one_stage_population_step}
\end{equation}

For the second term in \eqref{eq:supp_one_stage_split}, pathwise \((\{t_0,\ldots,t_0+T-1\},\pi,\bar s,\phi,\gamma)\)-SRS applies to the random set \(\widetilde S^t\), because \(|\widetilde S^t|\le \bar s\). Thus
\begin{equation}
\|e_{t,\widetilde S^t}\|_2
\le
\phi D_t+\gamma .
\label{eq:supp_one_stage_srs_step}
\end{equation}
Combining \eqref{eq:supp_one_stage_ht}, \eqref{eq:supp_one_stage_split}, \eqref{eq:supp_one_stage_population_step}, and \eqref{eq:supp_one_stage_srs_step}, we obtain the one-step recursion
\begin{equation}
D_{t+1}
\le
\sqrt{1+q_s}\,(\varphi+\eta\phi)D_t
+
\sqrt{1+q_s}\,\eta\gamma .
\label{eq:supp_one_stage_basic_recursion}
\end{equation}

It remains to choose constants independent of \(s\), \(s^*\), \(p\), and \(T\). For \(c>1\), define $A(c)=\left(1+{2}/{\sqrt{c-1}}\right)^{1/2}$.

If \(s\ge c_s s^*\), then \(q_s\le 2/\sqrt{c_s-1}\), and hence \(\sqrt{1+q_s}\le A(c_s)\). Fix any \(\rho\in(\varphi,1)\). Since \(A(c)\downarrow1\) as \(c\to\infty\), we can choose \(c_s>1\) large enough so that
\[
A(c_s)\varphi<\rho .
\]
Then define
\[
\bar\phi
=
\frac{\rho-A(c_s)\varphi}{\eta A(c_s)}
\qquad\text{and}\qquad
C_\gamma
=
\frac{\eta A(c_s)}{1-\rho}.
\]
These constants depend only on \(a\), \(b\), and \(\eta_0\). If \(s\ge c_s s^*\) and \(\phi\le\bar\phi\), then \eqref{eq:supp_one_stage_basic_recursion} yields
\[
D_{t+1}
\le
\rho D_t
+
\eta A(c_s)\gamma .
\]
Iterating this recursion from \(t=t_0\) to \(t=t_0+T-1\) gives
\[
D_{t_0+T}
\le
\rho^T D_{t_0}
+
\eta A(c_s)\sum_{j=0}^{T-1}\rho^j\gamma
\le
\rho^T D_{t_0}
+
C_\gamma\gamma .
\]
Equivalently,
\[
\|\bbm\theta^{t_0+T}-\bbm\theta^*\|_2
\le
\rho^T
\|\bbm\theta^{t_0}-\bbm\theta^*\|_2
+
C_\gamma\gamma ,
\]
which proves the lemma.
\end{proof}

\subsection{Proof of Theorem~\ref{thm:two_stage_right_convergence}}
We now prove the two-stage convergence result for the proposed algorithm, which applies the one-stage contraction lemma twice, first to the localization phase and then to the refinement phase.
\begin{proof}
We work on the event where the stated Phase I and Phase II SRS conditions hold.

First consider Phase I. 
By the uniform \((\Theta_s,\bar s,\phi_1,\gamma_1)\)-SRS condition on \(\mathcal D^{(1)}\), and because the Phase I iterates belong to \(\Theta_s\), we have, for every \(t=0,\ldots,T_1-1\) and every \(S\subset[p]\) with \(|S|\le \bar s\),
\[
\|[g(\bbm\theta^t;K,\mathcal D^{(1)})-G(\bbm\theta^t)]_S\|_2
\le
\phi_1\|\bbm\theta^t-\bbm\theta^*\|_2+\gamma_1 .
\]
Thus the Phase I sequence satisfies pathwise \((\mathcal I_1,\pi_1,\bar s,\phi_1,\gamma_1)\)-SRS for \(\mathcal I_1=\{0,\ldots,T_1-1\}\), \(\pi_1(t)=0\) and $\mathcal A_0=\mathcal D^{(1)}$.

 Since \(\phi_1\le\bar\phi\) and \(s\ge c_s s^*\), Lemma~\ref{lem:one_stage_contraction_pathwise_srs}, applied with \(t_0=0\), \(T=T_1\), \(\phi=\phi_1\), and \(\gamma=\gamma_1\), gives
\[
\|\bbm\theta^{T_1}-\bbm\theta^*\|_2
\le
\rho^{T_1}\|\bbm\theta^0-\bbm\theta^*\|_2
+
C_\gamma\gamma_1 .
\]
This proves the first inequality in the theorem.

Next consider Phase II. By assumption, the realized Phase II refinement path satisfies pathwise \((\mathcal I_2,\pi_2,\bar s,\phi_2,\gamma_2)\)-SRS. Since \(\phi_2\le\bar\phi\) and the Phase II recursion is initialized at the \(s\)-sparse vector \(\bbm\theta^{T_1}\), Lemma~\ref{lem:one_stage_contraction_pathwise_srs}, applied with \(t_0=T_1\), \(T=T_2\), \(\phi=\phi_2\), and \(\gamma=\gamma_2\), gives
\[
\|\bbm\theta^{T_1+T_2}-\bbm\theta^*\|_2
\le
\rho^{T_2}\|\bbm\theta^{T_1}-\bbm\theta^*\|_2
+
C_\gamma\gamma_2 .
\]
Since \(\widehat{\bbm\theta}=\bbm\theta^{T_1+T_2}\), this proves the second inequality in the theorem.

Substituting the Phase I bound into the Phase II bound yields
\begin{align*}
\|\widehat{\bbm\theta}-\bbm\theta^*\|_2
&\le
\rho^{T_2}
\left\{
\rho^{T_1}\|\bbm\theta^0-\bbm\theta^*\|_2
+
C_\gamma\gamma_1
\right\}
+
C_\gamma\gamma_2 \\
&=
\rho^{T_1+T_2}\|\bbm\theta^0-\bbm\theta^*\|_2
+
C_\gamma\rho^{T_2}\gamma_1
+
C_\gamma\gamma_2 ,
\end{align*}
which proves \eqref{eq:two_stage_combined_bound}.
\end{proof}


\section{Heavy-Tailed Linear Regression}
\label{append:D}

This section verifies the linear-regression results used in the main text. 
We first prove the SRCG condition for the population linear score. 
We then derive moment bounds for the design-driven and noise-driven components of the score. 
These moment bounds are combined with median-of-means concentration to establish fixed-point and uniform SRS for the robust gradient estimator. 
Finally, we combine the SRS bounds with the abstract two-stage convergence theorem to prove Theorem~\ref{thm:linear_two_stage_right}.

\subsection{Proof of SRCG for Linear Regression}
\begin{proof}[Proof of Proposition \ref{prop:linear_srcg}]
	Let $\bbm\Delta=\bbm\theta-\bbm\theta^*$ and
$G(\bbm\theta)=\mathbb E[\nabla_{\bbm\theta}\mathcal L(\bbm\theta;\bm x_i,y_i)]
=\bbm\Sigma_{\bm x}\bbm\Delta$.
For any admissible set $S$ with
$\operatorname{supp}(\bbm\Delta)\subseteq S$ and $|S|\le 2s+s^*$,
\[
\|G(\bbm\theta)_S\|_2
=
\sup_{\|\bm u\|_2=1,\operatorname{supp}(\bm u)\subseteq S}
\bm u^\top \bbm\Sigma_{\bm x}\bbm\Delta .
\]
Since both $\bm u$ and $\bbm\Delta$ are supported on $S$, by the Cauchy--Schwarz inequality in the $\bbm \Sigma_{\bm x}$ inner product and the restricted upper eigenvalue condition, we have
\[
\bm u^\top \bbm\Sigma_{\bm x}\bbm\Delta
\le
\sqrt{\bm u^\top\bbm \Sigma_{\bm x} \bm u}\sqrt{\bbm\Delta^\top\bbm \Sigma_{\bm x}\bbm\Delta}
\le
\kappa_+\|\bbm\Delta\|_2 .
\]
Moreover, the restricted lower eigenvalue condition gives
\[
\bbm\Delta^\top\bbm\Sigma_{\bm x}\bbm\Delta\ge \kappa_-\|\bbm\Delta\|_2^2 .
\]
Therefore,
\[
\langle G(\bbm\theta),\bbm\Delta\rangle
=
\bbm\Delta^\top\bbm\Sigma_{\bm x}\bbm\Delta
\ge
\frac{\kappa_-}{2\kappa_+^2}\|G(\bbm\theta)_S\|_2^2
+
\frac{\kappa_-}{2}\|\bbm\Delta\|_2^2 .
\]
Thus the SRCG condition holds with
$a=\kappa_-/(2\kappa_+^2)$ and $b=\kappa_-/2$.
\end{proof}

\subsection{Moment Bounds for Linear Regression Score Components}
The following lemma establishes that, under Assumption \ref{assump:linear_moments}, the coordinates of the two constituent terms of the gradient $\nabla_{\bbm \theta} \mathcal{L}(\bbm\theta;Z_i)$ for linear regression admit finite moments.
\begin{lemma}\label{supp:lemma:moment_result_for_linear_regression}
	Under Assumption \ref{assump:linear_moments}, for any fixed $\bbm \theta\in \mathbb{R}^p$ with $\Vert\bbm \theta\Vert_0\le s$, and for any $j\in \{1,2,\ldots,p\}$, the $j$-th elements of the two components of the gradient $\nabla_{\bbm \theta} \mathcal{L}(\bbm\theta;Z_i)$ have $(1+\lambda)$-th and $(1+\delta)$-th central moments respectively, i.e., 
 
	\begin{itemize}
		\item[(i)] $\mathbb{E} \left[ \left| \left( x_{ij}\bm{x}_i^{\top}-\bbm\Sigma_{{\bm x}j \cdot} \right) \bbm{\Delta} \right|^{1+\lambda} \right]\le M_{\bm{x},2+2\lambda,\bar s}\Vert \bbm{\Delta} \Vert_{2}^{1+\lambda},$
		\item[(ii)] $\mathbb{E}\left[|x_{ij}\epsilon_i|^{1+\delta}\right]\le M_{\epsilon,1+\delta}M_{\bm{x},1+\delta,\bar s}.$
	\end{itemize}
\end{lemma}

\begin{proof}[Proof of Lemma \ref{supp:lemma:moment_result_for_linear_regression}]
\begin{align*} \mathbb{E} \left[ \left| \left( x_{ij}\bm{x}_i^{\top}-\bbm\Sigma_{{\bm x}j\cdot} \right) \bbm{\Delta} \right|^{1+\lambda} \right]
	&\lesssim  \mathbb{E} \left[ |x_{ij}\bm{x}_i^{\top}\bbm{\Delta} |^{1+\lambda} \right] + |\bbm\Sigma_{{\bm x}j\cdot} \bbm{\Delta}|^{1+\lambda} \\ 
	&=  \mathbb{E} \left[ |x_{ij}\bm{x}_i^{\top}\bbm{\Delta} |^{1+\lambda} \right] + |\mathbb{E} [x_{ij}\bm{x}_i^{\top}\bbm{\Delta}]|^{1+\lambda} \\
	&\overset{\text{Jensen}}{\lesssim} \mathbb{E} \left[ |x_{ij}\bm{x}_i^{\top}\bbm{\Delta} |^{1+\lambda} \right] \\ 
	 &\overset{\text{Cauchy}}{\le} \left( \mathbb{E} \left[ |x_{ij}|^{2+2\lambda} \right] \right) ^{1/2}\left( \mathbb{E} \left[ |\bm{x}_i^{\top}\bbm{\Delta} |^{2+2\lambda} \right] \right) ^{1/2} \\
	  &= \left( \mathbb{E} \left[ |x_{ij}|^{2+2\lambda} \right] \right) ^{1/2}\left( \mathbb{E} \left[ \left| \bm{x}_i^{\top}\frac{\bbm{\Delta}}{\Vert \bbm{\Delta} \Vert _2} \right|^{2+2\lambda} \right] \cdot \Vert \bbm{\Delta} \Vert _{2}^{2+2\lambda} \right) ^{1/2} \\ 
	  &\le \left( M_{\bm{x},2+2\lambda,\bar s} \right)^{1/2} \left( M_{\bm{x},2+2\lambda,\bar s} \cdot \Vert \bbm{\Delta} \Vert_{2}^{2+2\lambda} \right)^{1/2} \tag{by Assumption} \\ 
	  &= M_{\bm{x},2+2\lambda,\bar s}\Vert \bbm{\Delta} \Vert_{2}^{1+\lambda}. 
\end{align*}
For the second part, we have
\begin{align*}
	\mathbb{E}\left[|x_{ij}\epsilon_i|^{1+\delta}\right] &= \mathbb{E}\left[|x_{ij}|^{1+\delta}\mathbb{E}\left[|\epsilon_i|^{1+\delta}|\bm{x}_i\right]\right] \\
	&\leq M_{\bm{x},1+\delta,\bar s} M_{\epsilon,1+\delta}.
\end{align*}
\end{proof}

\subsection{Fixed-Point and Uniform SRS for MoM Gradients}
\begin{proof}[Proof of Proposition~\ref{prop:linear_mom_srs}]
Let \(\bbm\Delta=\bbm\theta-\bbm\theta^*\), and write \(B=m/K\) for the block size. For notational simplicity we assume equal block sizes; nearly equal block sizes only change constants. For a block \(\mathcal B_k(\mathcal A)\), the block mean is
\[
\bar g_k(\bbm\theta;\mathcal A)
=
B^{-1}\sum_{Z_i\in\mathcal B_k(\mathcal A)}
\nabla_{\bbm\theta}\mathcal L(\bbm\theta;Z_i).
\]
Since \(y_i=\bm x_i^\top\bbm\theta^*+\epsilon_i\) and
\(\mathbb E(\epsilon_i\mid \bm x_i)=0\),
\[
\bar g_k(\bbm\theta;\mathcal A)-G(\bbm\theta)
=
- B^{-1}\sum_{Z_i\in\mathcal B_k(\mathcal A)}\epsilon_i\bm x_i
+
B^{-1}\sum_{Z_i\in\mathcal B_k(\mathcal A)}
(\bm x_i\bm x_i^\top-\bm\Sigma_{\bm x})\bbm\Delta .
\]
Set
\[
\phi_0
=
C M_{\bm x,2+2\lambda,\bar s}^{1/(1+\lambda)}
\left(\frac{K}{m}\right)^{\lambda/(1+\lambda)},
\qquad
\gamma_0
=
C M_{\mathrm{eff},1+\delta}^{1/(1+\delta)}
\left(\frac{K}{m}\right)^{\delta/(1+\delta)},
\]
where \(C>0\) is sufficiently large. By Lemma~\ref{supp:lemma:moment_result_for_linear_regression} and Lemma~\ref{supp:concentration_inequality_for_sample_mean_under_heavy_tailed_distribution}, applied also to the negated variables for two-sided control, for every \(k\), \(j\), and every fixed sparse direction \(\bm v\) with \(\|\bm v\|_2=1\) and \(\|\bm v\|_0\le s+s^*\),
\[
\mathbb P\left(
\left|
B^{-1}\sum_{Z_i\in\mathcal B_k(\mathcal A)}
\epsilon_i x_{ij}
\right|>\gamma_0
\right)
\le \frac1{16},
\]
and
\begin{equation}\label{eq:supp_linear_design_term_bound}
\mathbb P\left(
\left|
B^{-1}\sum_{Z_i\in\mathcal B_k(\mathcal A)}
\{x_{ij}\bm x_i^\top\bm v-\mathbb E(x_{ij}\bm x_i^\top\bm v)\}
\right|>\phi_0
\right)
\le \frac1{16}.
\end{equation}~

\noindent\textbf{We first prove fixed SRS.}

Fix a deterministic \(\bbm\theta\in\Theta_s\), and let \(r=\|\bbm\Delta\|_2\). If \(r>0\), write \(\bbm\Delta=r\bm v\), where
\(\|\bm v\|_2=1\) and \(\|\bm v\|_0\le s+s^*\); if \(r=0\), the design-driven term below is zero. The preceding block bounds imply
\[
\mathbb P\left(
|\bar g_{k,j}(\bbm\theta;\mathcal A)-G_j(\bbm\theta)|
>
\gamma_0+r\phi_0
\right)
\le \frac18 .
\]
For fixed \(j\), define
\[
I_{k,j}
=
\mathbf 1\left\{
|\bar g_{k,j}(\bbm\theta;\mathcal A)-G_j(\bbm\theta)|
>
\gamma_0+r\phi_0
\right\}.
\]
The variables \(I_{k,j}\) are independent over \(k\) and satisfy
\(\mathbb E I_{k,j}\le1/8\). If the coordinate-wise median violates
\[
|g_j(\bbm\theta;K,\mathcal A)-G_j(\bbm\theta)|
\le
\gamma_0+r\phi_0,
\]
then at least \(K/2\) block means violate the same bound. Therefore Hoeffding's inequality gives
\[
\mathbb P\left(
|g_j(\bbm\theta;K,\mathcal A)-G_j(\bbm\theta)|
>
\gamma_0+r\phi_0
\right)
\le
\exp(-cK).
\]
Taking a union bound over \(j\in[p]\), with probability at least
\(1-p\exp(-cK)\),
\[
\max_{1\le j\le p}
|g_j(\bbm\theta;K,\mathcal A)-G_j(\bbm\theta)|
\le
\gamma_0+\phi_0\|\bbm\theta-\bbm\theta^*\|_2 .
\]
Hence, for every \(S\subset[p]\) with \(|S|\le\bar s\),
\[
\|[g(\bbm\theta;K,\mathcal A)-G(\bbm\theta)]_S\|_2
\le
\sqrt{\bar s}\,
\{\gamma_0+\phi_0\|\bbm\theta-\bbm\theta^*\|_2\}.
\]
After absorbing constants into \(\phi_{\mathrm{lin}}(m,K)\) and
\(\gamma_{\mathrm{lin}}(m,K)\), this proves the fixed-point SRS condition
with parameters
\[
(\bar s,\phi_{\mathrm{lin}}(m,K),\gamma_{\mathrm{lin}}(m,K)).
\]
If \(K\ge C\log p\), then \(p\exp(-cK)\le Cp^{-c}\).

\noindent\textbf{Next we prove uniform SRS.}

We now prove simultaneous control over \(\Theta_s\). The block-level probability bounds used above remain valid. First define the noise event
\[
\mathcal E_{\mathrm{noise}}
=
\bigcap_{j=1}^p
\left\{
\frac1K\sum_{k=1}^K
\mathbf 1\left[
\left|
B^{-1}\sum_{Z_i\in\mathcal B_k(\mathcal A)}
\epsilon_i x_{ij}
\right|>\gamma_0
\right]
\le \frac16
\right\}.
\]
By Hoeffding's inequality and a union bound over \(j\),
\[
\mathbb P(\mathcal E_{\mathrm{noise}}^c)
\le p\exp(-cK).
\]

For the design-driven term, since the support of $\bbm \Delta$ may vary with $\bbm \theta$, we need to control the supremum over all sparse directions. We consider a fixed support set $U$ first and then take a union bound over all support sets.

Define
\[
A_{k,j}
=
B^{-1}\sum_{Z_i\in\mathcal B_k(\mathcal A)}
\{x_{ij}\bm x_i-\mathbb E(x_{ij}\bm x_i)\}\in\mathbb R^p .
\]
For each support set \(U\subset[p]\), let
$
\mathbb S_U
=
\{\bm v\in\mathbb R^p:\operatorname{supp}(\bm v)\subseteq U,\ \|\bm v\|_2=1\}$. For fixed \(j\) and \(U\) with \(|U|\le s+s^*\), set \(W_k=(A_{k,j})_U\). Take a supremum over \(\bm v\in\mathbb S_U\) in \eqref{eq:supp_linear_design_term_bound},
\[
\sup_{\bm v\in\mathbb S_U}
\mathbb P\{|A_{k,j}^\top\bm v|>\phi_0\}
\le \frac1{16}.
\]
Equivalently, for the class
\[
\mathcal A_U
=
\Big\{
\{w\in\mathbb R^{|U|}: |w^\top u|>\phi_0\}
:\|u\|_2=1
\Big\},
\]
we have \(\sup_{A\in\mathcal A_U}\nu(A)\le1/16\), where
\(\nu(A)=\mathbb P(W_1\in A)\). Applying Lemma~\ref{supp:lemma:vc_deviation_binary_class} and then Lemma~\ref{supp:lemma:shatter_bound_bad_direction_class},
\[
\mathbb P\left(
\sup_{\bm v\in\mathbb S_U}
\frac1K\sum_{k=1}^K
\mathbf 1\{|A_{k,j}^\top\bm v|>\phi_0\}
>
\frac16
\right)
\le
8\exp\left\{
C|U|\log\left(\frac{e(K\vee |U|)}{|U|}\right)-cK
\right\}.
\]
Taking a union bound over \(j\in[p]\) and all \(U\subset[p]\) with
\(|U|\le s+s^*\), and bounding the resulting sparsity factors by \(\bar s=2s+s^*\) since  $t\mapsto t\log\!\big(e(K\vee \bar s)/t \big)$ is increasing for \(1\le t \le \bar s\), we have
\[
\mathbb P(\mathcal E_{\mathrm{sig}}^c)
\le
8\exp\left\{
\log p
+
C\bar s\log\left(\frac{ep}{\bar s}\right)
+
C\bar s\log\left(\frac{e(K\vee \bar s)}{\bar s}\right)
-
cK
\right\},
\]
where
\[
\mathcal E_{\mathrm{sig}}
:=
\bigcap_{j=1}^p
\bigcap_{\substack{U\subset[p]\\ |U|\le s+s^*}}
\left\{
\sup_{\bm v\in\mathbb S_U}
\frac1K\sum_{k=1}^K
\mathbf 1\{|A_{k,j}^\top \bm v|>\phi_0\}
\le \frac{1}{6}
\right\}.
\]

On \(\mathcal E_{\mathrm{noise}}\cap\mathcal E_{\mathrm{sig}}\), consider any
\(\bbm\theta\in\Theta_s\) and any coordinate \(j\), at most \(K/6\) blocks violate the noise bound and at most \(K/6\) blocks violate the signal bound. Thus more than \(K/2\) blocks satisfy both bounds. For those blocks,
\[
|\bar g_{k,j}(\bbm\theta;\mathcal A)-G_j(\bbm\theta)|
\le
\gamma_0+r\phi_0 .
\]
Since \(g_j(\bbm\theta;K,\mathcal A)\) is the median of the \(K\) block means, the following bound
\[
|g_j(\bbm\theta;K,\mathcal A)-G_j(\bbm\theta)|
\le
\gamma_0+\phi_0\|\bbm\theta-\bbm\theta^*\|_2
\]
holds simultaneously for all \(j\in[p]\) and all \(\bbm\theta\in\Theta_s\). Therefore, for every
\(S\subset[p]\) with \(|S|\le\bar s\),
\[
\|[g(\bbm\theta;K,\mathcal A)-G(\bbm\theta)]_S\|_2
\le
\sqrt{\bar s}\,
\{\gamma_0+\phi_0\|\bbm\theta-\bbm\theta^*\|_2\}
\le
\gamma_{\mathrm{lin}}(m,K)
+
\phi_{\mathrm{lin}}(m,K)\|\bbm\theta-\bbm\theta^*\|_2 .
\]
Combining the two event probabilities,
\[
\mathbb P(\mathcal E_{\mathrm{noise}}^c\cup\mathcal E_{\mathrm{sig}}^c)
\le
p\exp(-cK)
+
8\exp\left\{
\log p
+
C\bar s\log\left(\frac{ep}{\bar s}\right)
+
C\bar s\log\left(\frac{e(K\vee \bar s)}{\bar s}\right)
-
cK
\right\}.
\]
If \(K\ge C\bar s\log p\), after increasing \(C\) and adjusting constants, the last display is bounded by
\(C\exp(-c\bar s\log p)\). Hence uniform
\((\Theta_s,\bar s,\phi_{\mathrm{lin}}(m,K),\gamma_{\mathrm{lin}}(m,K))\)-SRS holds with probability at least
\(1-C\exp(-c\bar s\log p)\).
\end{proof}

\subsection{Proof of Theorem \ref{thm:linear_two_stage_right}}
\begin{proof}
Throughout the proof, constants \(C,c>0\) may change from line to line. By Proposition~\ref{prop:linear_srcg}, the population score \(G(\bbm\theta)=\bm\Sigma_{\bm x}(\bbm\theta-\bbm\theta^*)\) satisfies the \((s,a,b)\)-SRCG condition with \(a=\kappa_-/(2\kappa_+^2)\) and \(b=\kappa_-/2\). Let \(\bar\phi,\rho,C_\gamma\) and \(c_s\) be the constants in Lemma~\ref{lem:one_stage_contraction_pathwise_srs}. We take \(s\ge c_s s^*\). We also assume that \(\bbm\theta^0\in\Theta_s=\{\bbm \theta\in \mathbb R^p:\Vert\bbm\theta\Vert_0\le s\}\); equivalently, one may replace the initialization by \(\mathcal P_s(\bbm\theta^0)\) before the first gradient step. Since every update is hard-thresholded at level \(s\), all subsequent iterates are in \(\Theta_s\).

\noindent\textbf{Phase I.}
Set
\[
\phi_1=\phi_{\mathrm{lin}}(n_1,K_1),
\qquad
\gamma_1=\gamma_{\mathrm{lin}}(n_1,K_1).
\]
Since \(K_1\asymp \bar s\log p\), \(n_1\asymp n\), and the sample-size condition implies \(K_1\le c n_1\) after enlarging constants, Proposition~\ref{prop:linear_mom_srs} gives a uniform \((\Theta_s,\bar s,\phi_1,\gamma_1)\)-SRS event on \(\mathcal D^{(1)}\) with probability at least \(1-C\exp(-c\bar s\log p)\). By Lemma~\ref{lem:srs_to_pathwise}, part (1), this uniform event implies pathwise \((\mathcal I_1,\pi_1,\bar s,\phi_1,\gamma_1)\)-SRS for all Phase I iterates, i.e., $\mathcal I_1=\{0,\ldots,T_1-1\}$, $\pi_1(t)=0$ and $\mathcal A_0=\mathcal D^{(1)}$.

The same choice \(K_1\asymp\bar s\log p\) gives
\[
\phi_1
\lesssim_{\lambda}
\sqrt{\bar s}\,
M_{\bm x,2+2\lambda,\bar s}^{1/(1+\lambda)}
\left(\frac{\bar s\log p}{n}\right)^{\lambda/(1+\lambda)} .
\]
Thus \(\phi_1\le\bar\phi\) is ensured by
\[
n\gtrsim_{\kappa,\lambda,\eta_0}
M_{\bm x,2+2\lambda,\bar s}^{1/\lambda}
\bar s^{1+(1+\lambda)/(2\lambda)}
\log p ,
\]
which is exactly the sample-size condition in the theorem up to constants. Applying Lemma~\ref{lem:one_stage_contraction_pathwise_srs} to Phase I gives
\[
\|\bbm\theta^{T_1}-\bbm\theta^*\|_2
\le
\rho^{T_1}\|\bbm\theta^0-\bbm\theta^*\|_2
+
C_\gamma\gamma_1 .
\]
Because \(\|\bbm\theta^0-\bbm\theta^*\|_2\le R_0\) and \(R_0\) is at most polynomial in \(n\), the implicit constant in \(T_1\asymp\log n\) can be chosen sufficiently large so that \(\rho^{T_1}R_0\le C\gamma_1\). 
Recall that $C_\gamma=(\eta A(c_s))/(1-\rho)$, where $\eta$ depends on $\eta_0$, and $A(c_s)$ depends on $\kappa$; therefore
\[
\|\bbm\theta^{T_1}-\bbm\theta^*\|_2
\lesssim_{\kappa,\eta_0}
\gamma_1
\lesssim_{\kappa,\lambda,\delta,\eta_0}
\sqrt{\bar s}\,
M_{\mathrm{eff},1+\delta}^{1/(1+\delta)}
\left(\frac{\bar s\log p}{n}\right)^{\delta/(1+\delta)} ,
\]
which proves \eqref{eq:linear_phase_one_rate}.

\noindent\textbf{Phase II.}
Let \(m_2=n_2/T_2\) denote the batch size in Phase II; unequal batch sizes only change constants if \(m_2\) is replaced by the minimum batch size. Set
\[
\phi_2=\phi_{\mathrm{lin}}(m_2,K_2),
\qquad
\gamma_2=\gamma_{\mathrm{lin}}(m_2,K_2).
\]
For \(\ell=1,\ldots,T_2\), let \(\mathcal F_{\ell-1}\) be the sigma-field generated by \(\mathcal D^{(1)}\) and the first \(\ell-1\) Phase II batches. Then \(\bbm\theta^{T_1+\ell-1}\) is \(\mathcal F_{\ell-1}\)-measurable, while \(\mathcal D^{(2)}_\ell\) is independent of \(\mathcal F_{\ell-1}\). Conditional on \(\mathcal F_{\ell-1}\), the current iterate is therefore fixed relative to the fresh batch \(\mathcal D^{(2)}_\ell\). Since \(K_2\asymp\log p\) and the sample-size condition ensures \(K_2\le c m_2\), the fixed-point part of Proposition~\ref{prop:linear_mom_srs} applies conditionally and gives failure probability at most \(Cp^{-c}\) for each \(\ell\). By Lemma~\ref{lem:srs_to_pathwise}, part (2), the whole Phase II path satisfies pathwise \((\mathcal I_2,\pi_2, \bar s,\phi_2,\gamma_2)\)-SRS for $\mathcal I_2=\{T_1,\ldots,T_1+T_2-1\}$, $\pi_2(T_1+\ell-1)=\ell$, with probability at least
\[
1-CT_2p^{-c}.
\]
Since \(T_2\asymp\log s\) and \(s\le p\), this probability is at least \(1-Cp^{-c}\) after decreasing \(c\) and increasing \(C\).

Next,
\[
\phi_2
\lesssim_{\lambda}
\sqrt{\bar s}\,
M_{\bm x,2+2\lambda,\bar s}^{1/(1+\lambda)}
\left(\frac{T_2\log p}{n}\right)^{\lambda/(1+\lambda)} .
\]
The condition \(\phi_2\le\bar\phi\) follows from
\[
n\gtrsim_{\kappa,\lambda,\eta_0}
M_{\bm x,2+2\lambda,\bar s}^{1/\lambda}
\bar s^{(1+\lambda)/(2\lambda)}
T_2\log p .
\]
This is implied by the theorem's sample-size condition because \(T_2\asymp\log s\) and \(\bar s=2s+s^*\) imply \(T_2\lesssim\bar s\). Hence both phases satisfy the multiplicative stability requirement \(\phi_1,\phi_2\le\bar\phi\).

Applying Theorem~\ref{thm:two_stage_right_convergence} on the intersection of the Phase I and Phase II pathwise SRS events yields
\[
\|\widehat{\bbm\theta}-\bbm\theta^*\|_2
\le
\rho^{T_2}\|\bbm\theta^{T_1}-\bbm\theta^*\|_2
+
C_\gamma\gamma_2 .
\]
Using the Phase I bound already proved, it remains to compare the carried-over localization term with the Phase II statistical floor. Since \(K_1\asymp\bar s\log p\), \(K_2\asymp\log p\), \(n_1\asymp n_2\asymp n\), and \(m_2=n_2/T_2\),
\[
\gamma_1
\lesssim
\sqrt{\bar s}\,
M_{\mathrm{eff},1+\delta}^{1/(1+\delta)}
\left(\frac{\bar s\log p}{n}\right)^{\delta/(1+\delta)},
\qquad
\gamma_2
\lesssim
\sqrt{\bar s}\,
M_{\mathrm{eff},1+\delta}^{1/(1+\delta)}
\left(\frac{T_2\log p}{n}\right)^{\delta/(1+\delta)} .
\]
Moreover, since $s$ is a constant multiple of $s^*$, so
\(\bar s\asymp s\). Choosing the implicit constant in \(T_2\asymp\log s\) sufficiently large gives
\[
\rho^{T_2}\gamma_1
\lesssim
\gamma_2 .
\]
Therefore
\[
\|\widehat{\bbm\theta}-\bbm\theta^*\|_2
\lesssim_{\kappa,\eta_0}
\gamma_2
\lesssim_{\kappa,\lambda,\delta,\eta_0}
\sqrt{\bar s}\,
M_{\mathrm{eff},1+\delta}^{1/(1+\delta)}
\left(\frac{\log s\cdot\log p}{n}\right)^{\delta/(1+\delta)} ,
\]
which proves \eqref{eq:linear_final_rate}.

Finally, the Phase I failure probability is at most \(C\exp(-c\bar s\log p)\), and the Phase II failure probability is at most \(Cp^{-c}\). A union bound gives the claimed probability
\[
1-C\exp(-c\bar s\log p)-Cp^{-c}.
\]
This completes the proof.
\end{proof}

\subsection{Auxiliary Lemmas}
We present some auxiliary lemmas that are used in the proofs of the main results.

The following lemma, adapted from \cite[Lemma 3]{bubeck2013bandits}, provides a concentration inequality for the sample mean under bounded $(1+\lambda)$-th moments. It is used to bound the MoM gradient estimator in Proposition \ref{prop:linear_mom_srs}. The proof is omitted.

\begin{lemma}\label{supp:concentration_inequality_for_sample_mean_under_heavy_tailed_distribution}
	Let $X, X_1, \dots, X_n$ be a real i.i.d. sequence with finite mean $\mu$. We assume that for some $\lambda \in (0, 1]$ and $v \ge 0$, one has $\mathbb{E}[|X - \mu|^{1+\lambda}] \le v$. Let $\widehat{\mu}$ be the empirical mean:
$
\widehat{\mu} = n^{-1} \sum_{t=1}^n X_t\,.
$
Then for any $\omega > 0$,
$$
\mathbb{P}(\widehat \mu-\mu>\omega)\le \frac{3v}{n^{\lambda}\omega^{1+\lambda}}.
$$
\end{lemma}

The following two lemmas are used to control the design-driven term in the MoM gradient estimator, which is a supremum over all sparse directions. The first lemma is a VC deviation bound for a binary class, and the second lemma bounds the shatter coefficient of the bad-direction class on a fixed support.
\begin{lemma}[VC deviation bound for a binary class]
\label{supp:lemma:vc_deviation_binary_class}
Let $W_1,\ldots,W_K$ be i.i.d. random variables taking values in a measurable space
$\mathcal W$, and let $\mathcal A$ be a class of measurable subsets of $\mathcal W$.
Define
\[
\nu(A)=\mathbb P(W_1\in A),
\qquad
\nu_K(A)=\frac1K\sum_{k=1}^K \mathbf 1\{W_k\in A\},
\qquad A\in\mathcal A.
\]
Let $s(\mathcal A,m)$ denote the $m$-th shatter coefficient of $\mathcal A$.
Then for any $t>0$,
\[
\mathbb P\!\left(
\sup_{A\in\mathcal A} |\nu_K(A)-\nu(A)|>t
\right)
\le
8\, s(\mathcal A,K)\exp\!\left(-\frac{Kt^2}{32}\right).
\]
Consequently, if $\sup_{A\in\mathcal A}\nu(A)\le q$, then for any $t>0$,
\[
\mathbb P\!\left(
\sup_{A\in\mathcal A} \nu_K(A)>q+t
\right)
\le
8\, s(\mathcal A,K)\exp\!\left(-\frac{Kt^2}{32}\right).
\]
\end{lemma}

\begin{proof}
The first inequality is a direct consequence of Theorem 12.5 of
\citet{devroye2013probabilistic}. The second follows from the inclusion
\[
\left\{
\sup_{A\in\mathcal A}\nu_K(A)>q+t
\right\}
\subseteq
\left\{
\sup_{A\in\mathcal A}\big(\nu_K(A)-\nu(A)\big)>t
\right\}
\subseteq
\left\{
\sup_{A\in\mathcal A}|\nu_K(A)-\nu(A)|>t
\right\}.
\]
\end{proof}

\begin{lemma}[Shatter coefficient of the bad-direction class on a fixed support]
\label{supp:lemma:shatter_bound_bad_direction_class}
Fix a support set $U\subset[p]$ with $|U|=d\ge 1$, and let $\phi_0>0$.
Define the class of subsets of $\mathbb R^d$ by
\[
\mathcal A_U
=
\Big\{
\{w\in\mathbb R^d: |w^\top u|>\phi_0\}
:\ u\in\mathbb R^d,\ \|u\|_2=1
\Big\}.
\]
Then for every integer $m\ge 1$,
\[
s(\mathcal A_U,m)
\le
\left(
2\sum_{\ell=0}^d \binom{m-1}{\ell}
\right)^2.
\]
In particular, there exists an absolute constant $C>0$ such that
\[
s(\mathcal A_U,m)
\le
\exp\!\left\{
C d \log\!\left(\frac{e(m\vee d)}{d}\right)
\right\},
\qquad \forall\, m\ge 1.
\]
\end{lemma}

\begin{proof}
Define
\[
\mathcal A_U^+
=
\Big\{
\{w\in\mathbb R^d: w^\top u>\phi_0\}
:\ u\in\mathbb R^d,\ \|u\|_2=1
\Big\},
\]
and
\[
\mathcal A_U^-
=
\Big\{
\{w\in\mathbb R^d: w^\top u<-\phi_0\}
:\ u\in\mathbb R^d,\ \|u\|_2=1
\Big\}.
\]
Each of $\mathcal A_U^+$ and $\mathcal A_U^-$ is a subclass of the class of affine halfspaces
in $\mathbb R^d$. Hence, by Corollary 13.1 of \citet{devroye2013probabilistic},
\[
s(\mathcal A_U^+,m)
\le
2\sum_{\ell=0}^d \binom{m-1}{\ell},
\qquad
s(\mathcal A_U^-,m)
\le
2\sum_{\ell=0}^d \binom{m-1}{\ell}.
\]

Now let
\[
\widetilde{\mathcal A}_U
=
\{A^+\cup A^-: A^+\in\mathcal A_U^+,\ A^-\in\mathcal A_U^-\}.
\]
Since
\[
\{w:|w^\top u|>\phi_0\}
=
\{w:w^\top u>\phi_0\}\cup \{w:w^\top u<-\phi_0\},
\]
we have $\mathcal A_U\subseteq \widetilde{\mathcal A}_U$. Therefore, by Theorem 13.5(iv)
of \citet{devroye2013probabilistic},
\[
s(\mathcal A_U,m)
\le
s(\widetilde{\mathcal A}_U,m)
\le
s(\mathcal A_U^+,m)\, s(\mathcal A_U^-,m)
\le
\left(
2\sum_{\ell=0}^d \binom{m-1}{\ell}
\right)^2.
\]

For the second claim, if $m\ge d$, then the standard combinatorial bound gives
\[
\sum_{\ell=0}^d \binom{m-1}{\ell}
\le
\left(\frac{e(m-1)}{d}\right)^d
\le
\left(\frac{em}{d}\right)^d.
\]
If $m<d$, then trivially $s(\mathcal A_U,m)\le 2^m\le 2^d$.
Combining the two cases yields
\[
s(\mathcal A_U,m)
\le
\exp\!\left\{
C d \log\!\left(\frac{e(m\vee d)}{d}\right)
\right\}
\]
for a sufficiently large absolute constant $C>0$.
\end{proof}


\section{Heavy-Tailed Logistic Regression}
\label{append:E}

\subsection{Proof of SRCG for Logistic Regression}
\begin{proof}[Proof of Proposition \ref{prop:logistic_srcg}]
Let $\boldsymbol\Delta=\boldsymbol\theta-\boldsymbol\theta^*$. Under the logistic model,
\[
\mathbb E[y_i\mid \bm x_i]
=
\sigma(\bm x_i^\top\boldsymbol\theta^*),
\qquad
\sigma(z)=\frac{e^z}{1+e^z}.
\]
Therefore,
\[
\mathbb E[\nabla_{\bbm \theta} \mathcal{L}(\bbm\theta;Z_i)]
=
\mathbb E\left[
\left\{
\sigma(\bm x_i^\top\boldsymbol\theta)
-
\sigma(\bm x_i^\top\boldsymbol\theta^*)
\right\}
\bm x_i
\right].
\]

By the integral form of the mean value theorem,
\[
\sigma(\bm x_i^\top\boldsymbol\theta)
-
\sigma(\bm x_i^\top\boldsymbol\theta^*)
=
(\bm x_i^\top\boldsymbol\Delta)
\int_0^1
\sigma'\left\{
\bm x_i^\top
(\boldsymbol\theta^*+u\boldsymbol\Delta)
\right\}
\,du .
\]
Since
\[
\boldsymbol\theta^*+u\boldsymbol\Delta
=
(1-u)\boldsymbol\theta^*+u\boldsymbol\theta
\in B(R)
\]
for every $u\in[0,1]$, and since
\[
\operatorname{supp}(\boldsymbol\theta^*+u\boldsymbol\Delta)
\subseteq
\operatorname{supp}(\boldsymbol\theta)\cup
\operatorname{supp}(\boldsymbol\theta^*)
\subseteq S,
\]
Assumption \ref{assump:sparse_logistic_curvature} applies along the whole line segment between
$\boldsymbol\theta^*$ and $\boldsymbol\theta$.

Define the $p\times p$ matrix
\[
\boldsymbol H
=
\int_0^1
\mathbb E\left[
\sigma'\left\{
\bm x_i^\top
(\boldsymbol\theta^*+u\boldsymbol\Delta)
\right\}
\bm x_{i}\bm x_{i}^\top
\right]du ,
\]
By Assumption \ref{assump:sparse_logistic_curvature} and Assumption \ref{assump:linear_sparse_eigenvalues}, for any $\bm v\in \mathbb R^p$ and $\operatorname{supp}(\bm v)\subseteq S$,
\[
\bm v^\top\boldsymbol H\bm v\ge c_\kappa\kappa_- \|\bm v\|_2^2.
\]
Moreover, since $\sigma'(z)\le 1/4$ for all $z$ and Assumption
\ref{assump:linear_sparse_eigenvalues} holds, for any $\bm v\in \mathbb R^p$ and $\operatorname{supp}(\bm v)\subseteq S$,
\[
\bm v^{\top}\boldsymbol H\bm v
\le
\frac14
\bm v^{\top} \mathbb E[\bm x_{i}\bm x_{i}^\top]\bm v
\le
\frac{\kappa_+}{4}\Vert\bm v\Vert_2^2.
\]
Note that $\operatorname{supp}(\boldsymbol\Delta)\subseteq S$, and
\(
\left\langle
\mathbb E[\nabla_{\bbm \theta} \mathcal{L}(\bbm\theta;Z_i)],
\boldsymbol\Delta
\right\rangle
=
\boldsymbol\Delta^\top
\boldsymbol H
\boldsymbol\Delta.
\)
The lower bound on $\boldsymbol H$ gives
\[
\left\langle
\mathbb E[\nabla_{\bbm \theta} \mathcal{L}(\bbm\theta;Z_i)],
\boldsymbol\Delta
\right\rangle
\ge
c_\kappa\kappa_-
\|\boldsymbol\Delta\|_2^2.
\]
Since \(\boldsymbol H\) is positive semidefinite,
the Cauchy--Schwarz inequality in the seminorm induced by
\(\boldsymbol H\) yields
\[
|\bm v^\top \boldsymbol H\boldsymbol\Delta|
\le
(\bm v^\top\boldsymbol H\bm v)^{1/2}
(\boldsymbol\Delta^\top\boldsymbol H\boldsymbol\Delta)^{1/2}.
\]
Therefore,
\[
\begin{aligned}
\left\|
\mathbb E[\nabla_{\bbm \theta} \mathcal{L}(\bbm\theta;Z_i)]_S
\right\|_2
&=
\sup_{\substack{\|\bm v\|_2=1\\ \operatorname{supp}(\bm v)\subseteq S}}
\left|
\bm v^\top
\mathbb E[\nabla_{\bbm \theta} \mathcal{L}(\bbm\theta;Z_i)]
\right|  \\
&=
\sup_{\substack{\|\bm v\|_2=1\\ \operatorname{supp}(\bm v)\subseteq S}}
\left|
\bm v^\top\boldsymbol H\boldsymbol\Delta
\right|  \\
&\le
\left(\frac{\kappa_+}{4}\right)^{1/2}
\left(
\boldsymbol\Delta^\top\boldsymbol H\boldsymbol\Delta
\right)^{1/2}.
\end{aligned}
\]
Equivalently,
\[
\left\|
\mathbb E[\nabla_{\bbm \theta} \mathcal{L}(\bbm\theta;Z_i)]_S
\right\|_2^2
\le
\frac{\kappa_+}{4}
\left\langle
\mathbb E[\nabla_{\bbm \theta} \mathcal{L}(\bbm\theta;Z_i)],
\boldsymbol\Delta
\right\rangle .
\]
Hence,
\[
\frac{2}{\kappa_+}
\left\|
\mathbb E[\nabla_{\bbm \theta} \mathcal{L}(\bbm\theta;Z_i)_S]
\right\|_2^2
\le
\frac12
\left\langle
\mathbb E[\nabla_{\bbm \theta} \mathcal{L}(\bbm\theta;Z_i)],
\boldsymbol\Delta
\right\rangle .
\]
Also,
\[
\frac{c_\kappa\kappa_-}{2}
\|\boldsymbol\Delta\|_2^2
\le
\frac12
\left\langle
\mathbb E[\nabla_{\bbm \theta} \mathcal{L}(\bbm\theta;Z_i)],
\boldsymbol\Delta
\right\rangle .
\]
Adding the last two inequalities yields
\[
\left\langle
\mathbb E[\nabla_{\bbm \theta} \mathcal{L}(\bbm\theta;Z_i)],
\boldsymbol\Delta
\right\rangle
\ge
\frac{2}{\kappa_+}
\left\|
\mathbb E[\nabla_{\bbm \theta} \mathcal{L}(\bbm\theta;Z_i)]_S
\right\|_2^2
+
\frac{c_\kappa\kappa_-}{2}
\|\boldsymbol\Delta\|_2^2.
\]
This proves the SRCG condition.
\end{proof}

\subsection{Moment Bounds for Logistic Regression Score Components}
\begin{lemma}\label{supp:lemma:moment_result_for_logistic_regression}
	Under Assumption \ref{assump:linear_sparse_eigenvalues}, for any fixed $\bbm \theta\in \mathbb{R}^p$ with $\Vert\bbm \theta\Vert_0\le s$, and for any $j\in \{1,2,\ldots,p\}$, the $j$-th elements of the gradient $\nabla_{\bbm \theta} \mathcal{L}(\bbm\theta;Z_i)$ in logistic regression have finite variance, i.e., 
 
\[
\mathrm{Var}\left(
\left(
\frac{e^{\langle \bm{x}_i,\bbm{\theta}\rangle}}
{1+e^{\langle \bm{x}_i,\bbm{\theta}\rangle}}
-y_i
\right)x_{ij}
\right)
\le M_{\bm{x},2,1}
.\]
\end{lemma}
\begin{proof}
By Assumption \ref{assump:linear_sparse_eigenvalues}, $M_{\bm{x},2,1}$ is finite and the second moment of $x_{ij}$ is bounded by $M_{\bm{x},2,1}$.

Note that
\[\left|\frac{e^{\langle \bm{x}_i,\bbm{\theta}\rangle}}
{1+e^{\langle \bm{x}_i,\bbm{\theta}\rangle}}
-y_i\right|\le 1,\]
so
\[\mathrm{Var}\left(
\left(
\frac{e^{\langle \bm{x}_i,\bbm{\theta}\rangle}}
{1+e^{\langle \bm{x}_i,\bbm{\theta}\rangle}}
-y_i
\right)x_{ij}
\right)
\le \mathbb{E}\left[\left(
\left(
\frac{e^{\langle \bm{x}_i,\bbm{\theta}\rangle}}
{1+e^{\langle \bm{x}_i,\bbm{\theta}\rangle}}
-y_i
\right)x_{ij}\right)^2
\right]
\le M_{\bm{x},2,1}.\]
\end{proof}

\subsection{Fixed-Point and Uniform SRS for MoM Gradients}
\begin{proof}[Proof of Proposition \ref{prop:logistic_mom_srs}]
Since \(1\le \bar s\), Assumption~\ref{assump:linear_sparse_eigenvalues} implies
\[
M_{\bm x,2,1}\le M_{\bm x,2,\bar s}\le \kappa_+<\infty .
\]
Hence all quantities below are well defined.\\

\noindent\textbf{Step 1. Fixed-point SRS.}

Fix $\bbm{\theta}\in \Theta_s$, and for each $j\in[p]$, define
\[
\bbm{\mu}(\bbm{\theta})_j
=
\mathbb E\left[
\left(
\frac{e^{\langle \bm{x}_i,\bbm{\theta}\rangle}}
{1+e^{\langle \bm{x}_i,\bbm{\theta}\rangle}}
-y_i
\right)x_{ij}
\right]
=
\mathbb E[\nabla \mathcal L(\bbm{\theta};\bm{x}_i,y_i)]_j ,
\]
 and the $j$-th coordinate of the block-wise mean gradient is
\[
\bar{g}_k(\bbm{\theta};\mathcal A)_j
=
B_{\mathcal A}^{-1}
\sum_{(\bm{x}_i,y_i)\in \mathcal{B}_k(\mathcal A)}
\left(
\frac{e^{\langle \bm{x}_i,\bbm{\theta}\rangle}}
{1+e^{\langle \bm{x}_i,\bbm{\theta}\rangle}}
-y_i
\right)x_{ij},
\qquad B_{\mathcal A}=\frac{m}{K}.
\]
By Lemma \ref{supp:lemma:moment_result_for_logistic_regression}, the variance of $\bar{g}_k(\bbm{\theta};\mathcal A)_j$ is bounded by $M_{\bm{x},2,1}/B_{\mathcal A}$. By Chebyshev's inequality, for any $\omega>0$,
\[
\mathbb P\!\left(
|\bar g_k(\bbm{\theta};\mathcal A)_j-\bbm{\mu}(\bbm{\theta})_j|\ge \omega
\right)
\le
\frac{M_{\bm{x},2,1}}{B_{\mathcal A}\omega^2}.
\]
Set $\omega_0=2\sqrt{{M_{\bm{x},2,1}}/{B_{\mathcal A}}}$,
then
\[
\mathbb P\!\left(
|\bar g_k(\bbm{\theta};\mathcal A)_j-\bbm{\mu}(\bbm{\theta})_j|>\omega_0
\right)
\le \frac{1}{4}.
\]

Define
\[
I_{k,j}
=
\mathbf 1\!\left\{
|\bar g_k(\bbm{\theta};\mathcal A)_j-\bbm{\mu}(\bbm{\theta})_j|>\omega_0
\right\},
\qquad k=1,\dots,K.
\]
By the standard median-of-means argument,
\[
\mathbb P\!\left(
|g(\bbm{\theta};K,\mathcal A)_j-\bbm{\mu}(\bbm{\theta})_j|>\omega_0
\right)
\le
\mathbb P\!\left(
\frac{1}{K}\sum_{k=1}^K I_{k,j}>\frac{1}{2}
\right)
\le
2\exp(-K/8).
\]

Taking a union bound over $j\in[p]$, we obtain
\[
\mathbb P\left(
\max_{1\le j\le p}
| g(\bbm{\theta};K,\mathcal A)_j-\bbm{\mu}(\bbm{\theta})_j |
>
2\sqrt{\frac{M_{\bm{x},2,1}}{B_{\mathcal A}}}
\right)
\le 2p\exp(-K/8).
\]
On the complement of this event, for every $S\subset[p]$ with $|S|\le 2s+s^*$,
\[
\begin{aligned}
\left\lVert
 g(\bbm{\theta};K,\mathcal A)_S
 - \mathbb E[\nabla \mathcal L(\bbm{\theta};\bm{x}_i,y_i)]_{S}
 \right\rVert _2 
&\le
\sqrt{|S|}
\max_{1\le j\le p}
\left| g(\bbm{\theta};K,\mathcal A)_j-\bbm{\mu}(\bbm{\theta})_j \right| \\
&\lesssim
\sqrt{M_{\bm{x},2,1}(2s+s^*)}
\sqrt{\frac{K}{m}}.
\end{aligned}
\]
This inequality holds with probability at least $1-2p\exp(-K/8)$. If $K\ge C\log p$ for some constant $C>0$, then $1-2p\exp(-K/8)\ge 1-Cp^{-c}$ for some constant $c>0$, which proves the fixed-point SRS condition.\\

\noindent\textbf{Step 2. Uniform SRS.}

 Fix $j\in[p]$. Define a class of functions $\mathcal F_j$ by
\[
\mathcal F_j
=
\left\{
f_{\bbm{\theta},j}:(\bm{x},y)\mapsto
\left(
\frac{e^{\langle \bm{x},\bbm{\theta}\rangle}}
{1+e^{\langle \bm{x},\bbm{\theta}\rangle}}
-y
\right)x_j
:\ \Vert\bbm{\theta}\Vert_0\le s
\right\}.
\]
For $f_{\bbm{\theta},j}\in\mathcal F_j$, write
\[
Pf_{\bbm{\theta},j}
=
\bbm{\mu}(\bbm{\theta})_j
=
\mathbb E[\nabla\mathcal L(\bbm{\theta};\bm{x}_i,y_i)]_j.
\]
The coordinate-wise MoM estimator satisfies
\[
g(\bbm{\theta};K,\mathcal A)_j
=
\operatorname{median}
\left\{
B_{\mathcal A}^{-1}
\sum_{(\bm{x}_i,y_i)\in\mathcal B_k(\mathcal A)}
f_{\bbm{\theta},j}(\bm{x}_i,y_i)
:\ 1\le k\le K
\right\}.
\]
We want to control the deviation of the MoM estimator from the population mean uniformly over $\Theta_{s}$, so we apply Lemma \ref{RIGHT:lemma:uniform_MoM_function_class} to $\mathcal F_j$ by verifying the three conditions therein.

First, we verify the block-tail condition. By Lemma \ref{supp:lemma:moment_result_for_logistic_regression}, we have, uniformly over $\Vert\bbm{\theta}\Vert_0\le s$,
\[
\mathbb E\!\left[f_{\bbm{\theta},j}(\bm{x}_i,y_i)^2\right]
\le
\mathbb E(x_{ij}^2)
\le
M_{\bm{x},2,1}.
\]
Therefore,
\[
\sup_{\Vert\bbm{\theta}\Vert_0\le s}
\mathrm{Var}\!\left(
B_{\mathcal A}^{-1}
\sum_{(\bm{x}_i,y_i)\in\mathcal B_1(\mathcal A)}
f_{\bbm{\theta},j}(\bm{x}_i,y_i)
\right)
\le
\frac{M_{\bm{x},2,1}}{B_{\mathcal A}}.
\]
By Chebyshev's inequality,
\[
\sup_{\Vert\bbm{\theta}\Vert_0\le s}
\mathbb P\!\left(
\left|
B_{\mathcal A}^{-1}
\sum_{(\bm{x}_i,y_i)\in\mathcal B_1(\mathcal A)}
f_{\bbm{\theta},j}(\bm{x}_i,y_i)
-
Pf_{\bbm{\theta},j}
\right|
>\omega
\right)
\le
\frac{M_{\bm{x},2,1}}{B_{\mathcal A}\omega^2}.
\]
Taking
\[
\omega_0=5\sqrt{\frac{M_{\bm{x},2,1}}{B_{\mathcal A}}},
\]
we obtain
\[
p_{B_{\mathcal A}}(\omega_0)
=
\sup_{f\in\mathcal F_j}
\mathbb P\!\left(
|\bar f_{1,\mathcal A}-Pf|>\omega_0
\right)
\le \frac{1}{25}<\frac{1}{20},
\]
where
\[
\bar f_{1,\mathcal A}
=
B_{\mathcal A}^{-1}
\sum_{(\bm{x}_i,y_i)\in\mathcal B_1(\mathcal A)}
f(\bm{x}_i,y_i).
\]

Second, we verify the packing condition. For every $f_{\bbm{\theta},j}\in\mathcal F_j$, under the law $P$ of $(\bm{x}_i,y_i)$,
\[
\|f_{\bbm{\theta},j}\|_{L_2(P)}^2
=
\mathbb E\!\left[f_{\bbm{\theta},j}(\bm{x}_i,y_i)^2\right]
\le
M_{\bm{x},2,1}.
\]
If we choose $\eta_1=4\sqrt{M_{\bm{x},2,1}}$, then the packing number $M(\mathcal F_j,\eta_1 D)=1$, where $D$ is the unit ball in $L_2(P)$.

Third, we verify the localized empirical-process condition. Let
\[
W_j=(\mathcal F_j-\mathcal F_j)\cap \eta_1 D,
\qquad
\bar W_j=\{w-Pw:\ w\in W_j\}.
\]
By Lemma \ref{RIGHT:lemma:shatter_coefficient_logistic_coordinate_class},
$\mathcal F_j$ is a VC-subgraph class with VC index
\[
V_j \le C s\log\!\left(\frac{ep}{s}\right),
\]
and envelope $F_j(\bm{x},y)=|x_j|$. Since $W_j\subset \mathcal F_j-\mathcal F_j$, the envelope of $W_j$ is bounded by
$2F_j$. Moreover, for every probability measure $Q$ and every $0<\varepsilon\le 1$,
\[
N\!\left(
\varepsilon \|2F_j\|_{Q,2},
W_j,
L_2(Q)
\right)
\le
N\!\left(
\varepsilon \|F_j\|_{Q,2},
\mathcal F_j,
L_2(Q)
\right)^2.
\]
Hence, by Lemma \ref{RIGHT:lemma:shatter_coefficient_logistic_coordinate_class},
\[
\log N\!\left(
\varepsilon \|2F_j\|_{Q,2},
W_j,
L_2(Q)
\right)
\le
C V_j \log\!\left(\frac{A}{\varepsilon}\right),
\qquad 0<\varepsilon\le 1.
\]
Therefore the uniform entropy integral of $W_j$ satisfies
\[
\sup_Q
\int_0^1
\sqrt{
1+
\log N\!\left(
\varepsilon \|2F_j\|_{Q,2},
W_j,
L_2(Q)
\right)
}\,d\varepsilon
\lesssim
\sqrt{V_j}.
\]
Since
\[
\|2F_j\|_{P,2}^2
=
4\mathbb E(x_{ij}^2)
\le
4M_{\bm{x},2,1},
\]
by the proof of \citet[Theorem 2.14.1]{van1996weak},
\[
\mathbb E\sup_{w\in W_j}
\left|
\sum_{(\bm{x}_i,y_i)\in\mathcal A}
\varepsilon_i w(\bm{x}_i,y_i)
\right|
\lesssim
\sqrt{m}\,\|2F_j\|_{P,2}\,\sqrt{V_j}
\lesssim
\sqrt{mM_{\bm{x},2,1}V_j}.
\]
By symmetrization,
\[
\mathbb E\sup_{\bar w\in\bar W_j}
\left|
\sum_{(\bm{x}_i,y_i)\in\mathcal A}
\varepsilon_i \bar w(\bm{x}_i,y_i)
\right|
\le
2\mathbb E\sup_{w\in W_j}
\left|
\sum_{(\bm{x}_i,y_i)\in\mathcal A}
\varepsilon_i w(\bm{x}_i,y_i)
\right|
\lesssim
\sqrt{mM_{\bm{x},2,1}V_j}.
\]

Now choose
\[
\eta_2 = A_0 \sqrt{\frac{K M_{\bm{x},2,1}}{m}},
\]
where $A_0>0$ is a sufficiently large absolute constant. Since
$B_{\mathcal A}=m/K$ and
$\eta_1=4\sqrt{M_{\bm{x},2,1}}$, we have
\[
\eta_2 \ge c_3\frac{\eta_1}{\sqrt {B_{\mathcal A}}}
\]
for all sufficiently large $A_0$. Moreover, if
\(
K\ge C V_j,
\)
then
\(
\sqrt{mM_{\bm{x},2,1}V_j}
\lesssim
\eta_2 m
\).

Thus all assumptions of Lemma \ref{RIGHT:lemma:uniform_MoM_function_class} are
satisfied, and therefore
\[
\mathbb P\left(
\sup_{\Vert\bbm{\theta}\Vert_0\le s}
|g(\bbm{\theta};K,\mathcal A)_j-\bbm{\mu}(\bbm{\theta})_j|
>
C\sqrt{M_{\bm{x},2,1}}\sqrt{\frac{K}{m}}
\right)
\le
2\exp(-cK).
\]

Finally, taking a union bound over $j\in[p]$, we obtain, after adjusting the constants, that
\[
\mathbb P\left(
\sup_{\Vert\bbm{\theta}\Vert_0\le s}
\max_{1\le j\le p}
|g(\bbm{\theta};K,\mathcal A)_j-\bbm{\mu}(\bbm{\theta})_j|
>
C\sqrt{M_{\bm{x},2,1}}\sqrt{\frac{K}{m}}
\right)
\le
2p\exp(-cK).
\]
On the complement of this event, for every $\bbm{\theta}$ with
$\Vert\bbm{\theta}\Vert_0\le s$ and every $S\subset[p]$ with $|S|\le 2s+s^*$,
\begin{align}
\left\lVert
g(\bbm{\theta};K,\mathcal A)_S
-
\mathbb E[\nabla \mathcal L(\bbm{\theta};\bm{x}_i,y_i)]_{S}
\right\rVert _2
&\le
\sqrt{|S|}
\sup_{\Vert\bbm{\theta}\Vert_0\le s}
\max_{1\le j\le p}
|g(\bbm{\theta};K,\mathcal A)_j-\bbm{\mu}(\bbm{\theta})_j|\\
&\lesssim
\sqrt{\frac{M_{\bm{x},2,1}\bar s K}{m}}.\label{supp:eq:uniform_srs_logistic_regression}
\end{align}
Since $K\ge C V_j$, we can take $K\ge C \bar s \log p$, then \eqref{supp:eq:uniform_srs_logistic_regression} holds with probability at least $1-C\exp(-c\bar s \log p)$ for adjusted constants $C,c>0$. This proves the uniform SRS condition.
\end{proof}

\subsection{Proof of Theorem \ref{thm:logistic_two_stage_right}}
\begin{proof}[Proof of Theorem~\ref{thm:logistic_two_stage_right}]
Throughout the proof, constants \(C,c>0\) may change from line to line.
Let $a={2}/{\kappa_+}$, $b={c_\kappa\kappa_-}/{2}$. By Proposition~\ref{prop:logistic_srcg}, the population logistic score map
satisfies the SRCG inequality on \(\Theta_s(R)\) with constants \(a\) and \(b\).
Let \(\bar\phi,\rho,C_\gamma\), and \(c_s\) be the constants in
Lemma~\ref{lem:one_stage_contraction_pathwise_srs}. We take
\(s\ge c_s s^*\) and $s\asymp s^*$. Since the logistic SRS bounds in
Proposition~\ref{prop:logistic_mom_srs} have no multiplicative term, the
corresponding values of \(\phi\) in both phases are equal to zero, and hence are
smaller than \(\bar\phi\).

The initialization satisfies \(\|\bbm\theta^0\|_0\le s\), and every update in
Algorithm~\ref{alg:two_stage_right} is hard-thresholded at level \(s\).
Therefore all iterates are \(s\)-sparse. 

\noindent\textbf{Phase I.}
Set
\[
\gamma_1=\gamma_{\mathrm{log}}(n_1,K_1)
=
C\sqrt{\bar s\,M_{\bm x,2,1}}
\left(\frac{K_1}{n_1}\right)^{1/2}.
\]
Since \(K_1\asymp \bar s\log p\), \(n_1\asymp n\), and the sample-size condition
implies \(K_1\le c n_1\) after enlarging the implicit constant,
Proposition~\ref{prop:logistic_mom_srs}, part (1), gives a uniform
\((\Theta_s,\bar s,0,\gamma_1)\)-SRS event on \(\mathcal D^{(1)}\) with
probability at least \(1-C\exp\{-c\bar s\log p\}\). Denote this event by \(E_1\).

On \(E_1\), since every update is
hard-thresholded at level \(s\), the uniform-to-pathwise route in
Lemma~\ref{lem:srs_to_pathwise}, part (1), applies with
\[
\mathcal I_1=\{0,\ldots,T_1-1\},
\qquad
\mathcal A_0=\mathcal D^{(1)},
\qquad
\pi_1(t)=0 .
\]
Consequently, the Phase I gradient evaluations satisfy pathwise
\((\mathcal I_1,\pi_1,\bar s,0,\gamma_1)\)-SRS.

We now verify the required $\|\bbm\theta^t\|_2\le R$ condition for SRCG in Proposition \ref{prop:logistic_srcg} by induction. We prove by induction that \(\bbm\theta^t\in B(R)\) for
\(t=0,\ldots,T_1\). Since
\[
\|\bbm\theta^0\|_2
\le
\|\bbm\theta^*\|_2+\|\bbm\theta^0-\bbm\theta^*\|_2
\le
\|\bbm\theta^*\|_2+R_0
<
R,
\]
the claim holds at \(t=0\). Suppose it holds for
\(t=0,\ldots,k\), where \(k<T_1\). Then, for each
\(t=0,\ldots,k\), we have
\(\bbm\theta^t\in\Theta_s(R)\), and hence the localized SRCG inequality in
Proposition~\ref{prop:logistic_srcg} applies at \(\bbm\theta^t\).

Inspecting the proof of Lemma~\ref{lem:one_stage_contraction_pathwise_srs},
the same deterministic recursion holds whenever the SRCG inequality is valid
along the realized path. Therefore, applying that pathwise contraction argument
to the first \(k+1\) updates yields
\[
\|\bbm\theta^{k+1}-\bbm\theta^*\|_2
\le
\rho^{k+1}\|\bbm\theta^0-\bbm\theta^*\|_2
+
C_\gamma\gamma_1 .
\]
By the sample-size condition, after choosing the implicit constant sufficiently
large, \(C_\gamma\gamma_1\le 1/2\). Hence
\[
\|\bbm\theta^{k+1}-\bbm\theta^*\|_2
\le
R_0+\frac12 .
\]
Consequently,
\[
\|\bbm\theta^{k+1}\|_2
\le
\|\bbm\theta^*\|_2+R_0+\frac12
<
R .
\]
This completes the induction. Thus all Phase I iterates belong to
\(\Theta_s(R)\), and the SRCG inequality is valid throughout Phase I.

Since both pathwise SRS and the SRCG inequality are valid, appling Lemma~\ref{lem:one_stage_contraction_pathwise_srs} to Phase I with
\(t_0=0\), \(T=T_1\), \(\phi=0\), and \(\gamma=\gamma_1\), we obtain
\[
\|\bbm\theta^{T_1}-\bbm\theta^*\|_2
\le
\rho^{T_1}\|\bbm\theta^0-\bbm\theta^*\|_2
+
C_\gamma\gamma_1 .
\]
By assumption, \(\|\bbm\theta^0-\bbm\theta^*\|_2\le R_0\), where \(R_0\) is at
most polynomial in \(n\). Therefore the implicit constant in
\(T_1\asymp\log n\) can be chosen sufficiently large so that
\[
\rho^{T_1}R_0\le C_\gamma\gamma_1 .
\]
It follows that, on \(E_1\),
\[
\|\bbm\theta^{T_1}-\bbm\theta^*\|_2
\le
C_\gamma'\gamma_1 .
\]
Using \(K_1\asymp\bar s\log p\) and \(n_1\asymp n\), we have
\[
\gamma_1
\lesssim
\sqrt{\bar s\,M_{\bm x,2,1}}
\left(\frac{\bar s\log p}{n}\right)^{1/2}.
\]
Hence
\[
\|\bbm\theta^{T_1}-\bbm\theta^*\|_2
\lesssim
\sqrt{\bar s\,M_{\bm x,2,1}}
\left(\frac{\bar s\log p}{n}\right)^{1/2},
\]
with the implicit constant depending on the SRCG constants and \(\eta_0\).
This proves the Phase I bound.

\noindent\textbf{Phase II.}
Let
\[
m_2=\frac{n_2}{T_2}
\]
denote the Phase II batch size. Unequal batch sizes only change constants if
\(m_2\) is replaced by the minimum batch size. Set
\[
\gamma_2=\gamma_{\mathrm{log}}(m_2,K_2)
=
C\sqrt{\bar s\,M_{\bm x,2,1}}
\left(\frac{K_2}{m_2}\right)^{1/2}.
\]
For each \(\ell=1,\ldots,T_2\), define the sigma-field
\[
\mathcal F_{\ell-1}
=
\sigma\!\left(
\mathcal D^{(1)},
\mathcal D^{(2)}_1,\ldots,\mathcal D^{(2)}_{\ell-1}
\right).
\]
The iterate \(\bbm\theta^{T_1+\ell-1}\) is
\(\mathcal F_{\ell-1}\)-measurable, while the fresh batch
\(\mathcal D^{(2)}_\ell\) is independent of \(\mathcal F_{\ell-1}\). Conditional
on \(\mathcal F_{\ell-1}\), the current iterate
\(\bbm\theta^{T_1+\ell-1}\) is therefore fixed relative to
\(\mathcal D^{(2)}_\ell\). Moreover, by the same induction argument as in Phase I, the bounded-region condition holds, so \(\bbm\theta^{T_1+\ell-1}\in\Theta_s(R)\).

Since \(K_2\asymp\log p\), \(m_2=n_2/T_2\), \(n_2\asymp n\), and the
sample-size condition ensures \(K_2\le c m_2\) after enlarging constants,
Proposition~\ref{prop:logistic_mom_srs}, part (2), applied conditionally on
\(\mathcal F_{\ell-1}\), gives
\[
\mathbb P\!\left(
\mathsf{SRS}\bigl(
\bbm\theta^{T_1+\ell-1};
K_2,\mathcal D^{(2)}_\ell,\bar s,0,\gamma_2
\bigr)^c
\;\middle|\;
\mathcal F_{\ell-1}
\right)
\le
Cp^{-c}.
\]
Now apply the independent-batch fixed-point route in
Lemma~\ref{lem:srs_to_pathwise}, part (2), with
\[
\mathcal I_2=\{T_1,\ldots,T_1+T_2-1\},
\qquad
\pi_2(T_1+\ell-1)=\ell .
\]
It follows that the whole Phase II path satisfies pathwise
\((\mathcal I_2,\pi_2,\bar s,0,\gamma_2)\)-SRS with probability at least
\[
1-CT_2p^{-c}.
\]
Since \(T_2\asymp\log s\) and \(s\le p\), this probability is at least
\(1-Cp^{-c}\) after adjusting \(C\) and \(c\).

On the intersection of the Phase I event and the Phase II pathwise SRS event,
Lemma~\ref{lem:one_stage_contraction_pathwise_srs}, applied to Phase II with
\(t_0=T_1\), \(T=T_2\), \(\phi=0\), and \(\gamma=\gamma_2\), yields
\[
\|\widehat{\bbm\theta}-\bbm\theta^*\|_2
=
\|\bbm\theta^{T_1+T_2}-\bbm\theta^*\|_2
\le
\rho^{T_2}\|\bbm\theta^{T_1}-\bbm\theta^*\|_2
+
C_\gamma\gamma_2 .
\]
Using the Phase I bound already proved,
\[
\|\widehat{\bbm\theta}-\bbm\theta^*\|_2
\le
C\rho^{T_2}\gamma_1
+
C_\gamma\gamma_2 .
\]

It remains to compare the carried-over Phase I error with the Phase II
statistical radius. Since \(K_1\asymp\bar s\log p\), \(K_2\asymp\log p\),
\(n_1\asymp n_2\asymp n\), and \(m_2=n_2/T_2\),
\[
\gamma_1
\lesssim
\sqrt{\bar s\,M_{\bm x,2,1}}
\left(\frac{\bar s\log p}{n}\right)^{1/2},
\qquad
\gamma_2
\lesssim
\sqrt{\bar s\,M_{\bm x,2,1}}
\left(\frac{T_2\log p}{n}\right)^{1/2}.
\]
Therefore
\[
\frac{\gamma_1}{\gamma_2}
\lesssim
\left(\frac{\bar s}{T_2}\right)^{1/2}.
\]
Because \(s\asymp s^*\), we have \(\bar s=2s+s^*\asymp s\). Choosing the
implicit constant in \(T_2\asymp\log s\) sufficiently large gives
\[
\rho^{T_2}
\left(\frac{\bar s}{T_2}\right)^{1/2}
\le C,
\]
where $C>0$ is an absolute constant, and hence
\[
\rho^{T_2}\gamma_1\lesssim \gamma_2 .
\]
Consequently,
\[
\|\widehat{\bbm\theta}-\bbm\theta^*\|_2
\lesssim
\gamma_2 .
\]
Finally, using \(K_2\asymp\log p\), \(m_2=n_2/T_2\), \(T_2\asymp\log s\),
and \(n_2\asymp n\), we obtain
\[
\gamma_2
\lesssim
\sqrt{\bar s\,M_{\bm x,2,1}}
\left(\frac{\log s\cdot\log p}{n}\right)^{1/2}.
\]
Thus
\[
\|\widehat{\bbm\theta}-\bbm\theta^*\|_2
\lesssim
\sqrt{\bar s\,M_{\bm x,2,1}}
\left(\frac{\log s\cdot\log p}{n}\right)^{1/2}.
\]

Combining the Phase I and Phase II failure probabilities by a union bound gives
the overall probability
\[
1-C\exp\{-c\bar s\log p\}-Cp^{-c}.
\]
This completes the proof.
\end{proof}

\subsection{Auxiliary Lemmas}
The following lemma is a convenient corollary of the uniform median-of-means
control developed in \citet[Theorem 2]{lugosi2019near}.
\begin{lemma}[Uniform MoM estimator on a function class]
\label{RIGHT:lemma:uniform_MoM_function_class}
Let $\mathcal F\subset L_2(P)$ be a class of measurable real-valued functions, and let
$Z_1,\dots,Z_n$ be i.i.d.\ random variables with common law $P$. Partition $[n]$ into
$K$ equal-size blocks $\mathcal B_1,\dots,\mathcal B_K$ of size $B=n/K$. For
$f\in\mathcal F$, write
\[
Pf=\mathbb E[f(Z_1)],
\qquad
\bar f_k = B^{-1}\sum_{i\in\mathcal B_k} f(Z_i),
\qquad
\operatorname{MoM}_K(f)=\operatorname{med}(\bar f_1,\dots,\bar f_K).
\]
Let
\[
D=\{h:\ \|h\|_{L_2(P)}\le 1\},
\]
and let $M(\mathcal F,\eta_1 D)$ denote the $L_2(P)$-packing number of $\mathcal F$ by
balls of radius $\eta_1$. Assume that there exist $\eta_0,\eta_1,\eta_2>0$ such that
\[
p_B(\eta_0):=\sup_{f\in\mathcal F}\mathbb P\bigl(|\bar f_1-Pf|>\eta_0\bigr)\le \frac{1}{20},
\]
\[
\log M(\mathcal F,\eta_1 D)\le c_1 K\log(e/p_B(\eta_0)),
\]
and, with
\[
W=(\mathcal F-\mathcal F)\cap \eta_1 D,
\qquad
\bar W=\{w-Pw:\ w\in W\},
\]
\[
\mathbb E\sup_{\bar w\in\bar W}
\left|
\sum_{i=1}^n \varepsilon_i \bar w(Z_i)
\right|
\le c_2 \eta_2 n,
\qquad
\eta_2\ge c_3\frac{\eta_1}{\sqrt{B}},
\]
where $\varepsilon_1,\dots,\varepsilon_n$ are i.i.d.\ Rademacher random variables,
independent of $Z_1,\dots,Z_n$. Then
\[
\mathbb P\left(
\sup_{f\in\mathcal F}
\bigl|\operatorname{MoM}_K(f)-Pf\bigr|
>
\eta_0+\eta_2
\right)
\le
2\exp(-c_4 K),
\]
where $c_1,c_2,c_3,c_4>0$ are absolute constants.
\end{lemma}
To characterize the complexity of the logistic coordinate class, we introduce the following lemma.
\begin{lemma}[VC-subgraph complexity of the logistic coordinate class]
\label{RIGHT:lemma:shatter_coefficient_logistic_coordinate_class}
For each fixed $j\in[p]$, define
\[
\mathcal F_j
=
\left\{
f_{\bbm\theta,j}:(\bm x,y)\mapsto
\left(
\frac{e^{\langle \bm x,\bbm\theta\rangle}}
{1+e^{\langle \bm x,\bbm\theta\rangle}}
-y
\right)x_j
:\ \Vert\bbm\theta\Vert_0\le s
\right\}.
\]
Let
\[
\mathcal A_j
=
\left\{
\left\{
(\bm x,y,t)\in\mathbb R^p\times\{0,1\}\times\mathbb R:
f_{\bbm\theta,j}(\bm x,y)>t
\right\}
:\ \Vert\bbm\theta\Vert_0\le s
\right\}.
\]
Then, for every integer $m\ge 1$,
\[
s(\mathcal A_j,m)
\le
\exp\!\left\{
C s\log\!\left(\frac{ep}{s}\right)
+
C s\log\!\left(\frac{e(m\vee s)}{s}\right)
\right\},
\]
where $C>0$ is an absolute constant.

Consequently, $\mathcal F_j$ is a VC-subgraph class with VC index $V_j$ satisfying
\[
V_j \le C s\log\!\left(\frac{ep}{s}\right).
\]
Moreover, if $F_j(\bm x,y)=|x_j|$, then there exist absolute constants $A>1$ and
$C>0$ such that, for every probability measure $Q$ with $\|F_j\|_{Q,2}>0$ and every
$0<\varepsilon\le 1$,
\[
\log N\!\left(
\varepsilon \|F_j\|_{Q,2},
\mathcal F_j,
L_2(Q)
\right)
\le
C V_j \log\!\left(\frac{A}{\varepsilon}\right).
\]
\end{lemma}

\begin{proof}
Fix $j\in[p]$. For each support set $U\subset[p]$ with $|U|=d$, define
\[
\mathcal A_{U,j}
=
\left\{
\left\{
(\bm x,y,t)\in\mathbb R^p\times\{0,1\}\times\mathbb R:
\left(
\frac{e^{\langle \bm x,\bbm\theta\rangle}}
{1+e^{\langle \bm x,\bbm\theta\rangle}}
-y
\right)x_j>t
\right\}
:\ \operatorname{supp}(\bbm\theta)\subseteq U
\right\}.
\]
Then
\[
\mathcal A_j = \bigcup_{U\subset[p],\,|U|\le s}\mathcal A_{U,j}.
\]

Fix such a set $U$ and write $U=\{u_1,\dots,u_d\}$. For any
$\bbm\theta\in\mathbb R^p$ with $\operatorname{supp}(\bbm\theta)\subseteq U$, define
\[
\bbm\theta^{(U)}=(\theta_{u_1},\dots,\theta_{u_d})^\top\in\mathbb R^d,
\qquad
\bm x^{(U)}=(x_{u_1},\dots,x_{u_d})^\top\in\mathbb R^d.
\]
Then
\[
\langle \bm x,\bbm\theta\rangle
=
\langle \bm x^{(U)},\bbm\theta^{(U)}\rangle.
\]

Now fix
\[
f_{\bbm\theta,j}(\bm x,y)
=
\left(
\frac{e^{\langle \bm x,\bbm\theta\rangle}}
{1+e^{\langle \bm x,\bbm\theta\rangle}}
-y
\right)x_j.
\]
We analyze the event $\{f_{\bbm\theta,j}(\bm x,y)>t\}$.

If $x_j=0$, then $f_{\bbm\theta,j}(\bm x,y)=0$ for all $\bbm\theta$, and hence
$\{f_{\bbm\theta,j}(\bm x,y)>t\}$ is the parameter-free event $\{0>t\}$.

Now suppose $x_j\neq 0$. Since
\[
\frac{e^{\langle \bm x,\bbm\theta\rangle}}
{1+e^{\langle \bm x,\bbm\theta\rangle}}
\in (0,1),
\]
the only nontrivial cases are
\[
\begin{array}{ll}
y=0,\ x_j>0,\ 0<t<x_j, &
y=0,\ x_j<0,\ x_j<t<0, \\[0.3em]
y=1,\ x_j>0,\ -x_j<t<0, &
y=1,\ x_j<0,\ 0<t<-x_j.
\end{array}
\]
Outside these four cases, the event $\{f_{\bbm\theta,j}(\bm x,y)>t\}$ is parameter-free.

On the nontrivial region, define $\varepsilon(\bm x,y,t)\in\{\pm1\}$ and
$\psi(\bm x,y,t)\in\mathbb R$ by
\[
\varepsilon(\bm x,y,t)
=
\begin{cases}
\ \ 1, & y=0,\ x_j>0,\ 0<t<x_j,\\
-1, & y=0,\ x_j<0,\ x_j<t<0,\\
\ \ 1, & y=1,\ x_j>0,\ -x_j<t<0,\\
-1, & y=1,\ x_j<0,\ 0<t<-x_j,
\end{cases}
\]
and
\[
\psi(\bm x,y,t)
=
\begin{cases}
\log\!\left(\dfrac{t}{x_j-t}\right),
& y=0,\ x_j>0,\ 0<t<x_j,\\[1em]
\log\!\left(\dfrac{t}{x_j-t}\right),
& y=0,\ x_j<0,\ x_j<t<0,\\[1em]
\log\!\left(\dfrac{x_j+t}{-t}\right),
& y=1,\ x_j>0,\ -x_j<t<0,\\[1em]
\log\!\left(\dfrac{x_j+t}{-t}\right),
& y=1,\ x_j<0,\ 0<t<-x_j.
\end{cases}
\]
A direct case-by-case verification shows that, on the nontrivial region,
\[
f_{\bbm\theta,j}(\bm x,y)>t
\iff
\varepsilon(\bm x,y,t)
\Bigl(
\langle \bm x^{(U)},\bbm\theta^{(U)}\rangle
-
\psi(\bm x,y,t)
\Bigr)>0.
\]
Define
\[
\Phi_U(\bm x,y,t)
=
\Bigl(
\varepsilon(\bm x,y,t)\bm x^{(U)},
-
\varepsilon(\bm x,y,t)\psi(\bm x,y,t)
\Bigr)\in\mathbb R^{d+1}.
\]
Then, on the nontrivial region,
\[
f_{\bbm\theta,j}(\bm x,y)>t
\iff
(\bbm\theta^{(U)},1)^\top \Phi_U(\bm x,y,t)>0.
\]

Now fix arbitrary points
\[
(\bm x_1,y_1,t_1),\dots,(\bm x_m,y_m,t_m)
\in \mathbb R^p\times\{0,1\}\times\mathbb R.
\]
Let $I\subset[m]$ be the set of indices belonging to the nontrivial region, and write
$q=|I|\le m$. For $r\notin I$, the label
\[
\mathbf 1\!\left\{
f_{\bbm\theta,j}(\bm x_r,y_r)>t_r
\right\}
\]
is independent of $\bbm\theta$. Therefore, the number of dichotomies induced by
$\mathcal A_{U,j}$ on the $m$ points is at most the number of dichotomies induced on
the $q$ nontrivial points by affine halfspaces in $\mathbb R^{d+1}$.

Hence, by Corollary 13.1 of \citet{devroye2013probabilistic},
\[
s(\mathcal A_{U,j},m)
\le
2\sum_{\ell=0}^{d+1}\binom{q-1}{\ell}
\le
2\sum_{\ell=0}^{d+1}\binom{m-1}{\ell}
\le
\exp\!\left\{
C d\log\!\left(\frac{e(m\vee d)}{d}\right)
\right\}.
\]

We now take the union over supports. Since
\[
\mathcal A_j = \bigcup_{U\subset[p],\,|U|\le s}\mathcal A_{U,j},
\]
we have
\[
s(\mathcal A_j,m)
\le
1+\sum_{r=1}^s \binom{p}{r}
\exp\!\left\{
C r\log\!\left(\frac{e(m\vee r)}{r}\right)
\right\}.
\]
Using
\[
\sum_{r=1}^s \binom{p}{r}
\le
\left(\frac{ep}{s}\right)^s,
\]
and enlarging $C$ if necessary, it follows that
\[
s(\mathcal A_j,m)
\le
\exp\!\left\{
C s\log\!\left(\frac{ep}{s}\right)
+
C s\log\!\left(\frac{e(m\vee s)}{s}\right)
\right\}.
\]
This proves the shatter-coefficient bound and hence $\mathcal F_j$ is a VC-subgraph class.

Next, let $V_j$ denote the VC index of $\mathcal F_j$. Since
\[
2^{V_j}\le s(\mathcal A_j,V_j),
\]
the displayed shatter-coefficient bound implies, after enlarging the constant if necessary,
that
\[
V_j \le C s\log\!\left(\frac{ep}{s}\right).
\]

Finally, the covering-number bound follows from the standard entropy estimate for
VC-subgraph classes; see \citet[Section 2.6.2]{van1996weak}. More precisely, if
$F_j(\bm x,y)=|x_j|$ is the envelope of $\mathcal F_j$, then there exist absolute
constants $A>1$ and $C>0$ such that, for every probability measure $Q$ with
$\|F_j\|_{Q,2}>0$ and every $0<\varepsilon\le 1$,
\[
\log N\!\left(
\varepsilon \|F_j\|_{Q,2},
\mathcal F_j,
L_2(Q)
\right)
\le
C V_j \log\!\left(\frac{A}{\varepsilon}\right).
\]
This completes the proof.
\end{proof}

We now present the following lemma to show that Assumption \ref{assump:sparse_logistic_curvature} can hold for heavy-tailed designs.
\begin{lemma}[Student-$t$ designs satisfy sparsity-restricted logistic curvature]
\label{lem:student_t_sparse_logistic_curvature}
Let $\sigma(z)={e^z}/{(1+e^z)}$. Suppose that $\bm x_i$ follows a centered elliptical multivariate
Student-$t_\nu$ distribution with $\nu>2$, normalized so that $\mathbb E(\bm x_i\bm x_i^\top) = \boldsymbol\Sigma_{\bm x}$.
 Suppose that Assumption \ref{assump:linear_sparse_eigenvalues} holds with parameters $\kappa_-$ and $\kappa_+$. Then there exists a constant $c_{\nu,R,\kappa_+}>0$, depending only on $\nu$, $R$, and $\kappa_+$, such that, for every
$S\subset[p]$ with $|S|\le 2s+s^*$, every
${\bbm\theta}\in B(R)$ satisfying
$\operatorname{supp}({\bbm\theta})\subseteq S$, and every
$\bm v\in\mathbb R^p$ satisfying
$\operatorname{supp}(\bm v)\subseteq S$,
\[
\mathbb E\left[
\sigma'(\bm x_i^\top{\bbm\theta})
(\bm x_i^\top\bm v)^2
\right]
\ge
c_{\nu,R,\kappa_+}
\mathbb E\left[
(\bm x_i^\top\bm v)^2
\right].
\]
\end{lemma}

\begin{proof}
The claim is trivial when $\bm v=\boldsymbol 0$, so assume
$\bm v\ne \boldsymbol 0$. Fix
$S\subset[p]$ with $|S|\le 2s+s^*$,
$\operatorname{supp}(\bbm{\theta})\subseteq S$, and
$\operatorname{supp}(\bm v)\subseteq S$.

Let
\[
A^2
=
\bbm{\theta}^{\top}
\boldsymbol\Sigma_{\boldsymbol x}
\bbm{\theta},
\qquad
B^2
=
\bm v^\top
\boldsymbol\Sigma_{\boldsymbol x}
\bm v .
\]
By the sparse upper eigenvalue condition and
$\bbm{\theta}\in B(R)$,
\[
A^2
\le
\kappa_+\|\bbm{\theta}\|_2^2
\le
\kappa_+R^2.
\]
Moreover, by the sparse lower eigenvalue condition,
\[
B^2
\ge
\kappa_-\|\bm v\|_2^2
>0.
\]

Since $\bm x_i$ is elliptically Student-$t_\nu$, every two-dimensional
linear projection is again a centered bivariate Student-$t_\nu$ vector. This follows from the closure of elliptical Student-$t$ distributions under
linear transformations. Hence
there exist $\rho\in[-1,1]$ and a centered bivariate Student-$t_\nu$ vector
$(U_1,U_2)$ with covariance matrix $I_2$ such that
\[
\left(
\bm x_i^\top\bbm{\theta},
\bm x_i^\top\bm v
\right)
\stackrel{d}{=}
\left(
A U_1,\,
B\{\rho U_1+\sqrt{1-\rho^2}U_2\}
\right).
\]
When $A=0$, the same representation holds with any fixed choice of
$\rho$, for instance $\rho=0$.

Therefore,
\[\mathbb E\left[
\sigma'(\bm x_i^\top\bbm{\theta})
(\bm x_i^\top\bm v)^2
\right]=B^2\mathbb E\left[\sigma'(A U_1)\{\rho U_1+\sqrt{1-\rho^2}U_2\}^2
\right].\]
Since $A\le R\sqrt{\kappa_+}$ and
\[
\sigma'(z)=\frac{e^z}{(1+e^z)^2}
\]
is positive and symmetric in $z$, we have, on the event
$\{|U_1|\le 1\}$,
\[
\sigma'(A U_1)
\ge
\frac{e^{R\sqrt{\kappa_+}}}
{(1+e^{R\sqrt{\kappa_+}})^2}.
\]
Thus,
\[
\begin{aligned}
&
\mathbb E\left[
\sigma'(A U_1)
\{\rho U_1+\sqrt{1-\rho^2}U_2\}^2
\right]
\\
&\quad\ge
\frac{e^{R\sqrt{\kappa_+}}}
{(1+e^{R\sqrt{\kappa_+}})^2}
\mathbb E\left[
\{\rho U_1+\sqrt{1-\rho^2}U_2\}^2
\mathbf 1_{\{|U_1|\le 1\}}
\right].
\end{aligned}
\]
By the symmetry of the bivariate Student-$t_\nu$ distribution,
\[
\mathbb E\left[
U_1U_2\mathbf 1_{\{|U_1|\le 1\}}
\right]=0.
\]
Hence
\[
\begin{aligned}
&
\mathbb E\left[
\{\rho U_1+\sqrt{1-\rho^2}U_2\}^2
\mathbf 1_{\{|U_1|\le 1\}}
\right]
\\
&\quad =
\rho^2
\mathbb E\left[
U_1^2\mathbf 1_{\{|U_1|\le 1\}}
\right]
+
(1-\rho^2)
\mathbb E\left[
U_2^2\mathbf 1_{\{|U_1|\le 1\}}
\right]
\\
&\quad \ge
\min\left\{
\mathbb E\left[
U_1^2\mathbf 1_{\{|U_1|\le 1\}}
\right],
\mathbb E\left[
U_2^2\mathbf 1_{\{|U_1|\le 1\}}
\right]
\right\}.
\end{aligned}
\]
The two expectations in the minimum are strictly positive because the
bivariate Student-$t_\nu$ distribution has a positive density on
$\mathbb R^2$. Define
\[
c_{\nu,R,\kappa_+}
=
\frac{e^{R\sqrt{\kappa_+}}}
{(1+e^{R\sqrt{\kappa_+}})^2}
\min\left\{
\mathbb E\left[
U_1^2\mathbf 1_{\{|U_1|\le 1\}}
\right],
\mathbb E\left[
U_2^2\mathbf 1_{\{|U_1|\le 1\}}
\right]
\right\}.
\]
Then $c_{\nu,R,\kappa_+}>0$, and
\[
\mathbb E\left[
\sigma'(\bm x_i^\top\bbm{\theta})
(\bm x_i^\top\bm v)^2
\right]
\ge
c_{\nu,R,\kappa_+}
B^2.
\]
Since
\[
B^2
=
\bm v^\top
\boldsymbol\Sigma_{\boldsymbol x}
\bm v
=
\mathbb E[
(\bm x_i^\top\bm v)^2],
\]
we obtain
\[
\mathbb E\left[
\sigma'(\bm x_i^\top\bbm{\theta})
(\bm x_i^\top\bm v)^2
\right]
\ge
c_{\nu,R,\kappa_+}
\mathbb E[
(\bm x_i^\top\bm v)^2].
\]
This proves the lemma.
\end{proof}


\section{Proofs of Lower Bounds}
\label{sec:appendix_minimax}

This appendix proves the phase-wise lower bounds stated in Section~\ref{sec:lower_bound}. Appendix~\ref{append:local_refinement_lower} proves the fixed-design local-refinement lower bound. Appendix~\ref{append:proof_sparse_operator_lower} proves the sparse operator-norm lower bound for design-driven gradient stability. Appendix~\ref{append:logistic_lower_bound_result} gives the corresponding lower-bound statement and proof for sparse logistic regression.

\subsection{Fixed-Design Local-Refinement Lower Bound}
\label{append:local_refinement_lower}

We first show that the final local refinement rate cannot be improved under only a finite \((1+\delta)\)-moment condition on the noise. This lower bound is stated in a fixed-design framework, following the adaptive Huber regression literature. Its purpose is to isolate the noise-driven difficulty: even when the design is fixed and well behaved, heavy-tailed noise limits the final statistical precision.

Let \(\bm X\in\mathbb R^{n\times p}\) be a fixed design matrix with rows \(\bm x_i^\top\). Consider the model
\[
y_i=\bm x_i^\top\bbm\theta^*+\epsilon_i,\qquad i=1,\ldots,n,
\]
where the errors are independent and mean zero. We impose the following fixed-design condition on a sparse support.

\begin{assumption}[Fixed-design lower-bound condition]
\label{assump:fixed_design_lower}
There exists a support \(S_0\subset[p]\) with \(|S_0|=s^*\) and a sign vector \(\bm u\in\{-1,1\}^n\cap\operatorname{col}(\bm X_{S_0})\) such that
\[
\left\|\frac{1}{n}\bm X_{S_0}^\top \bm u\right\|_2\ge \alpha\sqrt{s^*},
\qquad
\lambda_{\max}\left(\frac{1}{n}\bm X_{S_0}^\top\bm X_{S_0}\right)\le \kappa_u,
\]
for constants \(\alpha>0\) and \(\kappa_u<\infty\).
\end{assumption}

\begin{theorem}[Fixed-design lower bound for local refinement]
\label{thm:fixed_design_local_refinement_lower}
Let \(0<\delta\le1\), \(s^*\ge1\), and \(p\ge2\). Suppose Assumption~\ref{assump:fixed_design_lower} holds. For any \(A_0>0\) satisfying \(A_0\log p\le c n\) for a sufficiently small constant \(c>0\), there exist constants \(c_1,c_2>0\), depending only on \(\delta\), such that
\begin{equation}
\inf_{\widehat{\bbm\theta}}
\sup_{P}
\mathbb{P}_{P}\left[
\|\widehat{\bbm\theta}-\bbm\theta^*\|_2
\ge
c_1\,\frac{\alpha}{\kappa_u}\,
M_{\epsilon,1+\delta}^{1/(1+\delta)}
\sqrt{s^*}
\left(\frac{A_0\log p}{n}\right)^{\delta/(1+\delta)}
\right]
\ge
c_2p^{-A_0},
\label{eq:fixed_design_local_refinement_lower}
\end{equation}
where the infimum is over all estimators and the supremum is over all fixed-design linear models with \(\|\bbm\theta^*\|_0\le s^*\), independent mean-zero errors, and \(\max_{1\le i\le n}\mathbb{E}[|\epsilon_i|^{1+\delta}]\le M_{\epsilon,1+\delta}\).
\end{theorem}

\begin{proof}[Proof of Theorem~\ref{thm:fixed_design_local_refinement_lower}]
The proof follows the fixed-design lower-bound construction used for adaptive Huber regression. The main idea is to construct two sparse regression models that are well separated in parameter space but generate the same response vector with non-negligible probability.

Let \(S_0\subset[p]\) and \(\bm u\in\{-1,1\}^n\cap\operatorname{col}(\bm X_{S_0})\) be as in Assumption~\ref{assump:fixed_design_lower}. Since \(\bm u\in\operatorname{col}(\bm X_{S_0})\), there exists \(\bm b\in\mathbb R^{s^*}\) such that \(\bm X_{S_0}\bm b=\bm u\). Moreover,
\[
\frac{1}{n}\bm X_{S_0}^\top \bm u=\frac{1}{n}\bm X_{S_0}^\top\bm X_{S_0}\bm b.
\]
By Assumption~\ref{assump:fixed_design_lower},
\[
\alpha\sqrt{s^*}
\le
\left\|\frac{1}{n}\bm X_{S_0}^\top \bm u\right\|_2
\le
\lambda_{\max}\left(\frac{1}{n}\bm X_{S_0}^\top\bm X_{S_0}\right)\|\bm b\|_2
\le
\kappa_u\|\bm b\|_2.
\]
Thus
\begin{equation}
\|\bm b\|_2\ge \frac{\alpha}{\kappa_u}\sqrt{s^*}.
\label{eq:lower_bound_b_norm}
\end{equation}

Let \(0<\gamma<1/2\) and \(c_0>0\) be chosen later. Define two parameter vectors \(\bbm\theta_+\) and \(\bbm\theta_-\), both supported on \(S_0\), by
\[
(\bbm\theta_+)_{S_0}=c_0\gamma \bm b,\qquad
(\bbm\theta_-)_{S_0}=-c_0\gamma \bm b,
\]
and set all coordinates outside \(S_0\) equal to zero. Then
\[
\bm X\bbm\theta_+=c_0\gamma \bm u,
\qquad
\bm X\bbm\theta_-=-c_0\gamma \bm u.
\]
Their separation satisfies
\begin{equation}
\frac{1}{2}\|\bbm\theta_+-\bbm\theta_-\|_2
=
\|\bbm\theta_+\|_2
=
c_0\gamma\|\bm b\|_2
\ge
c_0\gamma\frac{\alpha}{\kappa_u}\sqrt{s^*}.
\label{eq:fixed_design_param_sep}
\end{equation}

We next define two response distributions. Under the \(+\) model, independently over \(i=1,\ldots,n\),
\[
y_i=
\begin{cases}
0, & \text{with probability }1-\gamma,\\
c_0u_i, & \text{with probability }\gamma.
\end{cases}
\]
Under the \(-\) model,
\[
y_i=
\begin{cases}
0, & \text{with probability }1-\gamma,\\
-c_0u_i, & \text{with probability }\gamma.
\end{cases}
\]
Under the \(+\) model, since \(\bm x_i^\top\bbm\theta_+=c_0\gamma u_i\), the regression error \(\epsilon_i^+=y_i-\bm x_i^\top\bbm\theta_+\) takes values
\[
-c_0\gamma u_i
\quad\text{with probability }1-\gamma,
\qquad
c_0(1-\gamma)u_i
\quad\text{with probability }\gamma.
\]
Thus \(\mathbb{E}[\epsilon_i^+]=0\). Similarly, under the \(-\) model, \(\mathbb{E}[\epsilon_i^-]=0\). In both models,
\[
\mathbb{E}\left[|\epsilon_i^\pm|^{1+\delta}\right]
=
c_0^{1+\delta}
\left\{(1-\gamma)\gamma^{1+\delta}+\gamma(1-\gamma)^{1+\delta}\right\}
\le
2c_0^{1+\delta}\gamma,
\]
where we used \(0<\gamma<1/2\). Choose
\[
c_0=\left(\frac{M_{\epsilon,1+\delta}}{2\gamma}\right)^{1/(1+\delta)}.
\]
Then \(\mathbb{E}[|\epsilon_i^\pm|^{1+\delta}]\le M_{\epsilon,1+\delta}\) for every \(i\).

Let \(\mathcal E_0=\{y_1=\cdots=y_n=0\}\). Under both models,
\[
\mathbb{P}_+(\mathcal E_0)=\mathbb{P}_-(\mathcal E_0)=(1-\gamma)^n.
\]
On \(\mathcal E_0\), the observed data are identical under the two models. For any estimator \(\widehat{\bbm\theta}\), at least one of the two inequalities
\[
\|\widehat{\bbm\theta}-\bbm\theta_+\|_2
\ge
\frac{1}{2}\|\bbm\theta_+-\bbm\theta_-\|_2,
\qquad
\|\widehat{\bbm\theta}-\bbm\theta_-\|_2
\ge
\frac{1}{2}\|\bbm\theta_+-\bbm\theta_-\|_2
\]
must hold on this event. Therefore,
\[
\sup_{P\in\{P_+,P_-\}}
\mathbb{P}_{P}\left[
\|\widehat{\bbm\theta}-\bbm\theta^*\|_2
\ge
c_0\gamma\frac{\alpha}{\kappa_u}\sqrt{s^*}
\right]
\ge
(1-\gamma)^n.
\]

Set \(\gamma=A_0\log p/(2n)\). By the assumption \(A_0\log p\le cn\), with \(c\) sufficiently small, we have \(0<\gamma<1/2\). Moreover, \((1-\gamma)^n\ge \exp(-2\gamma n)=p^{-A_0}\). Finally,
\[
c_0\gamma
=
\left(\frac{M_{\epsilon,1+\delta}}{2\gamma}\right)^{1/(1+\delta)}\gamma
=
2^{-1/(1+\delta)}
M_{\epsilon,1+\delta}^{1/(1+\delta)}
\gamma^{\delta/(1+\delta)}
\asymp
M_{\epsilon,1+\delta}^{1/(1+\delta)}
\left(\frac{A_0\log p}{n}\right)^{\delta/(1+\delta)}.
\]
Combining this display with \eqref{eq:fixed_design_param_sep} proves \eqref{eq:fixed_design_local_refinement_lower}.
\end{proof}

\subsection{Proof of the Sparse Operator-Norm Lower Bound}
\label{append:proof_sparse_operator_lower}

We prove Theorem~\ref{thm:sparse_operator_gradient_lower}. The proof reduces uniform sparse gradient stability to sparse covariance operator estimation. The construction below is a sparse version of the dense heavy-tailed operator-norm construction: all non-common perturbations are restricted to a fixed support \(S\) with \(|S|=s^*\), so the ambient dimension in the lower-bound construction is \(s^*\).

\begin{lemma}[Sparse heavy-tailed covariance family]
\label{lem:sparse_covariance_operator_lower}
Let \(S\subset[p]\) with \(|S|=s^*\), and let \(0<\lambda\le1\). Suppose \(s^*\le cn\) for a sufficiently small constant \(c>0\). There exists a collection of distributions \(\{P_1,\ldots,P_N\}\) for a random vector \(\bm x\in\mathbb R^p\), with \(N\ge \exp(c_1s^*)\), such that the following properties hold. Each \(P_j\) has mean zero, covariance matrix \(\bm\Sigma_j\), sparse eigenvalues bounded above and below on all sparse directions of the relevant size, and
\[
\sup_{\|\bm v\|_2=1,\,\|\bm v\|_0\le s^*}
\mathbb{E}_{P_j}\left[|\bm x^\top\bm v|^{2+2\lambda}\right]
\le C M_x.
\]
Moreover, for all \(j\ne k\),
\[
\|\bm\Sigma_j-\bm\Sigma_k\|_\textup{op}
\ge
c_2 M_x^{1/(1+\lambda)}
\left(\frac{s^*}{n}\right)^{\lambda/(1+\lambda)},
\]
and the average KL divergence satisfies
\[
\frac{1}{N^2}\sum_{j,k=1}^{N}D_{\mathrm{KL}}(P_j^{\otimes n},P_k^{\otimes n})
\le
c_3\log N,
\]
where \(c_3>0\) is sufficiently small.
\end{lemma}

\begin{proof}
Let \(d=s^*\) and identify the support \(S\) with \(\{1,\ldots,d\}\). Let \(\mathcal V=\{\bm v_1,\ldots,\bm v_N\}\subset\{d^{-1/2},-d^{-1/2}\}^d\) be a packing of sign vectors with \(N\ge\exp(c d)\) and \(|\langle\bm v_j,\bm v_k\rangle|\le 1/2\) for \(j\ne k\). Such a packing follows from the Varshamov--Gilbert lemma.

We first construct a rare-spike distribution on the coordinates in \(S\). Let \(\mathsf U\) denote the uniform distribution on \(\{-1,1\}^d\). Fix a sufficiently small constant \(\tau>0\). For each \(\bm v_j\), define a tilted sign distribution \(\mathsf R_j\) by
\[
\frac{d\mathsf R_j}{d\mathsf U}(\bm b)
=
\frac{\exp\{\tau(\bm b^\top\bm v_j)^2\}}{Z(\tau,\bm v_j)},
\qquad \bm b\in\{-1,1\}^d,
\]
where \(Z(\tau,\bm v_j)\) is the normalizing constant. Since \(\bm v_j\in\{d^{-1/2},-d^{-1/2}\}^d\), the law of \(\bm b^\top\bm v_j\) under \(\mathsf U\) does not depend on \(j\). Therefore \(Z(\tau,\bm v_j)\) is the same for all \(j\). For fixed small \(\tau\), standard moment bounds for Rademacher sums imply that
\[
D_{\mathrm{KL}}(\mathsf R_j,\mathsf R_k)\le C_\tau
\]
for all \(j,k\), and there exists \(c_\tau>0\) such that
\[
\left\|\mathbb{E}_{\mathsf R_j}[\bm b\bm b^\top]
-
\mathbb{E}_{\mathsf R_k}[\bm b\bm b^\top]\right\|_\textup{op}
\ge c_\tau
\]
whenever \(|\langle\bm v_j,\bm v_k\rangle|\le 1/2\). This follows by expanding the tilted second moment along the directions \(\bm v_j\) and \(\bm v_k\); the tilt increases the second moment in the direction \(\bm v_j\) by a constant amount while the packing condition keeps the corresponding rank-one directions separated. The constants depend only on \(\tau\).

Let \(\gamma=c_0 d/n\) with \(c_0>0\) sufficiently small, and let
\[
R=\left(\frac{M_x}{\gamma}\right)^{1/(2+2\lambda)}.
\]
Define a sparse rare-spike vector \(\bm y_j\in\mathbb R^p\) as follows. With probability \(1-\gamma\), set \(\bm y_j=0\). With probability \(\gamma\), draw \(\bm b\sim\mathsf R_j\) and set \((\bm y_j)_S=R\bm b\) and \((\bm y_j)_{S^c}=0\). Then \(\mathbb{E}[\bm y_j]=0\), and
\[
\operatorname{Cov}(\bm y_j)
=
\gamma R^2\,\mathbb{E}_{\mathsf R_j}[\bm b\bm b^\top]
\quad\text{on }S.
\]
For \(j\ne k\),
\[
\|\operatorname{Cov}(\bm y_j)-\operatorname{Cov}(\bm y_k)\|_\textup{op}
\ge
c_\tau\gamma R^2
=
c_\tau M_x^{1/(1+\lambda)}
\gamma^{\lambda/(1+\lambda)}.
\]
Since \(\gamma\asymp d/n\), this lower bound has the desired order.

We next control the KL divergence. The distributions of \(\bm y_j\) and \(\bm y_k\) share the atom at zero with mass \(1-\gamma\). Conditional on the rare-spike event, their sign distributions are \(\mathsf R_j\) and \(\mathsf R_k\). Thus
\[
D_{\mathrm{KL}}(\mathcal L(\bm y_j),\mathcal L(\bm y_k))
=
\gamma D_{\mathrm{KL}}(\mathsf R_j,\mathsf R_k)
\le
C_\tau\gamma.
\]
Consequently,
\[
D_{\mathrm{KL}}(\mathcal L(\bm y_j)^{\otimes n},\mathcal L(\bm y_k)^{\otimes n})
\le
C_\tau n\gamma
\le
C d.
\]
Choosing \(c_0\) sufficiently small makes the average KL divergence bounded by a small constant times \(\log N\), since \(\log N\asymp d\).

To ensure sparse eigenvalue lower bounds uniformly, add an independent common component. Let \(\bm z\sim N(0,I_p)\), independent of \(\bm y_j\), and define \(\bm x_j=\bm z+\bm y_j\). Let \(P_j\) denote the distribution of \(\bm x_j\). Since adding the same independent Gaussian component is a Markov kernel, the data-processing inequality gives
\[
D_{\mathrm{KL}}(P_j^{\otimes n},P_k^{\otimes n})
\le
D_{\mathrm{KL}}(\mathcal L(\bm y_j)^{\otimes n},\mathcal L(\bm y_k)^{\otimes n}).
\]
Moreover,
\[
\bm\Sigma_j=\operatorname{Cov}(\bm x_j)=I_p+\operatorname{Cov}(\bm y_j),
\]
so all sparse eigenvalues are bounded below by \(1\), and bounded above by a constant depending only on \(M_x\) and the choice of \(c_0\). The covariance separations are unchanged because the Gaussian component is common:
\[
\bm\Sigma_j-\bm\Sigma_k=\operatorname{Cov}(\bm y_j)-\operatorname{Cov}(\bm y_k).
\]

It remains to verify the moment condition. For any \(s^*\)-sparse unit vector \(\bm v\),
\[
|\bm x_j^\top\bm v|^{2+2\lambda}
\le
C_\lambda\left(|\bm z^\top\bm v|^{2+2\lambda}+|\bm y_j^\top\bm v|^{2+2\lambda}\right).
\]
The Gaussian term is bounded by a constant depending only on \(\lambda\). For the rare-spike term, Khintchine's inequality gives \(\mathbb{E}_{\mathsf R_j}[|\bm b^\top\bm v_S|^{2+2\lambda}]\le C_\lambda\) uniformly in \(\bm v\). Hence
\[
\mathbb{E}[|\bm y_j^\top\bm v|^{2+2\lambda}]
=
\gamma R^{2+2\lambda}
\mathbb{E}_{\mathsf R_j}[|\bm b^\top\bm v_S|^{2+2\lambda}]
\le
C_\lambda \gamma R^{2+2\lambda}
=
C_\lambda M_x.
\]
Assuming \(M_x\ge1\), the total moment is bounded by \(CM_x\). This proves all asserted properties.
\end{proof}

\begin{proof}[Proof of Theorem~\ref{thm:sparse_operator_gradient_lower}]
We work in the noiseless linear model with \(\epsilon=0\) and \(\bbm\theta^*=0\). Under a design distribution \(P_j\), the population score for a direction \(\Delta\) is
\[
G_j(\Delta)=\bm\Sigma_j\Delta.
\]
Thus estimating the gradient map \(\Delta\mapsto G_j(\Delta)\) uniformly over unit vectors supported on \(S\) is at least as hard as estimating the covariance operator \(\bm\Sigma_j\) on the subspace indexed by \(S\).

Let \(\widehat G\) be any estimator of the gradient map based on \(n\) observations. Define
\[
L_j(\widehat G)
=
\sup_{\Delta:\ \|\Delta\|_2=1,\ \operatorname{supp}(\Delta)\subseteq S}
\|\widehat G(\Delta)-\bm\Sigma_j\Delta\|_2.
\]
For \(j\ne k\),
\[
\sup_{\|\Delta\|_2=1,\operatorname{supp}(\Delta)\subseteq S}
\|(\bm\Sigma_j-\bm\Sigma_k)\Delta\|_2
=
\|(\bm\Sigma_j-\bm\Sigma_k)_{S,S}\|_\textup{op}.
\]
By Lemma~\ref{lem:sparse_covariance_operator_lower}, these pairwise distances are at least \(2r_n\), where
\[
r_n
=
c M_x^{1/(1+\lambda)}
\left(\frac{s^*}{n}\right)^{\lambda/(1+\lambda)}.
\]
If \(L_j(\widehat G)<r_n\), then \(\widehat G\) is closer to \(G_j\) than to any other \(G_k\) in the packing. Hence any estimator with loss below \(r_n\) can be converted into a test for the index \(j\). Fano's inequality and the KL bound in Lemma~\ref{lem:sparse_covariance_operator_lower} imply that
\[
\inf_{\widehat G}
\sup_{j=1,\ldots,N}
\mathbb{P}_{P_j}\left[L_j(\widehat G)\ge r_n\right]
\ge c_0
\]
for a universal constant \(c_0>0\). Consequently,
\[
\inf_{\widehat G}
\sup_{j=1,\ldots,N}
\mathbb{E}_{P_j}[L_j(\widehat G)]
\ge
c_0 r_n.
\]
Since the distributions \(P_j\) belong to the class \(\mathcal Q_\lambda(S)\), the desired lower bound follows.
\end{proof}

\subsection{Logistic Lower-Bound Result and Proof}
\label{append:logistic_lower_bound_result}

For completeness, we also record a standard sparse logistic lower bound. This result is not central to the two-regime decoupling phenomenon, but it shows that the square-root local refinement rate for logistic regression cannot be improved under the usual sparse high-dimensional scaling.

\begin{theorem}[Lower bound for sparse logistic regression]
\label{thm:logistic_lower_bound}
Consider the fixed-design logistic model
\[
\mathbb{P}(y_i=1)=\sigma(\bm x_i^\top\bbm\theta^*),\qquad i=1,\ldots,n,
\]
with \(\|\bbm\theta^*\|_0\le s^*\). Suppose the design satisfies a sparse fixed-design upper bound
\[
\frac{1}{n}\|\bm X(\bbm\theta-\bbm\theta')\|_2^2
\le
\kappa_u\|\bbm\theta-\bbm\theta'\|_2^2
\]
for all \(s^*\)-sparse \(\bbm\theta,\bbm\theta'\) in a sufficiently small neighborhood of the origin. Then there exist constants \(c_1,c_2>0\), depending only on \(\kappa_u\), such that
\[
\inf_{\widehat{\bbm\theta}}
\sup_{\|\bbm\theta^*\|_0\le s^*}
\mathbb{P}_{\bbm\theta^*}\left[
\|\widehat{\bbm\theta}-\bbm\theta^*\|_2
\ge
c_1\sqrt{\frac{s^*\log(ep/s^*)}{n}}
\right]
\ge c_2.
\]
\end{theorem}

\begin{proof}[Proof of Theorem~\ref{thm:logistic_lower_bound}]
The proof is a standard Fano argument. Let \(\mathcal V\subset\{0,1\}^p\)
be a packing of \(s^*\)-sparse vectors obtained from the
Varshamov--Gilbert bound, so that
\[
|\mathcal V|\ge \exp\{c s^*\log(ep/s^*)\},
\qquad
\|v-u\|_0\ge c s^* \quad (u\ne v).
\]
For a small amplitude \(a>0\), define \(\bbm\theta_v=a v\). Then the pairwise
separation satisfies
\[
\|\bbm\theta_v-\bbm\theta_u\|_2\ge c\sqrt{s^*}\,a.
\]

For the logistic model, the KL divergence between two parameter values is bounded by a constant multiple of the squared prediction distance:
\[
D_{\mathrm{KL}}(P_{\bbm\theta_v},P_{\bbm\theta_u})
\le
C\|\bm X(\bbm\theta_v-\bbm\theta_u)\|_2^2
\le
Cn\kappa_u\|\bbm\theta_v-\bbm\theta_u\|_2^2
\le
Cn\kappa_u s^* a^2.
\]
Choose \(a=c\sqrt{\log(ep/s^*)/n}\) with \(c>0\) sufficiently small. Then the pairwise KL divergence is bounded by a small constant times \(\log|\mathcal V|\). Fano's inequality implies that any estimator makes a testing error with probability bounded away from zero, and therefore
\[
\inf_{\widehat{\bbm\theta}}
\sup_{v\in\mathcal V}
\mathbb{P}_{\bbm\theta_v}\left[
\|\widehat{\bbm\theta}-\bbm\theta_v\|_2
\ge
c\sqrt{s^*}\,a
\right]\ge c',
\]
which gives the claimed result.
\end{proof}


\section{Implementation Details and Additional Simulation Results}
\label{append:G}

\subsection{Tuning Protocol for the Linear Simulation Comparison}
\label{app:tuning-linear-comparison}
This section describes the tuning protocol used in the linear simulation comparing TS-RIGHT with IHT, Lasso, adaptive Huber regression, and a shrinkage/truncation
baseline. 
\paragraph{Simulation configuration.}
The comparison is conducted under sparse linear regression with heavy-tailed
covariates and heavy-tailed noise.  The ambient dimension is fixed at $p=800$,
and the sample size varies over $n\in\{800,1200,1600,2000,2400\}$.
The true coefficient vector \(\theta^*\in\mathbb R^p\) has sparsity $s^*=5$, and the algorithmic sparsity level is set to $s=10$. The nonzero entries of \(\theta^*\) have magnitude \(5\).  The covariates are
generated from a heavy-tailed \(t\)-design with degrees of freedom \(\nu_{\bm x}=2.5\),
and scale \(1\), and the noise is generated from a \(t\)-distribution with
degrees of freedom $\nu_\epsilon=1.5$. The number of Monte Carlo repetitions for the final evaluation is \(100\).

\paragraph{Validation metric.}
All validation-based tuning steps use the same robust prediction criterion.
For an estimator \(\widehat{\bbm\theta}\) and a validation set \(V\), the validation
loss is
\[
\mathcal L_{\mathrm{val}}(\widehat{\bbm\theta};V)
=
\operatorname{median}_{i\in V}
\left|y_i-\bm x_i^\top\widehat{\bbm\theta}\right|.
\]
This criterion is used instead of the ordinary validation mean squared error
because the response is heavy-tailed.  It is therefore less sensitive to a small
number of extreme validation observations.

\paragraph{Pilot tuning and freezing rule.}
For TS-RIGHT and IHT, tuning is performed separately for each sample size \(n\).
For each candidate hyperparameter value \(h\) in a prespecified grid, we generate
\(B_{\mathrm{pilot}}=15\) independent pilot data sets.  Each pilot data set is
split into a training subset and a validation subset, with validation fraction
\(0.2\).  The method is fitted on the pilot training subset and evaluated on the
pilot validation subset using \(\mathcal L_{\mathrm{val}}\).

Let
\[
M(h)
=
\operatorname{median}_{b=1,\ldots,B_{\mathrm{pilot}}}
\mathcal L_{\mathrm{val}}^{(b)}(h)
\]
denote the median validation loss of candidate \(h\) over the pilot repetitions.
Let $M_{\min}=\min_{h\in\mathcal H}M(h)$ be the best pilot median validation loss over the candidate grid \(\mathcal H\).
We first form the near-best set
\[
\mathcal A
=
\left\{
h\in\mathcal H:
M(h)\le 1.05\,M_{\min}
\right\}.
\]
The selected candidate is then chosen from \(\mathcal A\) using the method-specific
simplicity rule described below.  The selected tuning parameters are frozen before
the final Monte Carlo evaluation.

\paragraph{TS-RIGHT tuning.}
For TS-RIGHT, the split fraction and stage lengths are selected over the grid $q\in\{0.50\}$, $T_1\in\{150,200\}$, $T_2\in\{8\}$.
Here \(q=n_2/n\) is the Stage II refinement fraction.  Given \(q,T_1,T_2\), the
sample split is $n_2=\lfloor qn\rfloor$ and $n_1=n-n_2$.
The working sparsity proxy used in the block-number rule is set to
\(s_{\mathrm{ref}}=s^*\) in this simulation.  The block numbers are
\[
K_1
=
\min\left\{
\left\lceil s_{\mathrm{ref}}\log p\right\rceil,
\left\lfloor\frac{n_1}{m_{\min}}\right\rfloor
\right\},
\]
and
\[
b_2=\left\lfloor\frac{n_2}{T_2}\right\rfloor,
\qquad
K_2
=
\min\left\{
\left\lceil\log p\right\rceil,
\left\lfloor\frac{b_2}{m_{\min}}\right\rfloor
\right\},
\]
with $m_{\min}=10$. The step size for TS-RIGHT is fixed at $\eta_{\mathrm{TS}}=0.02$.

\paragraph{IHT tuning.}
The IHT baseline uses the same algorithmic sparsity level \(s=10\) as TS-RIGHT.
The step size and number of iterations are selected from the grid
\[
\eta\in\{0.005,0.01,0.02\},
\qquad
T\in\{200,400,800\}.
\]
The same pilot validation procedure and the same validation loss
\(\mathcal L_{\mathrm{val}}\) are used.  Among IHT candidates in the near-best
set \(\mathcal A\), we choose the candidate with the smallest number of
iterations \(T\).  If there is still a tie, we choose the step size closest to
\(\eta=0.01\) on the logarithmic scale.  The selected \((\eta,T)\) pair is
frozen before the final Monte Carlo evaluation for the corresponding \(n\).

\paragraph{Lasso tuning.}
The Lasso baseline is tuned by \(K\)-fold cross-validation using
\texttt{cv.glmnet}, with \(K=5\) folds.  The design matrix is not internally
standardized and no intercept is included, matching the centered simulation
design.  The selected regularization parameter is the cross-validation minimizer,
\[
\widehat\lambda_{\mathrm{Lasso}}
=
\lambda_{\min}.
\]
The Lasso estimator is then refit on the full data set for that Monte Carlo
replication using \(\widehat\lambda_{\mathrm{Lasso}}\).

\paragraph{Adaptive Huber tuning.}
The adaptive Huber baseline uses a Lasso initializer to construct residual-based
Huber thresholds.  For each Monte Carlo replication, the data are split into a
training subset and a validation subset with validation fraction \(0.2\).  A
Lasso fit on the training subset gives initial residuals $\widehat r_i=y_i-\bm x_i^\top\widehat{\bbm\theta}_{\mathrm{init}}$.
For $\tau_q\in\{0.90,0.95,0.98\}$,
the Huber threshold is set to the empirical \(\tau_q\)-quantile of
\(|\widehat r_i|\) on the training subset.  The regularization parameter is
searched over an 8-point logarithmic grid from \(10^{-3}\) to \(1\).

For each candidate \((\tau_q,\lambda)\), the adaptive Huber estimator is fitted
on the training subset and evaluated on the validation subset using
\(\mathcal L_{\mathrm{val}}\).  The near-best rule is again applied with
multiplicative tolerance \(1.05\).  Among near-best candidates, we select the
largest \(\lambda\), favoring the simpler regularized model, and then choose the
\(\tau_q\) closest to \(0.95\) if needed.  The selected \(\lambda\) and residual
quantile \(\tau_q\) are then used to refit the adaptive Huber estimator on the
full data set for that Monte Carlo replication.  The Huber threshold used in
the final fit is recomputed from the full-data residuals at the selected
quantile \(\tau_q\).

\paragraph{Shrinkage/truncation tuning.}
The shrinkage baseline uses coordinatewise truncation of covariates and responses
followed by an \(\ell_1\)-regularized fit.  For each Monte Carlo replication, the
data are split into training and validation subsets with validation fraction
\(0.2\).  For each $\tau_q\in\{0.90,0.95,0.98\}$, the covariate truncation threshold \(\tau_X\) and response truncation threshold
\(\tau_Y\) are set to the corresponding empirical quantiles of \(|X_{ij}|\) and
\(|y_i|\) on the training subset.  After truncation, the regularization
parameter \(\lambda\) is selected by 5-fold cross-validation on the truncated
training data.

Each candidate \(\tau_q\) is evaluated on the validation subset using
\(\mathcal L_{\mathrm{val}}\).  The near-best rule with tolerance \(1.05\) is
applied, and the selected truncation quantile is the near-best candidate closest
to \(0.95\), with validation loss used as the final tie-breaker.  The selected
quantile is then used to recompute \(\tau_X\) and \(\tau_Y\) on the full data
set for that Monte Carlo replication, and the shrinkage estimator is refit using
the selected regularization parameter.

\paragraph{Final evaluation.}
After tuning, the final Monte Carlo evaluation is conducted on independent data
sets generated with seeds disjoint from the pilot tuning seeds.  For each sample
size \(n\), TS-RIGHT and IHT use the frozen hyperparameters selected from their
pilot tuning registries.  Lasso, adaptive Huber, and shrinkage use their
within-replication cross-validation or validation procedures described above.
The reported primary accuracy metric is the coefficient estimation error $\|\widehat{\bbm\theta}-\bbm\theta^*\|_2$. Across the \(100\) final repetitions, we report the median \(\ell_2\) error with the 0.25 and 0.75 quantiles. The median \(\ell_2\) error is used as the main summary statistic
because the heavy-tailed design and noise can produce occasional extreme errors.

\subsection{Baseline Comparisons for Logistic Regression}
\label{append:logistic_comparison}
We benchmark TS-RIGHT against leading robust and sparse regression methods: Lasso \citep{tibshirani1996regression}, the shrinkage method \citep{fan2021shrinkage}, and standard IHT \citep{jain2014iterative} applied to logistic regression.

We generate data with $p=600$, sparsity $s^*=5$, and varying $n$. The true coefficient vector is $\theta^*=(1,-1,1,-1,1,0,\dots,0)$. To create a challenging heavy-tailed design, the rows of the feature matrix $\mathbf{X} \in \mathbb{R}^{n \times p}$ are generated from a multivariate t-distribution with $\nu_x = 2.05$ degrees of freedom and scaled by a factor of $4$. The binary response $y_i$ is drawn from a Bernoulli distribution with success probability $p_i = (1 + \exp(-\bm x_i^\top \bbm\theta^*))^{-1}$. Each simulation setting is repeated for $100$ independent trials, and we report the 0.25, 0.5, and 0.75 quantiles of the $\ell_2$ estimation error, $\|\widehat{\bbm\theta} - \bbm\theta^*\|_2$.

Figure \ref{fig:logistic_comparison} summarizes the results. The Shrinkage method performs only slightly better than Lasso, since its robustness is limited when features lack finite fourth moments. Unlike in linear regression, IHT does not diverge in logistic regression due to the boundedness of the logistic loss gradient; nevertheless, it still uniformly underperforms TS-RIGHT across all sample sizes. In contrast to linear regression, TS-RIGHT exhibits a wider empirical interquartile range than IHT. This behavior is expected under the heavy-tailed covariate design. Unlike the squared loss, whose population Hessian is a constant $\bbm\Sigma_{\bm x}$, the logistic Hessian is weighted by $\sigma(\bm x_i^\top \bbm\theta)(1-\sigma(\bm x_i^\top \bbm\theta))$, which becomes vanishingly small in saturated regions. Consequently, the effective curvature is only determined by a relatively small and random subset of near-boundary observations. Thus, the MoM estimator in TS-RIGHT effectively reduces the bias while introducing more variance. The baseline methods employ full-sample objectives or gradients and appear more stable in this experiment, but their errors are substantially larger; their narrow intervals primarily reflect stable bias rather than improved accuracy.
 \begin{figure}[htbp]
  \centering
     \includegraphics[width=0.85\textwidth]{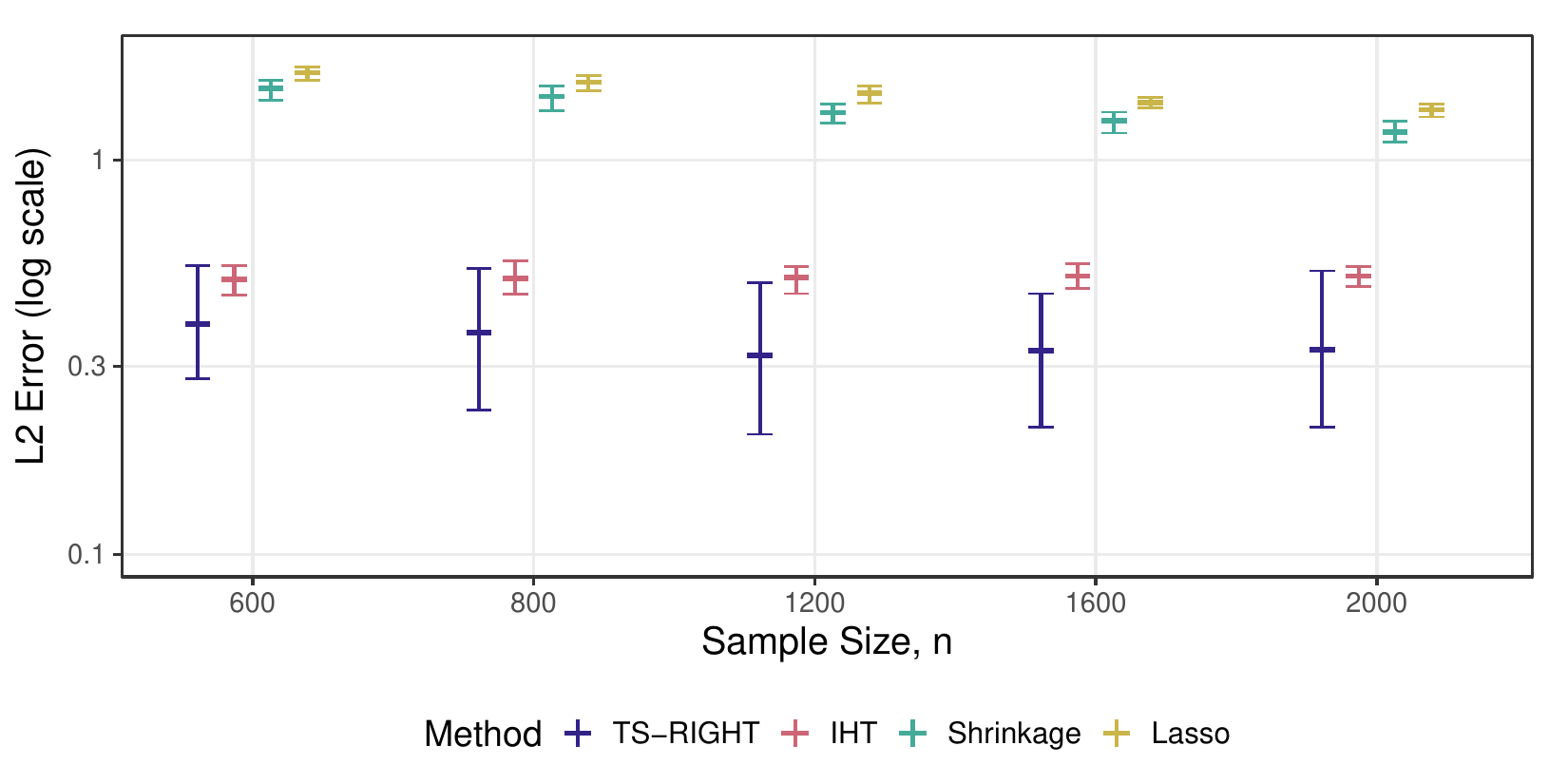}
    \caption{Logistic regression performance comparison.}
     \label{fig:logistic_comparison}
 \end{figure}

\subsection{Advantage of the Delayed-Splitting Strategy}
\label{append:delayed_splitting_advantage}
For simplicity, we refer to three variants of the robust-gradient algorithm introduced in Section \ref{subsec:path_dependence_delayed_splitting} as follows: our TS-RIGHT algorithm, the no-splitting version (NS-RIGHT), and the equal-splitting version (ES-RIGHT).

We demonstrate the advantage of the TS-RIGHT algorithm over NS-RIGHT and ES-RIGHT in linear regression through two simulations. First, we compare the two-stage errors of TS-RIGHT against NS-RIGHT. Second, we evaluate the final estimation errors of all three methods as the true sparsity $s^*$ increases. NS-RIGHT and TS-RIGHT are said to have a matched computational budget if they have the same epoch number, i.e., $T_{NS}\times n=T_1\times n_1+T_2\times n_2$.

In the first simulation, we set $n=p=800$ and $s=2s^*=10$. We generate data with heavy-tailed design ($df_X=2.5$, $\Sigma_X = I_p$) and heavy-tailed noise ($df_\epsilon=1.5$). In the second simulation, we set $n=p=1800$ and vary $s^*$ from $5$ to $30$, using the same heavy-tailed settings. All methods run for 100 iterations, and we report the median of final estimation errors.

\begin{figure}[bhtp]
  \centering
     \includegraphics[width=0.85\textwidth]{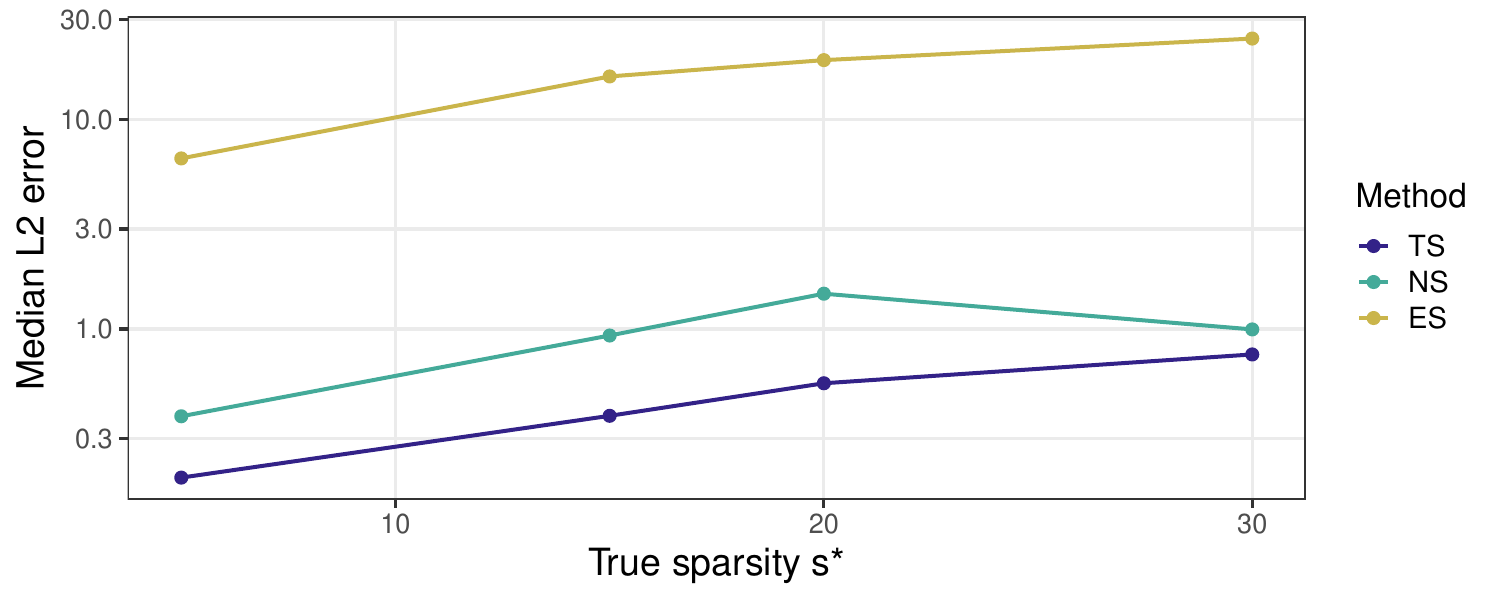}
    \caption{Comparison of TS-RIGHT, NS-RIGHT, and ES-RIGHT.}
     \label{fig:fs_ts_fd_sweep}
 \end{figure}

\begin{table}
\center
\caption{Median $l_2$ estimation errors of TS and NS under a matched computational budget.}
\label{tab:ts_fd_comparison}
\small
\setlength{\tabcolsep}{3pt}
\renewcommand{\arraystretch}{1.2}
\begin{tabular}{lccc}
\toprule
Method & Stage 1 endpoint &  Stage 2 endpoint & NS (Matched Budget) \\
\midrule
TS ($q=0.50$, $(T_1,T_2)=(150,8)$) & $0.428$ & \shortstack[c]{$\mathbf{0.330}$\\{\scriptsize $\downarrow 24.1\%$ from Stage 1}} & $0.752$ \\
TS ($q=0.50$, $(T_1,T_2)=(175,8)$) & $0.390$ & \shortstack[c]{$\mathbf{0.315}$\\{\scriptsize $\downarrow 23.7\%$ from Stage 1}} & $0.522$ \\
TS ($q=0.50$, $(T_1,T_2)=(200,8)$) & $0.403$ & \shortstack[c]{$\mathbf{0.305}$\\{\scriptsize $\downarrow 24.6\%$ from Stage 1}} & $0.410$ \\
\bottomrule
\end{tabular}
\end{table}
Table \ref{tab:ts_fd_comparison} demonstrates that for TS-RIGHT, Stage II errors are substantially lower than Stage I errors, indicating that the second stage provides meaningful refinement despite the limited number of iterations. Moreover, TS-RIGHT consistently outperforms NS-RIGHT across all settings, confirming the sharper statistical rate of TS-RIGHT. 
In the second simulation (Figure \ref{fig:fs_ts_fd_sweep}), ES-RIGHT exhibits the largest errors due to poor sample efficiency, while NS-RIGHT's errors remain uniformly higher than TS-RIGHT. These results underscore the advantage of our two-stage strategy for balancing computational efficiency and statistical accuracy in the presence of heavy-tailed data.


{\setstretch{0.9}

\bibliography{bibliography.bib}

}

\end{document}